\RequirePackage{fix-cm}
\documentclass[a4paper]{article}       
\usepackage{graphicx}
\usepackage{tikz}
\usepackage{rotating}
\usetikzlibrary{calc}
\usetikzlibrary{decorations.markings,decorations.pathmorphing,decorations.pathreplacing}
\usetikzlibrary{arrows,shapes}
\usepackage{enumitem}
\usepackage{url}
\usepackage{fullpage}
\usepackage{amssymb,amsmath,amsthm}
\usepackage{booktabs}
\usepackage{caption}
\usepackage[sort&compress,numbers]{natbib}
\usepackage{hyperref}
\newtheorem{lemma}{Lemma}
\newtheorem{theorem}{Theorem}
\newtheorem{corollary}{Corollary}
\newtheorem{proposition}{Proposition}
\newtheorem{claimu}{Claim}
\newtheorem{observation}{Observation}
\DeclareMathOperator{\suc}{succ}
\DeclareMathOperator{\pre}{pred}
\begin{document}

\title{Stable Matchings with Covering Constraints:\\ A Complete Computational Trichotomy\thanks{A preliminary version of the results presented in here appeared at the Proceedings of the 10th International Symposium on Algorithmic Game Theory~\cite{MnichSchlotter2017}.
        Research of M. Mnich supported by ERC Starting Grant 306465 (BeyondWorstCase) and DFG Grant MN 59/4-1.
        Research of I. Schlotter supported by the Hungarian National Research Fund (OTKA grants no. K 128611 and no. K 124171).}
}

\author{Matthias Mnich\thanks{Universit{\"a}t Bonn, Bonn, Germany. \texttt{mmnich@uni-bonn.de}}
          \and
        Ildik{\'o} Schlotter\thanks{Budapest University of Technology and Economics, Budapest, Hungary. \texttt{ildi@cs.bme.hu}}
}

\date{October 5, 2019}

\maketitle

\begin{abstract}
  Stable matching problems with lower quotas are fundamental in academic hiring and ensuring operability of rural hospitals.
  Only few tractable (polynomial-time solvable) cases of stable matching with lower quotas have been identified; most such problems are $\mathsf{NP}$-hard and also hard to approximate (Hamada et al., Algorithmica 74(1), 2016).

  We thus consider stable matching problems with lower quotas under a relaxed notion of tractability, namely fixed-parameter tractability.
  By cloning hospitals we focus on the case when all hospitals have upper quota equal to~1, which generalizes the setting of ``arranged marriages'' first considered by Donald Knuth in 1976.
  We investigate how a set of natural parameters, namely the maximum length of preference lists for men and women, the number of distinguished men and women, and the number of blocking pairs allowed determine the computational tractability of this problem.

  Our main result is a complete complexity trichotomy: for each choice of parameters we either provide a polynomial-time algorithm, or an $\mathsf{NP}$-hardness proof and fixed-parameter algorithm, or $\mathsf{NP}$-hardness proof and $\mathsf{W}[1]$-hardness proof.
  As corollary, we negatively answer a question by Hamada et al. (Algorithmica 74(1), 2016) by showing fixed-parameter intractability parameterized by optimal solution size.
  We also classify all cases of one-sided constraints where only women may be distinguished.

\end{abstract}

\section{Introduction}
\label{sec:introduction}
The {\sc Stable Marriage (SM)} problem is a fundamental problem first studied by Gale and Shapley~\cite{GaleShapley1962} in 1962.
An instance of SM consists of a set $\mathcal{M}$ of men, a set $\mathcal{W}$ of women, and a preference list for each person ordering members of the opposite sex.
We aim to find a stable matching, i.e., a matching for which there exists no pair of a man and a woman who prefer each other to their partners given by the matching; such a pair is called a blocking pair.
Gale and Shapley proved~\cite{GaleShapley1962} that any instance of SM admits at least one stable matching, and gave a polynomial-time algorithm, known as the Gale-Shapley algorithm, to find one.
Gale and Shapley also considered the many-to-one extension of SM, known as the {\sc Hospitals/Residents} (HR) problem.
In HR, the two sets $\mathcal H$ and $\mathcal R$ that correspond to men and women in the SM problem are residents and hospitals, respectively.
Each hospital $H\in\mathcal H$ has an upper quota on the number of residents in $\mathcal R$ that it can accept.
For HR it still holds true that a stable matching always exists, and can be found efficiently.

An extension of HR that is motivated by several real-world applications is the \textsc{Hospitals/Residents with Lower Quota} (HRLQ) problem, where hospitals declare both \emph{lower} and \emph{upper quotas} which bound the number of residents they can accept; as before, residents rank hospitals and vice versa.
Now it is no longer true that a stable assignment always exists.
The possible non-existence of stable assignments motivates the design of algorithms that find an assignment with a minimum number of blocking pairs; this is the task of the HRLQ problem.
Indeed, the HRLQ problem and its variants have recently gained quite some interest from the algorithmic community~
\cite{BiroEtAl2010,CechlarovaFleiner2016,FleinerKamiyama2014,GotoEtAl2016,HamadaEtAl2014,Huang2010,Kamiyama2013,MonteTumennasan2013,Yokoi2017b,Yokoi2017}. 
In his book, Manlove~\cite[Chapter 5.2]{Manlove2013} devotes an entire chapter to the algorithmics of different versions of the HRLQ problem.

The reason for this high interest in HRLQ is explained by its importance in several real-world matching markets~\cite{FragiadakisEtAl2016,FragiadakisTroyan2017,Veskioja2005} such as school admission systems, centralized assignment of residents to hospitals, or of cadets to military branches.
Lower quotas are a common feature of such admission systems.
Their purpose is often to remedy the effects of under-staffing that are explained by the well-known Rural Hospitals Theorem~\cite{GaleSotomayor1985}: as an example, governments usually want to assign at least a small number of medical residents to each rural hospital to guarantee a minimum service level.
Minimum quotas are also discussed in controlled school choice programs~\cite{EhlersEtAl2014,KominersSonmez2016,Westkamp2013}
where students are divided into a small number of types, and schools set lower bounds for each type.
Such models can represent various forms of affirmative actions taken by schools to, e.g., admit a certain number of minority students~\cite{EhlersEtAl2014}.
Another example is the German university admission system for admitting students to highly oversubscribed subjects, where a certain percentage of study places is assigned according to high school grades or waiting time~\cite{Westkamp2013}.
But lower quotas may also arise due to financial considerations: for instance, a business course with too few (tuition-paying) attendees may not be profitable.
Certain aspects of airline preferences for seat upgrade allocations can be also modelled by lower quotas~\cite{KominersSonmez2016}.

Much of the algorithmic research found that the HRLQ problem (in its different variants) is $\mathsf{NP}$-hard, and thus considered intractable.
A common approach then to identify tractable (polynomial-time solvable) cases of HRLQ; this avenue has lead to several beautiful algorithms~\cite{FleinerKamiyama2014,Huang2010,Yokoi2017}.
However, restricting to polynomial-time solvability necessarily means (if $\mathsf{P} \not= \mathsf{NP}$) that some original features of HRLQ must be restricted more or less, which may be undesirable in the application.
Another approach aimed at addressing the intractability of HRLQ was the design of approximation algorithms.
Unfortunately, it turns out that non-trivial approximation algorithms for HRLQ are highly unlikely to exist: Hamada et al.~\cite{HamadaEtAl2014} showed that, unless $\mathsf{P} = \mathsf{NP}$, no algorithm with approximation guarantee $(|\mathcal H| + |\mathcal R|)^{1-\varepsilon}$ can exist for any $\varepsilon > 0$ (which they complement by an algorithm with approximation guarantee $(|\mathcal H| + |\mathcal R|)$).
In light of their strong inapproximability bounds, Hamada et al. explicitly suggested to consider a more relaxed notion of tractability for HRLQ, namely, fixed-parameter tractability.
They particularly asked whether HRLQ is fixed-parameter tractable parameterized by the minimum number~$b$ of blocking pairs over all matchings meeting all lower and upper quota requirements.

In this paper we follow this avenue, and study the fixed-parameter tractabi\-lity of HRLQ.
It allows us to provide a fine-grained analysis of HRLQ, and the design of efficient algorithms for $\mathsf{NP}$-hard variants of HRLQ for small parameter values.
Our main focus will be the case of HRLQ when each hospital has unit upper quota.
The reason is that by the frequently applied method of ``cloning'' hospitals, stable instances of HRLQ reduce to the case where each hospital has unit upper quota---and thus we can reduce stable instances of HRLQ to the SM setting with lower quotas.
In fact, this is equivalent to the special case of SM where only a subset of women (or, equivalently, men) are distinguished by having also a unit lower quota.
From now on, we refer to HRLQ with unit upper quotas as SMC (where C means that we have to \emph{cover} the women/men who have a unit lower quota), to the special case of SMC with one-sided covering constraints, linking SMC and HRLQ, as SMC-1.
%
%
%
%
So formally, in SMC a set $\mathcal W^\star$ of women and a set $\mathcal M^\star$ of men are distinguished, and a feasible matching is one where each person in~$\mathcal W^\star\cup\mathcal M^\star$ gets matched.
By the Rural Hospitals Theorem~\cite{GaleSotomayor1985} we know that the set of unmatched men and women is the same in all stable matchings, so clearly, feasible stable matchings may not exist.
Thus, we define the task in SMC as finding a feasible matching with a minimum number of blocking pairs.

Apart from the recent interest in HRLQ, its reduction SMC also serves as a ``modern version'' of a classical problem first introduced by Donald Knuth.
Knuth~\cite{Knuth1976} considered SM with ``arranged marriages'', which are a set $\mathcal X$ of man-woman pairs that must be matched with each other.
He showed that the Gale-Shapley algorithm can be extended to decide the existence of a stable matching still in time $O(n^2)$ for $n$-person instances, but in case of absence of a stable matching, minimizing the number of blocking pairs is $\mathsf{NP}$-hard.
Now in SMC, one does not prescribe any more which woman has to marry which man, but only requires certain women and men to marry without dictating their partner.
There is a natural Turing reduction from SMC to the variant considered by Knuth.
Coarse analysis into polynomial-time solvable and $\mathsf{NP}$-hard cases has been studied by several researchers~\cite{ArulselvanEtAl2016}.


%

Another motivation for studying the SMC problem comes from the following scenario that we dub \emph{Control for Stable Marriage}. 
Consider a two-sided market where each participant of the market expresses its preferences over members of the other party, and some central agent (e.g., a government) performs the task of finding a stable matching in the market. 
It might happen that this central agency wishes to apply a certain \emph{control} on the stable matching produced: it may favour some participants by trying to assign them a partner in the resulting matching.
Such a behaviour might be either malicious (e.g., the central agency may accept bribes and thus favour certain participants) or beneficial (e.g., it may favour those who are at disadvantage, like handicapped or minority participants).
However, there might not be a stable matching that covers all participants the agency wants to favour; thus arises the need to produce a matching that is \emph{as stable as possible} among those that fulfil our constraints---the most natural aim in such a case is to minimize the number of blocking pairs in the produced matching, which yields exactly the SMC problem. 
Similar control problems have been extensively studied in the area of social choice for voting systems~\cite[Chapter 7]{comsoc-handbook} 
and recently also for fair division scenarios~\cite{aziz-walsh-schlotter}, 
but have not yet been considered in connection to stable matchings.

\subsection{Our Results}
\label{sec:ourresults}
We provide an extensive algorithmic analysis of the SMC problem and its special case SMC-1.
In our analysis, we examine how different aspects of the input influence the tractability of these problems.
To this end, we apply the framework of parameterized complexity, which deals with computationally hard problems and focuses on how certain \emph{parameters} of a problem instance influence its tractability; for background, we refer to the book by Cygan et al.~\cite{CyganEtAl2015}.
We aim to design so-called \emph{fixed-parameter algorithms}, which perform well in practice if the value of the parameter on hand is small (for the precise definitions, see Section~\ref{sec:preliminaries}). 

The parameters we consider are
\begin{itemize}
  \item the number $b$ of blocking pairs allowed, 
  \item the number $|\mathcal W^\star|$ of women with covering constraint, 
  \item the number $|\mathcal M^\star|$ of men with covering constraint, 
  \item the maximum length $\Delta_{\mathcal W}$ of women's preference lists, and 
  \item the maximum length $\Delta_{\mathcal M}$ of men's preference lists.
\end{itemize}
The choice of each of these parameters is motivated by the aforementioned applications.
For instance, we seek matchings where ideally no blocking pairs at all or at least only few of them appear, to ensure stability of the matching and happiness of those getting matched. The number of women/men with covering constraints corresponds, for instance, to the number of rural hospitals for which a minimum quota specifically must be enforced, which we can expect to be small among the set of all hospitals accepting medical residents.
Finally, preference lists of hospitals and residents can be expected to be small, as each hospital might not rank many more candidates than the number of positions it has to fill, whereas residents might rank only their top choices of hospitals.

We investigate in detail how these parameters influence the complexity of the SMC problem.
A \emph{parameterized restriction of SMC} with respect to the set $S=\{b,|\mathcal W^\star|,|\mathcal M^\star|,\Delta_{\mathcal M},\Delta_{\mathcal W}\}$ means a (possibly parameterized) special case of SMC where each element of $S$ is either restricted to be some constant integer, or regarded as a parameter, or left unbounded. 
Intuitively, these different choices for the elements of $S$ correspond to their expected ``range'' in applications, from very small to mid-range to large (compared to the size of the entire system).
By considering all combinations, we can flexibly model the whole range of applications mentioned above.
We can even cover some cases of \emph{master lists}, where all men's preference lists are restrictions of the exact same total order over the women, and likewise all women's preference lists are restrictions of the exact same total order over the men.

\begin{theorem}
\label{thm:trichotomy-main}
  Any parameterized restriction of SMC with respect to $\{b,|\mathcal W^\star|,|\mathcal M^\star|,\Delta_{\mathcal M},\Delta_{\mathcal W}\}$ is in $\mathsf{P}$, or $\mathsf{NP}$-hard and fixed-parameter tractable, or $\mathsf{NP}$-hard and $\mathsf{W[1]}$-hard with the given parameterization\footnote{Restrictions without any parameters are simply classified as polynomial-time solvable or $\mathsf{NP}$-hard.}, and is covered by one of the results shown in  Table~\ref{tab:summary_of_results}. 
    
  In particular, SMC is $\mathsf{W}[1]$-hard parameterized by $b + |\mathcal W^\star|$, even if there are no distinguished men (i.e., $|\mathcal M^\star|=0$), there is a master list over men as well as one over women, $\Delta_{\mathcal{M}} = 3$, $\Delta_{\mathcal{W}} = 3$ and each distinguished woman finds only a single man acceptable.
\end{theorem}

We give a decision diagram in Section~\ref{sec:discussion} to show that the presented results indeed cover all restrictions of SMC with respect to $\{b,|\mathcal W^\star|,|\mathcal M^\star|,\Delta_{\mathcal M},\Delta_{\mathcal W}\}$. 
Table~\ref{tab:summary_of_results} summarizes our results on the complexity of SMC. 
Note that some results are implied directly by the symmetrical roles of men and women in SMC, and thus are not stated explicitly. 
Here and later, we assume for simplicity that $\Delta_{\mathcal M} \geq 2$ and $\Delta_{\mathcal W} \geq 2$.

\begin{table}[ht]
  \centering
  \begin{tabular}{lll}
    \toprule
    constants & parameters~~~~ &  complexity\\
    \midrule
    $|\mathcal M^\star| = 0$, $|\mathcal W^\star| = 0$  & & in $\mathsf{P}$ (Gale-Shapley alg.)\\
    $|\mathcal M^\star| = 0, |\mathcal W^\star|,\Delta_{\mathcal M}$ & & in $\mathsf{P}$~(Theorem~\ref{thm:smc-allconstant})\\
    $|\mathcal M^\star|,|\mathcal W^\star|,\Delta_{\mathcal M},\Delta_{\mathcal W}$ && in $\mathsf{P}$~(Theorem~\ref{thm:smc-allconstant})\\
    $|\mathcal M^\star| = 0, \Delta_{\mathcal M}\leq 2$ & & in $\mathsf{P}$~(Theorem~\ref{thm:restricted-matching-poly})\\
    $\Delta_{\mathcal W}\leq 2, \Delta_{\mathcal M}\leq 2$ & & in $\mathsf{P}$~(Observation~\ref{observ:maxlist2})\\            
    $b$ & & in $\mathsf{P}$~(Observ.~\ref{observ:bounded-b})\\
    $|\mathcal M^\star| = 0, \Delta_{\mathcal W} = 2, \Delta_{\mathcal M}\geq 3$ & & $\mathsf{NP}$-hard~(Theorem~\ref{thm-case23})\\
    $|\mathcal W^\star| = 1,\Delta_{\mathcal W} = 2$ && $\mathsf{NP}$-hard (Theorem~\ref{thm:smc2-npc-1woman})\\
    $|\mathcal M^\star| = 0, \Delta_{\mathcal W}\geq 3, \Delta_{\mathcal M}\geq 3, \Delta^{\star} = 1$~~ & $b+|\mathcal W^\star|$ & $\mathsf{W}[1]$-hard~(Theorem~\ref{thm:smc-minblock-mainhardness})\\
    $|\mathcal M^\star| = 0, |\mathcal W^\star|\geq 1, \Delta_{\mathcal W} \geq 3, \Delta^\star = 1$~~~ & $b + \Delta_{\mathcal M}$ & $\mathsf{W}[1]$-hard~(Theorem~\ref{thm:1woman-strict})\\
    $\Delta_{\mathcal W}\leq 2$ & $|\mathcal W^\star|+|\mathcal M^\star|$           & $\mathsf{FPT}$ (Theorem~\ref{thm:fpt-case2x})\\
    $\Delta_{\mathcal W}\leq 2$ & $b$           & $\mathsf{FPT}$ (Corollary~\ref{thm:fpt-case2xCor})\\
    $\Delta_{\mathcal W}\leq 2$ & $|\mathcal W^\star|+\Delta_{\mathcal M}$           & $\mathsf{FPT}$ (Theorem~\ref{thm:fpt-case2x-extended})\\

            %
    \bottomrule
  \end{tabular}
  \caption{Summary of our results for {\sc Stable Marriage with Covering Constraints}.
    Here,~$\Delta^\star$ denotes the maximum length of the preference list of any distinguished person. 
    \label{tab:summary_of_results}}
\end{table}

As a special case, we answer a question by Hamada et al.~\cite{HamadaEtAl2014} who gave an exponential-time algorithm that in time $O(|I|^{b+1})$ decides for a given instance~$I$ of HRLQ whether it admits a feasible matching with at most~$b$ blocking pairs\footnote{Hamada et al. claim only a run time $O((|\mathcal W||\mathcal M|)^{b+1})$, but their algorithm can easily be implemented to run in time $O(|I|^{b+1})$.}; the authors asked whether HRLQ is fixed-parameter tractable parameterized by~$b$.
As shown by Theorem~\ref{thm:trichotomy-main}, SMC-1 and therefore also HRLQ is $\mathsf{W}[1]$-hard when parameterized by $b$, already in a very restricted setting. 
Thus, the answer to the question by Hamada et al.~\cite{HamadaEtAl2014} is negative: 
\mbox{SMC-1}, and hence HRLQ, admits no fixed-parameter algorithm with parameter $b$ unless $\mathsf{FPT}=\mathsf{W}[1]$.

\subsection{Related Work}
\label{sec:relatedwork}
There is a dynamically growing literature on matching markets with lower quotas~\cite{BiroEtAl2010,CechlarovaFleiner2016,FleinerKamiyama2014,FragiadakisEtAl2016,FragiadakisTroyan2017,GotoEtAl2016,HamadaEtAl2014,Huang2010,Kamiyama2013,MonteTumennasan2013,Yokoi2017b,Yokoi2017}. 
These papers study several variants of HRLQ, 
adapting the general model to the various particularities of practical problems. 
However, there are only a few papers which consider the problem of minimizing the number of blocking pairs~\cite{FragiadakisEtAl2016,HamadaEtAl2014}.
The most closely related work to ours is the paper by Hamada et al.~\cite{HamadaEtAl2014}: they prove that the HRLQ problem is $\mathsf{NP}$-hard and give strong inapproximability results; they also consider the SMC-1 problem directly and propose an $O(|I|^{b+1})$ time algorithm for it.  

A different line of research connected to SMC is the problem of \emph{arranged marriages}, an early extension of SM suggested by  
Knuth~\cite{Knuth1976} in 1976.
Here, a set $\mathcal Q^\star$ of man-woman pairs is distinguished, and we seek a stable matching that contains~$\mathcal Q^\star$ as a subset.
Thus, as opposed to SMC, we not only require that each distinguished person is assigned \emph{some} partner, but instead prescribe its partner exactly. 
Initial work on arranged marriages \cite{Knuth1976,GusfieldIrving1989} was extended by Dias et al.~\cite{DiasEtAl2003} to consider also \emph{forbidden marriages}, and was further generalized by Fleiner et al.~\cite{Fleineretal2007} and Cseh and Manlove~\cite{CsehManlove2016}. 
Despite the similar flavour of the studied problems, none of these papers have a direct consequence on the complexity of SMC. 

Our work also fits into the line of research that addresses computationally hard problems in the area of stable matchings by focusing on instances with bounded preference lists~\cite{BiroMcDermid2012,ImmorlicaMahdian2005,IrvingEtAl2009,KojimaEtAl2013,RothPeranson1999} or by applying the more flexible approach of parameterized complexity~\cite{ArulselvanEtAl2016,AzizEtAl2015,BiroEtAl2011,MarxSchlotter2010,MarxSchlotter2011}.

\medskip
\noindent
\textbf{Organization.} 
After the preliminaries in Section~\ref{sec:preliminaries}, 
we start with the main intractability result in Section~\ref{sec:hrlq-min-blocking-pairs}, which answers Hamada et al.'s question. 
This result shows $\mathsf{W}[1]$-hardness of SMC parameterized by $b+|\mathcal{W}^\star|$ 
even when $\mathcal{M}^\star=\emptyset$ and $\Delta_{\mathcal M}=\Delta_{\mathcal W}=3$. 
Thus, we explore three directions to achieve tractability: 
(i) to lower $b$ to be a constant, 
(ii) to lower $|\mathcal W^\star|$ to be a constant, or 
(iii) to lower either $\Delta_{\mathcal W}$ or $\Delta_{\mathcal M}$ to $2$. 
We cover the cases (i) and (ii) in Section~\ref{sec:boundednumberofwomentobecovered}, 
and case (iii) in Section~\ref{sec:prefatmost2}. 
In addition, Section~\ref{sec:hrlq-min-blocking-pairs-shortlists} provides polynomial-time approximation results for HRLQ and SMC, used also in the polynomial-time algorithms of Section~\ref{sec:boundednumberofwomentobecovered}. 

\section{Preliminaries}
\label{sec:preliminaries}
An instance $I$ of the {\sc Stable Marriage (SM)} problem consists of a set~$\mathcal M$ of men and a set $\mathcal W$ of women.
Each person $x \in \mathcal M\cup\mathcal W$ has a preference list~$L(x)$ that strictly orders the members of the other party acceptable for $x$.
We thus write $L(x)$ as a vector $L(x) = (y_1,\hdots,y_t)$, denoting that~$y_i$ is (strictly) \emph{preferred} by $x$ over $y_j$ for each $i$ and $j$ with $1 \leq i<j \leq t$.
A \emph{matching}~$M$ for~$I$ is a set of man-woman pairs appearing in each other's preference lists such that each person is contained in at most one pair of $M$; some persons may be left unmatched by~$M$.
For each person~$x$ we denote by $M(x)$ the person assigned by~$M$ to~$x$.
For a matching~$M$, a man $m$ and a woman~$w$ included in each other's preference lists form a \emph{blocking pair} if (i) $m$ is either unmatched or prefers~$w$ to~$M(m)$, and (ii)~$w$ is either unmatched or prefers $m$ to $M(w)$.
In the {\sc Stable Marriage with Covering Constraints (SMC)} problem, we are given additional subsets $\mathcal W^\star\subseteq\mathcal W$ and $\mathcal M^\star\subseteq\mathcal M$ of \emph{distinguished} people that must be matched; a matching~$M$ is \emph{feasible} if it matches everybody in $\mathcal W^\star\cup\mathcal M^\star$.
The objective of SMC is to find a feasible matching for $I$  with minimum number of blocking pairs.
If only people from one gender are distinguished, then without loss of generality, we assume these to be women; this special case will be denoted by SMC-1.

The many-to-one extension of SMC-1 is the {\sc Hospitals/Residents with Lower Quotas (HRLQ)} problem whose input consists of a set $\mathcal R$ of residents and a set $\mathcal H$ of hospitals that have ordered preferences over the acceptable members of the other party.
%
Each hospital $h\in\mathcal H$ has a quota lower bound~$\mathsf{\underline{q}}(h)$ and a quota upper bound~$\mathsf{\overline{q}}(h)$, which bound the number of residents that can be assigned to $h$ from below and above.
One seeks an \emph{assignment}~$M$ that maps a subset of the residents to hospitals that respects acceptability
and is \emph{feasible}, that is, $\mathsf{\underline{q}}(h) \leq |M(h)| \leq \mathsf{\overline{q}}(h)$ for each hospital~$h$.
Here, $M(h)$ is the set of residents assigned to some $h \in \mathcal H$ by $M$.
We say that a hospital~$h$ is \emph{under-subscribed} if $|M(h)| < \mathsf{\overline{q}}(h)$.
For an assignment~$M$ of an instance of HRLQ, a pair $\{r,h\}$ of a resident~$r$ and a hospital~$h$ is \emph{blocking} if (i) $r$ is unassigned or prefers $h$ to the hospital assigned to $r$ by $M$, and 
(ii)~$h$ is under-subscribed or prefers~$r$ to one of the residents in~$M(h)$.
The task in HRLQ is to find a feasible assignment with minimum number of blocking pairs.

Some instances of SMC may admit a \emph{master list} over women, which is a total ordering~$L_{\mathcal{W}}$ of all women, such that for each man $m \in \mathcal{M}$, the preference list~$L(m)$ is the restriction of~$L_{\mathcal{W}}$ to those women that $m$ finds acceptable. 
Similarly, we consider master lists over men.

With each instance $I$ of SMC (or HRLQ) we can naturally associate a bipartite graph~$G_I$ whose vertex partitions correspond to $\mathcal M$ and $\mathcal W$ (or $\mathcal R$ and~$\mathcal H$, respectively), and there is an edge between a man $m\in \mathcal M$ and a woman $w\in\mathcal W$ (or between a resident $r\in\mathcal R$ and a hospital $h\in\mathcal H$, respectively) if they appear in each other's preference lists.
We may refer to entities of~$I$ as vertices, or a pair of entities as edges, without mentioning $G_I$ explicitly.
For a graph~$G$, we denote its vertex set by $V(G)$ and its edge set by $E(G)$; furthermore, let~$d_G(v)$ denote the degree of vertex $v\in V(G)$ in $G$.
A \emph{path} $P$ in $G$ is a series of vertices that contains each vertex at most once, with an edge of $G$ connecting any two consecutive vertices of $P$.
With a slight abuse of the notation, we will often identify paths with their edge sets; 
we will write $P' \subseteq P$ to express that $P'$ is a subpath of $P$.
A \emph{matching} in $G$ is a set of edges $M \subseteq E(G)$ such that no vertex in $G$ is adjacent to more than one edges of $M$;
note that a matching in an instance $I$ of SMC indeed corresponds to a matching in the graph $G_I$. 
For a matching $M$ in $G$, a path~$P$ is called \emph{$M$-alternating}, if among any two consecutive edges along~$P$ exactly one belongs to $M$.
We will use a few other notions from the theory of matchings about the symmetric difference of matchings, 
see e.g., the book by Lov{\'a}sz and Plummer~\cite{LovaszPlummer-book} for an introduction to this topic.

\paragraph{Parameterized complexity.}
The framework of parameterized complexity deals with computationally hard problems, examining their complexity in a more detailed way than classical complexity theory.
In a parameterized problem problem $\Pi$, each input instance $I$ is associated with an integer~$k$ called the \emph{parameter}.
An algorithm which decides instances $I$ of $\Pi$ in time $f(k) \cdot |I|^{O(1)}$ for some computable function~$f$ is called a \emph{fixed-parameter algorithm}.
Note that the dependence of the polynomial in the run time is constant, but the dependence on the parameter~$k$ can be arbitrary (and is typically exponential).
However, if the parameter of a given instance is small, then such an algorithm can be useful 
in practice even if the overall size of the instance is large. 

The class of problems admitting fixed-parameter algorithms is denoted by $\mathsf{FPT}$.
To argue that a problem is \emph{not} in $\mathsf{FPT}$, parameterized complexity provides a hardness theory.
For two parameterized problems $\Pi_1$ and $\Pi_2$, a \emph{parameterized reduction} from $\Pi_1$ to $\Pi_2$ is a function~$f$, computable by a fixed-parameter algorithm, that maps each instance $(I_1,k_1)$ of $\Pi_1$ to an instance $f(I_1,k_1)=(I_2,k_2)$ of $\Pi_2$ such that
(i) $(I_1,k_1)$ is a ``yes''-instance of $\Pi_1$ if and only if $(I_2,k_2)$ is a ``yes''-instance of~$\Pi_2$, and 
(ii) $k_2 \leq g(k_1)$ for some function~$g$.
The basic class of parameterized intractability is $\mathsf{W}[1]$: proving a problem $\Pi$ to be $\mathsf{W}[1]$-hard is strong evidence that $\Pi \notin \mathsf{FPT}$.
Given some problem $\Pi$ that is known to be $\mathsf{W}[1]$-hard, a parameterized reduction from $\Pi$ to some parameterized problem $\Pi'$ implies $\mathsf{W}[1]$-hardness of $\Pi'$ as well. 

For more on parameterized complexity, we refer the reader to the book by Cygan et al.~\cite{CyganEtAl2015}.

\section{Strong Parameterized Intractability of SMC}
\label{sec:hrlq-min-blocking-pairs}
This section provides parameterized intractability and inapproximability results for SMC showing the hardness of finding feasible matchings with minimum number of blocking pairs.
Namely, we prove SMC-1 to be  $\mathsf{W}[1]$-hard parameterized by the number $b$ of blocking pairs we aim for plus the number~$|\mathcal W^\star|$ of distinguished women, even in a very restricted setting.

\begin{theorem}
\label{thm:smc-minblock-mainhardness}
  SMC-1 is $\mathsf{W}[1]$-hard parameterized by $b+|\mathcal W^\star|$, even if there is a master list over men as well as one over women, all preference lists are of length at most $3$, and $|L(w)|=1$ for each woman $w \in \mathcal W^\star$.
\end{theorem}

Before proving Theorem~\ref{thm:smc-minblock-mainhardness}, 
let us quickly state a simple but useful claim.
\begin{proposition}
\label{prop:alt-paths-blocking-pairs}
  Let $M_1$ and $M_2$ be two matchings in an instance of SMC.
  Let $v_0 v_1 \dots v_p$ (with $p \geq 1$) be a maximal path in the symmetric differenc of $M_1$ and $M_2$, denoted by $M_1 \triangle M_2$.
  Then
  \begin{enumerate}
    \item[(a)] $P$ contains an edge that blocks either $M_1$ or $M_2$, and
    \item[(b)] if $i \in \{1,\dots, p-1\}$ is such that $v_i$ prefers $v_{i-i}$ to $v_{i+1}$, then $P_i=v_0 v_1 \dots v_i$ contains an edge that blocks either $M_1$ or $M_2$. 
  \end{enumerate}
\end{proposition}
\begin{proof}
  We call a person $v_i$ a \emph{leftist} if either $i=p$, or $i \in \{1,\dots, p-1\}$ and~$v_i$ prefers $v_{i-1}$ to $v_{i+1}$.
  Similarly, we call $v_i$ a \emph{rightist} if either $i=0$, or $i \in \{1,\dots, p-1\}$ and $v_i$ prefers $v_{i+1}$ to $v_{i-1}$.
  Observe that any person on $P$ is either a leftist or a rightist. 
  Moreover, the path $P$, and under the conditions of (b) also the path $P_i$, must contain an edge $\{x,y\}$ such that $x$ is a leftist and~$y$ is a rightist.
  Let $M_i$ be the matching that does \emph{not} contain $\{x,y\}$ (where $i \in \{1,2\}$). 
  Then both $x$ and $y$ prefer being matched to each other as opposed to their situation in $M_i$ (where they may or may not be matched), proving both (a) and (b).
\end{proof}

\begin{proof}[Proof of Theorem~\ref{thm:smc-minblock-mainhardness}]
  We give a reduction from the $\mathsf{W}[1]$-hard {\sc Multicolored Clique} parameterized by the size of the solution~\cite{FellowsEtAl2009}. 
  Let $G$ be the input graph, with its vertex set partitioned into $k$ sets $V_1, \hdots, V_k$; the task is to find a clique of size $k$ in~$G$ containing exactly one vertex from each of the sets~$V_i$. 
  We let~$E_{i,j}$ denote those edges that run between~$V_i$ and~$V_j$ for some $1 \leq i < j \leq k$. 
  We fix an ordering on the vertices and edges of $G$ that places vertices of $V_i$ before vertices of~$V_j$ whenever $i<j$ (the ordering on the edges of $G$ can be chosen arbitrarily). 
  We will write $\suc(x)$ to denote the vertex following~$x$ in this ordering, and we let $v_i^1$ and $v_i^{\infty}$ denote the first and last vertices in~$V_i$, respectively.
  Similarly, we write $\suc(\{x,y\})$ for the edge following $\{x,y\}$, and we let~$e_{i,j}^1$ and $e_{i,j}^{\infty}$ denote the first and last edges in~$E_{i,j}$, respectively.
  We will also write $\pre(x)$ and $\pre(\{x,y\})$ for the predecessor of~$x$ or $\{x,y\}$, respectively.
  Also, we denote the $h$-th neighbor of some vertex~$x$ as $n(x,h)$.
  For simplicity, we assume that $G$ has no isolated vertices. 

  We construct an instance $I$ of {\sc SMC} as follows; see Figure~\ref{fig:2a} and Figure~\ref{fig:2b} for an illustration.
  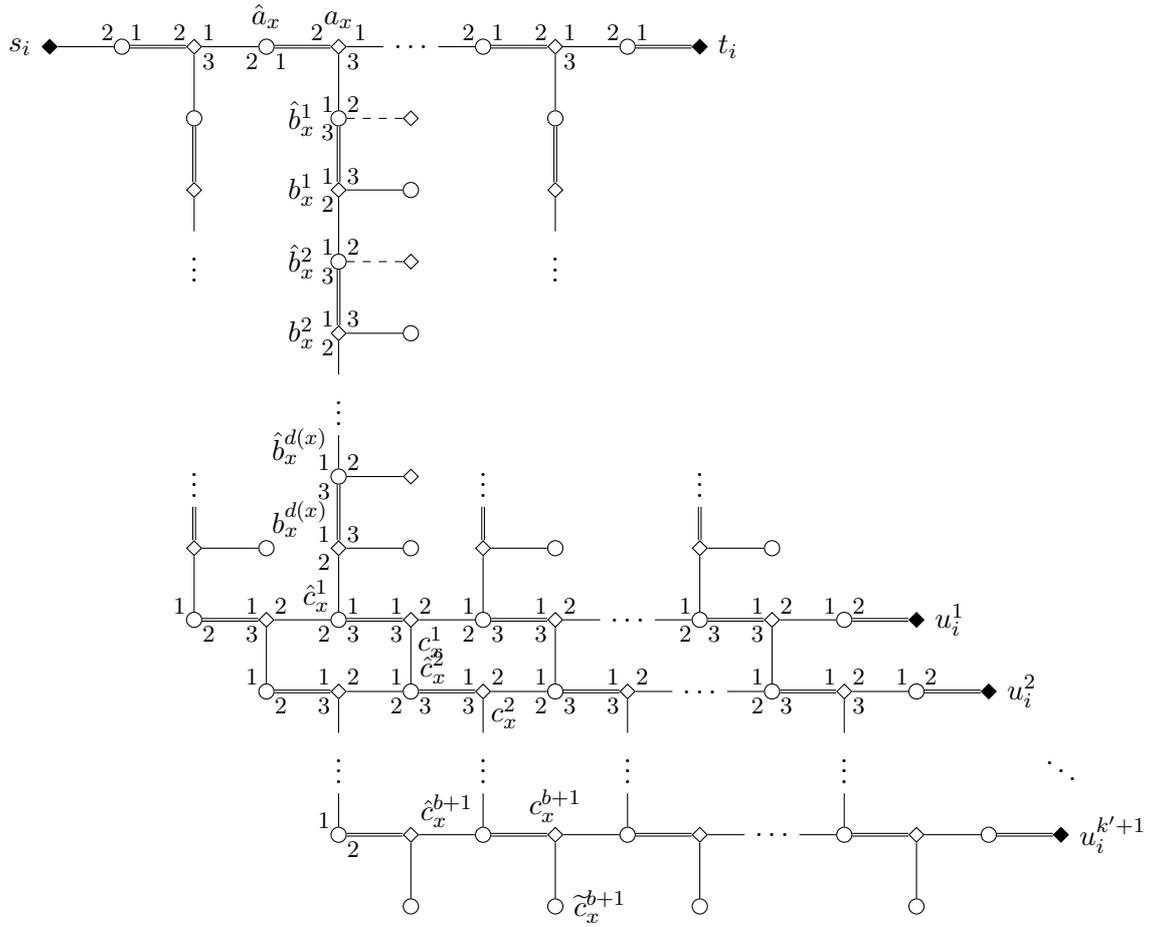
\begin{figure}[thpb]
    \tikzset{rdvertex/.style={minimum size=2mm,circle,fill=white,draw, inner sep=0pt},
      sqvertex/.style={minimum size=2mm,diamond,fill=black,draw, inner sep=0pt},
      opvertex/.style={minimum size=2mm,diamond,fill=white,draw, inner sep=0pt},
      decoration={markings,mark=at position .5 with {\arrow[black,thick]{stealth}}}}
    \centering
    \begin{tikzpicture}[scale=0.7]
        \node (sx) at (0,0)[sqvertex,label=left:$s_i$]{};
        \node (f1) at (1,0)[rdvertex]{};
        \node (f1-1) at (0.8,0.2){{\footnotesize $2$}};
        \node (f1-2) at (1.2,0.2){{\footnotesize $1$}};
        \node (f2) at (2,0)[opvertex]{};
        \node (f2-1) at (1.8,0.2){{\footnotesize $2$}};
        \node (f2-2) at (2.2,0.2){{\footnotesize $1$}};
        \node (f2-3) at (2.2,-0.2){{\footnotesize $3$}};
        \node (f3) at (2,-1)[rdvertex]{};
        \node (f4) at (2,-2)[opvertex]{};
        \node (vdots1) at (2,-3){$\vdots$};
        \node (f5) at (3,0)[rdvertex,label=above:$\hat{a}_x$]{};
        \node (f5-1) at (2.8,-0.2){{\footnotesize $2$}};
        \node (f5-2) at (3.2,-0.2){{\footnotesize $1$}};
        \node (f6) at (4,0)[opvertex,label=above:$a_x$]{};
        \node (f6-1) at (3.7,0.2){{\footnotesize $2$}};
        \node (f6-2) at (4.3,0.2){{\footnotesize $1$}};
        \node (f6-3) at (4.2,-0.2){{\footnotesize $3$}};
        \node (hdots1) at (5,0){$\hdots$};
        \node (f7) at (4,-1)[rdvertex,label={[label distance=2pt]left:$\hat{b}_x^1$}]{};
        \node (f7-1) at (3.85,-0.8){{\footnotesize $1$}};
        \node (f7-2) at (3.85,-1.2){{\footnotesize $3$}};
        \node (f7-3) at (4.2,-0.8){{\footnotesize $2$}};
        \node (f8) at (5,-1)[opvertex]{};
        \node (f9) at (4,-2)[opvertex,label={[label distance=2pt]left:$b_x^1$}]{};
        \node (f9-1) at (3.85,-1.8){{\footnotesize $1$}};
        \node (f9-2) at (3.85,-2.2){{\footnotesize $2$}};
        \node (f9-3) at (4.2,-1.8){{\footnotesize $3$}};
        \node (f10) at (5,-2)[rdvertex]{};
        \node (f11) at (4,-3)[rdvertex,label={[label distance=2pt]left:$\hat{b}_x^2$}]{};
        \node (f11-1) at (3.85,-2.8){{\footnotesize $1$}};
        \node (f11-2) at (3.85,-3.2){{\footnotesize $3$}};
        \node (f11-3) at (4.2,-2.8){{\footnotesize $2$}};
        \node (f12) at (5,-3)[opvertex]{};
        \node (f13) at (4,-4)[opvertex,label={[label distance=2pt]left:$b_x^2$}]{};
        \node (f13-1) at (3.85,-3.8){{\footnotesize $1$}};
        \node (f13-2) at (3.85,-4.2){{\footnotesize $2$}};
        \node (f13-3) at (4.2,-3.8){{\footnotesize $3$}};
        \node (f14) at (5,-4)[rdvertex]{};
        \node (vdots2) at (4,-5){$\vdots$};
        \node (f15) at (4,-6)[rdvertex,label={[label distance=-2pt]100:$\hat{b}_x^{d_G(x)}$}]{};
        \node (f15-1) at (3.8,-5.8){{\footnotesize $1$}};
        \node (f15-2) at (3.8,-6.2){{\footnotesize $3$}};
        \node (f15-3) at (4.2,-5.8){{\footnotesize $2$}};
        \node (f16) at (5,-6)[opvertex]{};
        \node (f17) at (4,-7)[opvertex,label={[label distance=-2pt]100:$b_x^{d_G(x)}$}]{};
        \node (f17-1) at (3.8,-6.8){{\footnotesize $1$}};
        \node (f17-2) at (3.8,-7.2){{\footnotesize $2$}};
        \node (f17-3) at (4.2,-6.8){{\footnotesize $3$}};
        \node (f18) at (5,-7)[rdvertex]{};
        \draw (sx)--(f1);
        \draw[double] (f1)--(f2);
        \draw (f2)--(f3);
        \draw[double] (f3)--(f4);
        \draw (f4)--(vdots1);
        \draw (f2)--(f5);
        \draw[double] (f5)--(f6);
        \draw (f6)--(hdots1);
        \draw (f6)--(f7);
        \draw[dashed] (f7)--(f8);
        \draw[double] (f7)--(f9);
        \draw (f9)--(f10);
        \draw (f9)--(f11);
        \draw[dashed] (f11)--(f12);
        \draw[double] (f11)--(f13);
        \draw (f13)--(f14);
        \draw (f13)--(vdots2);
        \draw (vdots2)--(f15);
        \draw (f15)--(f16);
        \draw[double](f15)--(f17);
        \draw (f17)--(f18);
        \node (f19) at (6,0)[rdvertex]{};
        \node (f19-1) at (5.8,0.2){{\footnotesize $2$}};
        \node (f19-2) at (6.2,0.2){{\footnotesize $1$}};
        \node (f20) at (7,0)[opvertex]{};
        \node (f20-1) at (6.8,0.2){{\footnotesize $2$}};
        \node (f20-2) at (7.2,0.2){{\footnotesize $1$}};
        \node (f20-3) at (7.2,-0.2){{\footnotesize $3$}};
        \node (f21) at (7,-1)[rdvertex]{};
        \node (f22) at (7,-2)[opvertex]{};
        \node (vdots3) at (7,-3){$\vdots$};
        \node (f23) at (8,0)[rdvertex]{};
        \node (f23-1) at (7.8,0.2){{\footnotesize $2$}};
        \node (f23-2) at (8.2,0.2){{\footnotesize $1$}};
        \node (f24) at (9,0)[sqvertex,label=right:$t_i$]{};
        \draw (hdots1)--(f19);
        \draw[double] (f19)--(f20);
        \draw (f20)--(f21);
        \draw[double] (f21)--(f22);
        \draw (f22)--(vdots3);
        \draw (f20)--(f23);
        \draw[double] (f23)--(f24);
        \node (vdots4) at (2,-6){$\vdots$};
        \node (f25) at (2,-7)[opvertex]{};
        \node (f26) at (3,-7)[rdvertex]{};
        \node (f27) at (2,-8)[rdvertex]{};
        \node (f27-1) at (1.8,-7.8){{\footnotesize $1$}};
        \node (f27-2) at (2.2,-8.2){{\footnotesize $2$}};
        \node (f28) at (3,-8)[opvertex]{};
        \node (f28-1) at (2.8,-7.8){{\footnotesize $1$}};
        \node (f28-2) at (3.2,-7.8){{\footnotesize $2$}};
        \node (f28-3) at (2.8,-8.2){{\footnotesize $3$}};
        \node (f29) at (4,-8)[rdvertex,label={[label distance=-4pt]120:$\hat{c}_x^1$}]{}; 
        \node (f29-1) at (4.2,-7.8){{\footnotesize $1$}};
        \node (f29-2) at (4.2,-8.2){{\footnotesize $3$}};
        \node (f29-3) at (3.8,-8.2){{\footnotesize $2$}};
        \node (f30) at (5,-8)[opvertex,label={[label distance=-5pt]below right:$c_x^1$}]{};
        \node (f30-1) at (4.8,-7.8){{\footnotesize $1$}};
        \node (f30-2) at (5.2,-7.8){{\footnotesize $2$}};
        \node (f30-3) at (4.8,-8.2){{\footnotesize $3$}};   
        \node (f31) at (6,-8)[rdvertex]{};
        \node (f31-1) at (5.8,-7.8){{\footnotesize $1$}};
        \node (f31-2) at (6.2,-8.2){{\footnotesize $3$}};
        \node (f31-3) at (5.8,-8.2){{\footnotesize $2$}};
        \node (f32) at (6,-7)[opvertex]{};
        \node (f33) at (7,-7)[rdvertex]{};
        \node (vdots5) at (6,-6){$\vdots$};
        \node (f34) at (7,-8)[opvertex]{};
        \node (f34-1) at (6.8,-7.8){{\footnotesize $1$}};
        \node (f34-2) at (7.2,-7.8){{\footnotesize $2$}};
        \node (f34-3) at (6.8,-8.2){{\footnotesize $3$}};   
        \node (hdots2) at (8,-8){$\hdots$};
        \node (f35) at (3,-9)[rdvertex]{};
        \node (f35-1) at (2.8,-8.8){{\footnotesize $1$}};
        \node (f35-2) at (3.2,-9.2){{\footnotesize $2$}};
        \node (f36) at (4,-9)[opvertex]{};
        \node (f36-1) at (3.8,-8.8){{\footnotesize $1$}};
        \node (f36-2) at (4.2,-8.8){{\footnotesize $2$}};
        \node (f36-3) at (3.8,-9.2){{\footnotesize $3$}};   
        \node (f37) at (5,-9)[rdvertex,label={[label distance=-5.5pt]60:$\hat{c}_x^2$}]{};
        \node (f37-1) at (4.8,-8.8){{\footnotesize $1$}};
        \node (f37-2) at (5.2,-9.2){{\footnotesize $3$}};
        \node (f37-3) at (4.8,-9.2){{\footnotesize $2$}};
        \node (f38) at (6,-9)[opvertex,label={[label distance=-4pt]-60:$c_x^2$}]{};
        \node (f38-1) at (5.8,-8.8){{\footnotesize $1$}};
        \node (f38-2) at (6.2,-8.8){{\footnotesize $2$}};
        \node (f38-3) at (5.8,-9.2){{\footnotesize $3$}};   
        \node (f39) at (7,-9)[rdvertex]{};
        \node (f39-1) at (6.8,-8.8){{\footnotesize $1$}};
        \node (f39-2) at (7.2,-9.2){{\footnotesize $3$}};
        \node (f39-3) at (6.8,-9.2){{\footnotesize $2$}};
        \node (f40) at (8,-9)[opvertex]{};
        \node (f40-1) at (7.8,-8.8){{\footnotesize $1$}};
        \node (f40-2) at (8.2,-8.8){{\footnotesize $2$}};
        \node (f40-3) at (7.8,-9.2){{\footnotesize $3$}};   
        \node (hdots3) at (9,-9){$\hdots$};
        \node (vdots6) at (4,-10){$\vdots$};
        \node (vdots7) at (6,-10){$\vdots$};
        \node (vdots8) at (8,-10){$\vdots$};
        \node (f41) at (4,-11)[rdvertex]{};
        \node (f41-1) at (3.8,-10.8){{\footnotesize $1$}};
        \node (f41-2) at (4.2,-11.2){{\footnotesize $2$}};
        \node (f42) at (5,-11)[opvertex]{};
        \node (f43) at (6,-11)[rdvertex,label={[label distance=-4pt]100:$\hat{c}_x^{b+1}$}]{};
        \node (f44) at (7,-11)[opvertex,label={[label distance=-2pt]above:$c_x^{b+1}$}]{};
        \node (f45) at (8,-11)[rdvertex]{};
        \node (f46) at (9,-11)[opvertex]{};
        \node (hdots4) at (10,-11){$\hdots$};
        \node (f47) at (5,-12)[rdvertex]{};
        \node (f48) at (7,-12)[rdvertex,label=right:$\widetilde{c}_x^{b+1}$]{};
        \node (f49) at (9,-12)[rdvertex]{};
        \draw[double] (vdots4)--(f25);
        \draw (f25)--(f26);
        \draw (f25)--(f27);
        \draw[double] (f27)--(f28);
        \draw (f28)--(f29);
        \draw (f29)--(f17); 
        \draw[double] (f29)--(f30);    
        \draw (f30)--(f31);
        \draw (f31)--(f32);
        \draw[double] (f32)--(vdots5);
        \draw (f32)--(f33);
        \draw[double] (f31)--(f34);
        \draw (f34)--(hdots2);
        \draw (f28)--(f35);
        \draw[double] (f35)--(f36);
        \draw (f36)--(f37);
        \draw (f37)--(f30);
        \draw[double] (f37)--(f38);
        \draw (f38)--(f39);
        \draw (f39)--(f34);
        \draw[double] (f39)--(f40);
        \draw (f40)--(hdots3);
        \draw (f36)--(vdots6);
        \draw (f38)--(vdots7);
        \draw (f40)--(vdots8);
        \draw (vdots6)--(f41);
        \draw[double] (f41)--(f42);
        \draw (f42)--(f43);
        \draw (vdots7)--(f43);
        \draw[double] (f43)--(f44);
        \draw (f44)--(f45);
        \draw (vdots8)--(f45);
        \draw[double] (f45)--(f46);
        \draw (f46)--(hdots4);
        \draw (f47)--(f42);
        \draw (f48)--(f44);
        \draw (f49)--(f46);
        \node (f50) at (9,-8)[rdvertex]{};
        \node (f50-1) at (8.8,-7.8){{\footnotesize $1$}};
        \node (f50-2) at (9.2,-8.2){{\footnotesize $3$}};
        \node (f50-3) at (8.8,-8.2){{\footnotesize $2$}};
        \node (f51) at (9,-7)[opvertex]{};
        \node (f52) at (10,-7)[rdvertex]{};
        \node (vdots9) at (9,-6){$\vdots$};
        \node (f53) at (10,-8)[opvertex]{};
        \node (f53-1) at (9.8,-7.8){{\footnotesize $1$}};
        \node (f53-2) at (10.2,-7.8){{\footnotesize $2$}};
        \node (f53-3) at (9.8,-8.2){{\footnotesize $3$}};
        \node (f54) at (11,-8)[rdvertex]{};
        \node (f54-1) at (10.8,-7.8){{\footnotesize $1$}};
        \node (f54-2) at (11.2,-7.8){{\footnotesize $2$}};
        \node (f55) at (12,-8)[sqvertex,label=right:$u_i^1$]{};
        \node (f56) at (10,-9)[rdvertex]{};
        \node (f56-1) at (9.8,-8.8){{\footnotesize $1$}};
        \node (f56-2) at (10.2,-9.2){{\footnotesize $3$}};
        \node (f56-3) at (9.8,-9.2){{\footnotesize $2$}};
        \node (f57) at (11,-9)[opvertex]{};
        \node (f57-1) at (10.8,-8.8){{\footnotesize $1$}};
        \node (f57-2) at (11.2,-8.8){{\footnotesize $2$}};
        \node (f57-3) at (11.2,-9.2){{\footnotesize $3$}};
        \node (f58) at (12,-9)[rdvertex]{};
        \node (f58-1) at (11.8,-8.8){{\footnotesize $1$}};
        \node (f58-2) at (12.2,-8.8){{\footnotesize $2$}};    
        \node (f59) at (13,-9)[sqvertex,label=right:$u_i^2$]{};
        \node (vdots10) at (11,-10){$\vdots$};
        \node (ddots1) at (14,-10){$\ddots$};
        \node (f60) at (11,-11)[rdvertex]{};
        \node (f61) at (12,-11)[opvertex]{};
        \node (f62) at (13,-11)[rdvertex]{};
        \node (f63) at (14,-11)[sqvertex,label=right:$u_i^{b+1}$]{};
        \node (f64) at (12,-12)[rdvertex]{};
        \draw (hdots2)--(f50);
        \draw (f50)--(f51);
        \draw (f51)--(f52);
        \draw[double] (f51)--(vdots9);
        \draw[double] (f50)--(f53);
        \draw (f53)--(f54);
        \draw[double] (f54)--(f55);
        \draw (hdots3)--(f56);
        \draw (f53)--(f56);
        \draw[double] (f56)--(f57);
        \draw (f57)--(vdots10);
        \draw (f57)--(f58);
        \draw[double] (f58)--(f59);
        \draw (hdots4)--(f60);
        \draw (f60)--(vdots10);
        \draw[double] (f60)--(f61);
        \draw (f61)--(f62);
        \draw[double] (f62)--(f63);
        \draw (f61)--(f64);
        \end{tikzpicture}
  \caption{Node selecting gadget $G_i$ in the proof of Theorem~\ref{thm:smc-minblock-mainhardness}.
    Throughout the paper, we use squares for women, circles for men; 
    distinguished persons are denoted by filled squares/circles.  
    The numbering of edges incident to some vertex (or, sometimes, arrows between edges) indicate preferences. 
    Double edges denote edges of the stable matching $M_s$, and dashed edges are those leaving some gadget.
  \label{fig:2a}}
\end{figure}

\begin{figure}[thpb]
    \tikzset{rdvertex/.style={minimum size=2mm,circle,fill=white,draw, inner sep=0pt},
        sqvertex/.style={minimum size=2mm,diamond,fill=black,draw, inner sep=0pt},
        opvertex/.style={minimum size=2mm,diamond,fill=white,draw, inner sep=0pt},
        decoration={markings,mark=at position .5 with {\arrow[black,thick]{stealth}}}}
        \centering
        \begin{tikzpicture}
        \node (sij) at (0,0)[sqvertex,label=left:$s_{i,j}$]{};
        \node (f1) at (1,0)[rdvertex]{};
        \node (f1-1) at (0.8,-0.2){{\footnotesize $2$}};
        \node (f1-2) at (1.2,-0.2){{\footnotesize $1$}};
        \node (f2) at (2,0)[opvertex]{};
        \node (f2-1) at (1.8,-0.2){{\footnotesize $2$}};
        \node (f2-2) at (2.2,0.2){{\footnotesize $1$}};
        \node (f2-3) at (2.2,-0.2){{\footnotesize $3$}};
        \node (f2p) at (2,-1)[rdvertex]{};
        \node (vdots1) at (2,-2){$\vdots$};
        \node (f3) at (3,0)[rdvertex]{};
        \node (f3-1) at (2.8,-0.2){{\footnotesize $2$}};
        \node (f3-2) at (3.2,-0.2){{\footnotesize $1$}};
        \node (f4) at (4,0)[opvertex,label=above:$a_{\{x,y\}}$]{};
        \node (f4-1) at (3.8,-0.2){{\footnotesize $2$}};
        \node (f4-2) at (4.5,0.2){{\footnotesize $1$}};
        \node (f4-3) at (4.2,-0.2){{\footnotesize $3$}};
        \node (f5) at (4,-1)[rdvertex,label=right:$\hat{b}_{x\rightarrow y}$]{};
        \node (f5-1) at (3.8,-0.8){{\footnotesize $1$}};
        \node (f5-2) at (3.8,-1.2){{\footnotesize $2$}};
        \node (f6) at (4,-2)[opvertex,label=left:$b_{x\rightarrow y}$]{};
        \node (f6-1) at (3.8,-1.8){{\footnotesize $1$}};
        \node (f6-2) at (3.8,-2.2){{\footnotesize $3$}};
        \node (f6-3) at (4.2,-1.8){{\footnotesize $2$}};
        \node (f7) at (5,-2)[rdvertex,label=right:$\hat{b}_x^h$]{};
        \node (f8) at (4,-3)[rdvertex,label=right:$\hat{b}_{y \rightarrow x}$]{};
        \node (f8-1) at (3.8,-2.8){{\footnotesize $1$}};
        \node (f8-2) at (3.8,-3.2){{\footnotesize $2$}};
        \node (f9) at (4,-4)[opvertex,label=left:$b_{y\rightarrow x}$]{};
        \node (f9-1) at (3.8,-3.8){{\footnotesize $1$}};
        \node (f9-2) at (3.8,-4.2){{\footnotesize $3$}};
        \node (f9-3) at (4.2,-4.2){{\footnotesize $2$}};
        \node (f10) at (5,-4)[rdvertex,label=right:$\hat{b}_y^{\ell}$]{};
        \node (f11) at (4,-5)[rdvertex,label=right:$\widetilde{b}_{y\rightarrow x}$]{};
        \node (hdots) at (5,0){$\hdots$};
        \node (f12) at (6,0)[rdvertex]{};
        \node (f12-1) at (5.8,-0.2){{\footnotesize $2$}};
        \node (f12-2) at (6.2,-0.2){{\footnotesize $1$}};
        \node (f13) at (7,0)[opvertex]{};
        \node (f13-1) at (6.8,-0.2){{\footnotesize $2$}};
        \node (f13-2) at (7.2,0.2){{\footnotesize $1$}};
        \node (f13-3) at (7.2,-0.2){{\footnotesize $3$}};
        \node (f14) at (7,-1)[rdvertex]{};
        \node (vdots2) at (7,-2){$\vdots$};
        \node (f15) at (8,0)[rdvertex]{};
        \node (f15-1) at (7.8,-0.2){{\footnotesize $2$}};
        \node (f15-2) at (8.2,-0.2){{\footnotesize $1$}};
        \node (tij) at (9,0)[sqvertex,label=right:$t_{i,j}$]{};
        \draw (sij)--(f1);
        \draw[double](f1)--(f2);
        \draw (f2)--(f3);
        \draw (f2)--(f2p);
        \draw[double] (f2p)--(vdots1);
        \draw[double] (f3)--(f4);
        \draw (f4)--(f5);
        \draw[double] (f5)--(f6);
        \draw[dashed] (f6)--(f7);
        \draw (f6)--(f8);
        \draw[double] (f8)--(f9);
        \draw[dashed] (f9)--(f10);
        \draw (f9)--(f11);
        \draw (f4)--(hdots);
        \draw (hdots)--(f12);
        \draw[double] (f12)--(f13);
        \draw (f13)--(f14);
        \draw[double] (f14)--(vdots2);
        \draw (f13)--(f15);
        \draw[double](f15)--(tij);
        \end{tikzpicture}
    \caption{Edge selecting gadget $G_{i,j}$ in the reduction of Theorem~\ref{thm:smc-minblock-mainhardness}.\label{fig:2b}}
\end{figure}
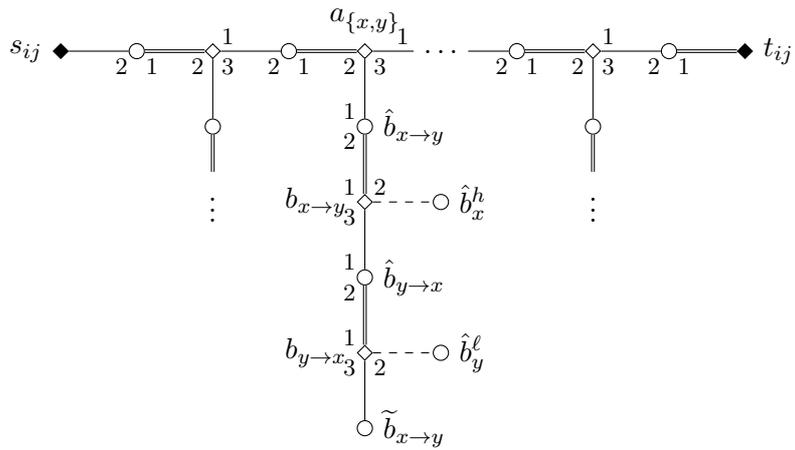
We set the number of blocking pairs allowed for $I$ to be $b=2k+\binom{k}{2}$. 
Together with the instance~$I$, we will define a stable (but not feasible) matching~$M_s$ for $I$ as well. If a woman $w$ of~$I$ is matched by $M_s$, we will denote the man $M_s(w)$ by~$\hat{w}$. 
Some women will need ``dummy'' partners in their preference lists: 
we denote the dummy of $w$ by~$\widetilde{w}$.
The dummy $\widetilde{w}$ will always appear as the last item on $w$'s preference list, and its preference list will always be $L(\widetilde{w})=(w)$. 

For each $i$ and $j$ with $1 \leq i <j \leq k$, we construct an \emph{edge selecting gadget}~$G_{i,j}$ that involves women~$s_{i,j}$ and~$t_{i,j}$, together with women $a_{\{x,y\}}$, $b_{x \to y}$, and $b_{y \to x}$ for each edge $\{x,y\} \in E_{i,j}$.
All women in $G_{i,j}$ are matched by~$M_s$ except for $s_{i,j}$, and~$G_{i,j}$ contains the man~$\hat{w}$ for each of these women $w$, together with additional dummies~$\widetilde{b}_{y \to x}$ for each $\{x,y\} \in E_{i,j}$ with $x$ preceding~$y$. 

For each $i \in \{1, \dots, k\}$, we also construct a \emph{node selecting gadget} $G_i$ involving women $s_i,t_i$, and $u_i^1, \dots, u_i^{b+1}$, together with women $a_x$, $b_x^1, \dots, b_x^{d_G(x)}$, and $c_x^1, \dots, c_x^{b+1}$ for each $x \in V_i$.   
The men in~$G_i$ include $\hat{w}$ for each woman~$w$ of $G_i$ except for $s_i$, and additional dummies $\widetilde{b}_x^1, \dots, \widetilde{b}_x^{d_G(x)}$ and~$\widetilde{c}^{b+1}_x$ for each $x \in V_i$.

We define the following sets of women: 
\begin{equation*}
  \begin{array}{l@{}l@{}ll@{}l@{}l}
    A  & =\ &\{a_x \mid x \in V(G)\} &
    C  & = &\{c_x^h \mid x \in V(G), 1 \leq h \leq b+1\} \\
    A' & = &\{ a_{\{x,y\}} \mid \{x,y\} \in E(G)\} &
    S  & = &\{s_i \mid 1 \leq i \leq k\} \cup \{s_{i,j} \mid 1 \leq i <j \leq k\} \\
    B  & = &\{b_x^h \mid x \in V(G), 1 \leq h \leq d_G(x)\} &
    T  & = &\{t_i \mid 1 \leq i \leq k\} \cup \{t_{i,j} \mid 1 \leq i <j \leq k\} \\
    B' & = &\{ b_{x \to y},b_{y\to x} \mid \{x,y\} \in E(G)\} \qquad &
    U  & = &\{u_i^h \mid 1 \leq i \leq k, 1 \leq h \leq b+1\}
  \end{array}
\end{equation*}
To define the set $\mathcal{W}^{\star}$ of women in $I$ with covering constraint we let $\mathcal{W}^{\star} = S \cup T \cup U$; note $|\mathcal{W}^{\star}|=2 \binom{k}{2} + 2k + k(\binom{k}{2}+2k+1)$.
To finish the definition of~$I$, we define the precise structure of these gadgets as well as the connections between them by the preference lists shown in Tables~\ref{table-prefs-node} and~\ref{table-prefs-edge}; when not stated otherwise, indices take all possible values.
For simplicity, we write $b_x^0=a_x$, $b_x^{d_G(x)+1}=c_x^1$, and $c_x^0=b_x^{d_G(x)}$ 
for any vertex $x \in V(G)$. 

\begin{table}[thpb]
\renewcommand{\arraystretch}{1.3}
        \begin{tabular}{lll} 
            $L(a_x)$ & $= (\hat{a}_{\suc(x)},\hat{a}_x,\hat{b}_x^1),$ & \textrm{ where $x \in V_i \setminus \{ v_i^{\infty} \}$,} \\ 
            $L(a_x)$ & $= (\hat{t}_i,\hat{a}_x,\hat{b}_x^1),$ & \textrm{ where $x=v_i^{\infty}$,} \\ 
            $L(b_x^h)$&$= (\hat{b}_x^h,\hat{b}_x^{h+1},\widetilde{b}_x^h),$ & \textrm{ where $1 \leq h \leq d_G(x)$,} \\
            $L(c_x^h)$&$= (\hat{c}_x^h,\hat{c}_{\suc(x)}^h,\hat{c}_x^{h+1}),$ & 
            \textrm{ where $1 \leq h \leq b$, $x \in V_i \setminus \{ v_i^{\infty} \}$,} \\
            $L(c_x^h)$&$= (\hat{c}_x^h,\hat{u}_i^h,\hat{c}_x^{h+1},),$ & 
            \textrm{ where $1 \leq h \leq b$ and $x=v_i^{\infty}$,} \\  
            $L(c_x^{b+1})$&$= (\hat{c}_x^{b+1},\hat{c}_{\suc(x)}^h,\widetilde{c}_x^{b+1}),$ & 
            \textrm{ where $x \in V_i \setminus \{ v_i^{\infty} \}$,} \\
            $L(c_x^{b+1})$&$= (\hat{c}_x^{b+1},\hat{u}_i^{b+1},\widetilde{c}_x^{b+1},),$ & 
            \textrm{ where $x=v_i^{\infty}$,} \\
            $L(s_i)$&$= (\hat{a}_x),$ & \textrm{ where $x=v_i^1$,} \\ 
            $L(t_i)$&$= (\hat{t}_i),$ & \\ 
            $L(u_i^h)$&$= (\hat{u}_i^h),$ & \\ 
            $L(\hat{a}_x)$&$= (a_{x},a_{\pre(x)}),$ & 
            \textrm{ where $x \in V_i \setminus \{v_i^1\}$,} \\ 
            $L(\hat{a}_x)$&$= (a_x,s_i),$ & \textrm{ where $x=v_i^1$,} \\ 
            $L(\hat{b}_x^h)$&$= (b_x^{h-1},b_{x \to y},b_x^h),$ & 
            \textrm{ where $y=n(x,h)$, $x \in V_i$, $y \in V_j$ and $i<j$} \\
            $L(\hat{b}_x^h)$&$= (b_x^{h-1},b_{y \to x},b_x^h),$ & \textrm{ where $y=n(x,h)$, $x \in V_i$, $y \in V_j$ and $i>j$} \\
            $L(\hat{c}_x^h)$&$= (c_x^{h-1},c_{\pre(x)}^h,c_x^h),$ & \textrm{ where $x \in V_i \setminus \{v_i^1\}$,} \\
            $L(\hat{c}_x^h)$&$= (c_x^{h-1},c_x^h),$ & \textrm{ where $x=v_i^1$,} \\
            $L(\hat{t}_i)$ & $= (t_i,a_x),$ & \textrm{ where $x=v_i^{\infty}$,} \\
            $L(\hat{u}_i^h)$   & $= (c_x^h,u_i^h),$ & \textrm{ where $x=v_i^{\infty}$,} \\
            $L(\widetilde{w})$  &$= (w),$ & \textrm{ for any dummy woman~$\widetilde{w}$.} \\  
        \end{tabular}
  \caption{Preference lists of women and men in node selecting gadgets.\label{table-prefs-node}}
\end{table}

\begin{table}[thpb]
        \renewcommand{\arraystretch}{1.3}
        \begin{tabular}{lll} 
            $L(a_{\{x,y\}})$&$= (\hat{a}_{\suc(\{x,y\})},\hat{a}_{\{x,y\}},\hat{b}_{x \to y}),$ & 
            \textrm{ where $\{x,y\} \in E_{i,j} \setminus \{e_{i,j}^{\infty}, \}$ and $x$ precedes $y$,} \\ 
            $L(a_{\{x,y\}})$&$= (\hat{t}_{i,j},\hat{a}_{\{x,y\}},\hat{b}_{x \to y}),$ & 
            \textrm{ where $\{x,y\}=e_{i,j}^{\infty}$ and $x$ precedes $y$,} \\ 
            $L(b_{x \to y})$&$= (\hat{b}_{x \to y},\hat{b}_x^h,\hat{b}_{y \to x}),$ & 
            \textrm{ where $y=n(x,h)$ and $x$ precedes $y$ in $V(G)$,} \\
            $L(b_{y \to x})$&$= (\hat{b}_{y \to x},\hat{b}_y^h,\widetilde{b}_{y \to x}),$ & 
            \textrm{ where $y=n(x,h)$ and $x$ precedes $y$ in $V(G)$,} \\
            $L(s_{i,j})$&$= (\hat{a}_{\{x,y\}}),$ & 
            \textrm{ where $\{x,y\}=e_{i,j}^1$,} \\ 
            $L(t_{i,j})$&$= (\hat{t}_{i,j}),$ & \\ 
            $L(\hat{a}_{\{x,y\}})$&$= (a_{\{x,y\}},a_{\pre(\{x,y\})}),$ & 
            \textrm{ where $\{x,y\} \in E_{i,j} \setminus \{ e_{i,j}^1\}$,} \\ 
            $L(\hat{a}_{\{x,y\}})$&$= (a_{\{x,y\}},s_{i,j}),$ & 
            \textrm{ where $\{x,y\}=e_{i,j}^1$,} \\ 
            $L(\hat{b}_{x \to y})$&$= (a_{\{x,y\}},b_{x \to y}),$ & 
            \textrm{ where $x$ precedes $y$ in $V(G)$,} \\ 
            $L(\hat{b}_{y \to x})$&$= (b_{x \to y},b_{y \to x}),$ & 
            \textrm{ where $x$ precedes $y$ in $V(G)$,} \\   
            $L(\hat{t}_{i,j})$&$= (t_{i,j},a_{\{x,y\}}),$ & 
            \textrm{ where $\{x,y\} = e_{i,j}^{\infty}$,} \\
            $L(\widetilde{w})$&$= (w),$ & 
            \textrm{ for any dummy woman~$\widetilde{w}$.} \\
        \end{tabular}
  \caption{Preference lists of women and men in edge selecting gadgets.\label{table-prefs-edge}}
\end{table} 

Let us define a master list~$L_{\mathcal{W}}$ over all women as follows.
The first women in~$L_{\mathcal{W}}$ are those in~$T$, in any ordering.
They are followed by women in $A$, ordered according to the reversed ordering over~$V(G)$, that is,~$a_x$ precedes $a_y$ 
exactly if $y$ precedes $x$. 
Next follow women of~$A'$, ordered  according to the reversed ordering over $E(G)$. 
Next come women in~$B \cup B'$.
To order them, we first order those in $B$ by putting~$b_x^h$ before $b_y^{\ell}$ in~$L_{\mathcal{W}}$ if and only if $x$ precedes~$y$ or $x=y$ and $h<\ell$, then for each edge $\{x,y\} \in E(G)$ with $x$ preceding~$y$, $y=n(x,h)$ and $x=n(y,\ell)$ we insert $b_{x \to y}$ just before $b_x^h$, and we insert $b_{y \to x}$ just before $b_y^{\ell}$, thus determining the ordering of $B \cup B'$. 
After women in $B \cup B'$ come women of $C$, with~$c_x^h$ preceding~$c_y^{\ell}$ exactly if $h <\ell$ or $h=\ell$ and $x$ precedes~$y$. 
We finish the definition of the master list $L_{\mathcal{W}}$ by putting all women in $S \cup U$ at the end of $L_{\mathcal{W}}$ in an arbitrary order. 

The master list over men is derived from $L_{\mathcal{W}}$ by letting $\hat{w_1}$ precede $\hat{w_2}$ whenever~$w_1$ precedes~$w_2$ in~$L_{\mathcal{W}}$, and adding all dummies at the end in an arbitrary order. 
It is easy to check that the preference lists given in Tables~\ref{table-prefs-node} and~\ref{table-prefs-edge} are indeed compatible with these master lists. 
This completes the construction of the instance.

\medskip 

We are going to prove that the constructed instance $I$ admits a feasible assignment with at most~$b$ blocking pairs if and only if there is a clique of size~$k$ in the graph $G$.

''$\Rightarrow$'':
Suppose there is a feasible matching $M$ of men to women with at most~$b$ blocking pairs.   
Let~$G_{\Delta}$ be the symmetric difference $M \triangle M_s$.
Notice that for each woman $s \in S$, the difference $G_{\Delta}$ must contain exactly one maximal path containing $s$ as its endpoint, since the women in~$S$ must be matched in~$M$, but are unmatched in~$M_s$.
Similarly, no path of $G_{\Delta}$ can contain a woman in $T \cup U$, 
because these women are matched by $M_s$ to their only possible partners, 
and they must be matched by $M$ as well, since $T \cup U$ is contained in~$\mathcal W^{\star}$.
We call a maximal path~$P$ in $G_{\Delta}$ with an endpoint $s$ in $S$ an \emph{improving path}. 
We say that $P$ \emph{starts} at $s$ and \emph{ends} at its other endpoint, and we refer to the path starting at $s_i$ (or~$s_{i,j}$) as~$P_i$ (or $P_{i,j}$, respectively). 

We define the \emph{cost} of some path $P$ of $G_{\Delta}$ as the number of blocking pairs $\{m,w\}$ for $M$ involving a  woman~$w$ that appears on $P$. 
By Proposition~\ref{prop:alt-paths-blocking-pairs}, each improving  path contains at least one edge that is blocking for $M$, because no edge can block $M_s$.
Therefore, each path in $G_{\Delta}$ has cost at least $1$.

As there are exactly $k+ \binom{k}{2}$ improving  paths (as all women in $S$ must be matched by $M$), we get a minimum cost of $k+ \binom{k}{2}$. 
Note also that the total cost of all paths in $G_{\Delta}$ cannot exceed $b=2k+ \binom{k}{2}$. 
Claim~\ref{claim-path-costs} is therefore crucial. 
  \begin{claimu}
  \label{claim-path-costs}
  The following holds for any improving  path $P$ of $G_{\Delta}$:
    \begin{itemize}
      \item[(a)] $P$ cannot end at a dummy $\widetilde{c}_x^{b+1}$ for some $x \in V(G)$.
      \item[(b)] $P$ contains an edge $\{a, \hat{a}\}$ for some $a \in A \cup A'$ that blocks $M$.
      \item[(c)] If $P$ is not disjoint from $G_i$ for some $i$, then $P$ has cost at least $2$.
    \end{itemize}
  \end{claimu}
  \begin{proof}[Proof of Claim~\ref{claim-path-costs}.]
    To prove (a), suppose for contradiction that $P$ ends at~$\widetilde{c}_x^{b+1}$, where $x \in V_i$.
    Clearly, $P$ must contain at least one woman from each of the $b+1$ sets $\{c_v^h \mid v \in V_i\}$, $h=1, \dots, b+1$.
    Fix $h$, and let us consider the last $v \in V_i$ for which $c_v^h$ is incident to an edge of $G_{\Delta}$.
    Let $w=c_{\suc(v)}^h$ if $v \neq v_i^{\infty}$,  or otherwise let $w=u_i^h$.  
    Then the edge $\{c_v^h, \hat{w}\}$ yields a blocking pair in $M$, as $M(w)=M_s(w)=\hat{w}$, and thus~$\hat{w}$ prefers~$c_v^h$ to $w$.
    This reasoning gives us $b+1$ different blocking pairs for~$M$, one for each index~$h$, contradicting our assumption on $M$. 

    To prove (b), let us consider the case when $P=P_i$ for some $i$; the argument goes the same way for the case where $P=P_{i,j}$ for some $i$ and $j$.
    If $P$ ends at~$a_x$ for some $x \in V_i$, then~$a_x$ forms a blocking pair with $\hat{a}_x$ in $M$.
    If $P$ does not end at a woman in $A$, then it must contain the edge $\{a_x, b_x^1\}$ for some $x$, in which case $\{a_x, \hat{a}_x\}$ is again blocking in $M$, showing~(b).

    To see (c), first observe that if $P$ is not disjoint from $G_i$, then $P$ ends in~$G_i$, simply because of its property that it contains edges from $M$ and $M_s$ in an alternating fashion.
    Therefore, the last woman~$w$ on $P$ must be in $B \cup C$.
    If $w=b_x^h$ for some $b \in B$, then the edge $\{b_x^h, \hat{b}_x^{h+1}\}$ is blocking $M$, as $b_x^h$ cannot get its first choice~$\hat{b}_x^h$ in $M$ (and $\hat{b}_x^{h+1}$ cannot be on $P$, as that would imply that $b_x^{h+1}$ is on~$P$, contradicting the choice of~$w$). 
    If, by contrast,  $w=c$ for some $c \in C$, then~$P$ must end at~$w$ by (a), and then $c$ forms a blocking pair with the third man in its preference list (for whom $c$ is the first choice).
    In either case, $w$ is involved in a blocking pair, which together with the blocking pair guaranteed by (b) implies that $P$ has cost at least~$2$.

    This completes the proof of Claim~\ref{claim-path-costs}.
  \renewcommand{\qedsymbol}{$\Diamond$}
  \end{proof}

  Claim~\ref{claim-path-costs} proves that for each $i \in \{1,\hdots, k\}$ the improving  path $P_i$ has cost at least~$2$.
  Since all the remaining $\binom{k}{2}$ improving  paths have cost at least $1$, and the total cost of these paths must be at most $b=2k+\binom{k}{2}$, we get that any path~$P_i$ (or $P_{i,j}$) must have cost \emph{exactly}~$2$ (or~$1$, respectively).
  Furthermore, it also follows that no other path of $G_{\Delta}$ can enter or start in~$G_i$, for any $i$, as that would imply that the number of blocking pairs for $M$ is more than $b$.
  In addition, it is not hard to see that $G_{\Delta}$ does not contain any cycle, because all cycles in the graph underlying $I$ contain two consecutive edges not in $M_s$. 
  Hence, it follows that the only connected component in $G_{\Delta}$ that is not disjoint from $G_i$ is $P_i$. 

  To deal with the possible courses the path $P_i$ may take in the graph for some $i\in\{1,\hdots,k\}$, let $x_i$ denote the vertex in $V_i$ for which $\{a_{x_i}, \hat{a}_{x_i}\}$ is the blocking edge guaranteed by statement (b) of Claim~\ref{claim-path-costs}.
  Observe that $P_i$ either ends at~$a_{x_i}$ or contains the edge $\{a_{x_i}, \hat{b}_{x_i}^1\}$.
  In either case, we say that~$P_i$ \emph{selects}~$x_i$ from $V_i$; clearly, there can be only one vertex in $V_i$ selected by $P_i$. 

  Consider now $P_{i,j}$ for some $1 \leq i<j \leq k$.
  Recall that $P_{i,j}$ has cost $1$.
  Therefore, statement~(b) of Claim~\ref{claim-path-costs} proves that the only blocking edge incident to some woman on~$P_{i,j}$ must be $\{a_{\{x,y\}}, \hat{a}_{\{x,y\}}\}$ for some $\{x,y\} \in E_{i,j}$.
  We say that $P_{i,j}$ \emph{selects} the edge $\{x,y\}$; without loss of generality, let us assume that~$x$ precedes~$y$. 
  By statement (c) of Claim~\ref{claim-path-costs}, we also know that~$P_{i,j}$ cannot leave~$G_{i,j}$, which means that it can only have cost $1$ if it ends at $\widetilde{b}_{y \to x}$. 
  In particular, it contains the edges $\{b_{x \to y}, \hat{b}_{y \to x}\}$ and $\{b_{y \to x}, \widetilde{b}_{y \to x}\}$.
  Observe that the edge $\{b_{x \to y}, \hat{b}_x^h\}$ where~$h$ is such that $y=n(x,h)$ cannot be blocking in~$M$ (as this would indicate a cost of $2$ for $P_{i,j}$), yielding that $\hat{b}_x^h$ must be matched to $b_x^{h-1}$ in $M$. 
  By the arguments of the previous paragraph, this means that~$P_i$ must contain the subpath $(a_x, \hat{b}_x^1, b_x^1, \dots, \hat{b}_x^h, b_x^h)$.
  Hence, we obtain that $x$ must be selected by $P_i$.
  Similarly, from the fact that the edge $\{b_{y \to x}, \hat{b}_y^{\ell}\}$ where $x=n(y,\ell)$ is not blocking in $M$ we get that $y$ must be selected by $P_j$. 

  Thus, we obtain that if an edge is selected by $P_{i,j}$ for some $i$ and $j$, then its endpoints must be selected by $P_i$ and $P_j$. As this must hold for each pair of indices with $1 \leq i<j \leq k$, we obtain that there must be $\binom{k}{2}$ edges in $G$ whose endpoints are among the $k$ selected vertices.
  This can only happen if these edges are the edges of a clique of size $k$.

  ''$\Leftarrow$'': 
  Suppose now that $G$ has a clique of size $k$ formed by the vertices $x_1,\hdots,x_k$, with $x_i \in V_i$ for each $i \in \{1,\dots,k\}$.
  Instead of directly defining the required matching $M$ that is feasible and admits at most $b$ blocking pairs, we give~$M_s \triangle M$ as the union of paths $P_{i}$ for $i \in \{1,\dots,k\}$, and paths $P_{i,j}$ for $1 \leq i<j \leq k$, defined as follows. 

  We set $P_i$ as the path
  \begin{equation*}
    P_i=(s_i, \hat{a}_{v_i^1}, a_{v_i^1}, \dots, \hat{a}_{x_i}, a_{x_i}, \hat{b}_{x_i}^1, b_{x_i}^1, \dots, \hat{b}_{x_i}^{d_G(x_i)}, b_{x_i}^{d_G(x_i)}, \widetilde{b}_{x_i}^{d_G(x_i)}) \enspace .
  \end{equation*}
  Similarly, we define
  \begin{equation*}
    P_{i,j}=(s_{i,j}, \hat{a}_{e_{i,j}^1}, a_{e_{i,j}^1}, \dots, \hat{a}_{\{x_i, x_j\}}, a_{\{x_i, x_j\}}, 
    \hat{b}_{x_i \to x_j}, b_{x_i \to x_j}, \hat{b}_{x_j \to x_i}, b_{x_j \to x_i}, 
    \widetilde{b}_{x_j \to x_i}) \enspace .
  \end{equation*}
  It is straightforward to verify that the blocking pairs for $M$ are then the~$k$ edges $\{a_{x_i}, \hat{a}_{x_i}\}$, $i \in \{1, \dots, k\}$,  the $k$ edges $\{b_{x_i}^{d_G(x_i)}, c_{x_i}^1\}$, and the $\binom{k}{2}$ edges $\{a_{\{x_i, x_j\}}, \hat{a}_{\{x_i, x_j\}}\}$, $1 \leq i <j \leq k$.
  The feasibility of $M$ is trivial; this completes the proof of Theorem~\ref{thm:smc-minblock-mainhardness}.
\end{proof}

A fundamental hypothesis about the complexity of $\mathsf{NP}$-hard problems is the \emph{Exponential Time Hypothesis} (ETH), which stipulates that algorithms solving all {\sc Satisfiability} instances in subexponential time cannot exist~\cite{ImpagliazzoEtAl2001}.
Assuming ETH, the fundamental {\sc Clique} problem parameterized by solution size~$k$ was shown not to admit any algorithm giving the correct answer in time $f(k)\cdot n^{o(k)}$ for all $n$-vertex instances and any computable function~$f$~\cite[Thm. 5.4]{ChenEtAl2004}.
The known reduction from {\sc Clique} to {\sc Multicolored Clique} does not change the parameter~\cite{FellowsEtAl2009}.
Finally, in the proof of Theorem~\ref{thm:smc-minblock-mainhardness}, an instance of {\sc Multicolored Clique} with solution size $k$ is reduced to an instance of SMC-1 with parameter $b = O(k^2)$.
\begin{corollary}
  Assuming ETH, SMC-1 cannot be solved in time $f'(b) \cdot n^{o(\sqrt{b})}$ for any com\-pu\-table function~$f'$, even if there is a master list over men and over women, all preference lists have length at most~$3$, and each woman in~$\mathcal W^\star$ finds only a single man acceptable.
\end{corollary}

\section{Polynomial-Time Approximation}
\label{sec:hrlq-min-blocking-pairs-shortlists}
Here we first provide a polynomial-time algorithm that yields an approximation for HRLQ with factor $(\Delta_{\mathcal R}-1) \mathsf{\underline{q}}_{\Sigma}$, where $\Delta_{\mathcal R}$ is the maximum length~$\Delta_{\mathcal R}$ of residents' preference lists and $\mathsf{\underline{q}}_{\Sigma}$ is the total sum of all lower quotas.
Then we use this result to propose an exact poly\-nomial-time algorithm for {\sc HRLQ} for the case where both~$\Delta_{\mathcal R}$ and $\mathsf{\underline{q}}_{\Sigma}$ are constant.
Recall that in {\sc HRLQ}, our objective is to find an assignment that satisfies all quota lower and upper bounds and minimizes the number of blocking pairs.

\begin{theorem}
\label{thm-poly-approx}
  Let $I$ be an instance of HRLQ.
  Let~$\Delta_{\mathcal R}$ denote the maximum length of residents' preference lists, and let $\mathsf{\underline{q}}_{\Sigma}$ denote the sum of lower quota bounds taken over all hospitals in $I$.
  There is an algorithm that in polynomial time either outputs a feasible assignment for~$I$ with at most $(\Delta_{\mathcal R}-1) \mathsf{\underline{q}}_{\Sigma}$ blocking pairs, involving only $\mathsf{\underline{q}}_{\Sigma}$ residents, or concludes that no feasible assignment exists.
\end{theorem}
\begin{proof}
  Let $\mathcal H^\star$ denote the set of hospitals with positive lower quotas. 
  We start by finding an assignment~$M_q$ that assigns $\mathsf{\underline{q}}(h)$ residents to each hospital $h \in \mathcal H^\star$, and has the following property:
  \begin{equation}
  \label{eqn:propdagger}
    \tag{$\dagger$}
    \parbox{0.9\textwidth}{
      \emph{for each hospital $h \in \mathcal H^\star$, all residents that are
      not in~$M_q(h)$ but preferred by~$h$ to the least preferred resident in~$M_q(h)$~are contained in~$\bigcup_{h' \in \mathcal H^\star\setminus\{h\}} M_q(h')$.}
    }
  \end{equation}
  Such an assignment can be obtained as follows.
  We start from an arbitrary assignment $M$ that assigns $\mathsf{\underline{q}}(h)$ residents to each $h \in \mathcal H^\star$ (if no such assignment exists, then we can stop and reject); such an assignment, if existent, can be found in polynomial time by an algorithm of Hopcroft and Karp~\cite{HopcroftKarp1973}.  
  Then we greedily re-assign residents to hospitals of $\mathcal H^\star$, one-by-one: at each step, we take a hospital $h \in \mathcal H^\star$, and if there exists a resident $r$ not assigned to any other hospital in~$\mathcal H^\star$ that~$h$ prefers to the least preferred resident~$r'$ in $M(h)$, then we replace~$r'$ with~$r$ in $M(h)$.
  If this step cannot be applied anymore, then we arrive at an assignment~$M_q$ with the desired property~\eqref{eqn:propdagger}.

  Given $M_q$, we reduce the upper quotas of each hospital $h \in \mathcal H^\star$ by $\mathsf{\underline{q}}(h)$, set all lower quotas to~$0$, and delete all residents in $\mathcal R^\star := M_q(\mathcal H^\star)$.
  We then find a stable assignment~$M_s$ in the resulting instance~$I'$; note that  $I'$ is an instance of \textsc{HR}, so we can find~$M_s$ in polynomial time~\cite{GaleShapley1962}.
  Finally, we output $M^{\textnormal{out}}=M_s \cup M_q$.
  Clearly, $M^{\textnormal{out}}$ is feasible.
  Also, any blocking pair that~$M^{\textnormal{out}}$ admits must involve either a hospital from $\mathcal H^\star$ or a resident from $\mathcal R^\star = M_q(\mathcal H^\star)$ by the stability of~$M_s$ with respect to $I'$.
  Observe that if some $h \in \mathcal H^\star$ is involved in some blocking pair $\{r,h\}$ of $M^{\textnormal{out}}$, then we must have $r \in \mathcal R^\star$.
  To see this, recall that each resident that is preferred by $h$ to its least preferred resident in $M_q(h)$ must be in~$\mathcal R^\star$ because of property~\eqref{eqn:propdagger}, and furthermore,~$h$ is under-subscribed in~$M^{\textnormal{out}}$ (within~$I$) if and only if $h$ is under-subscribed in~$M_s$ (within~$I'$).
  Therefore, we can conclude that each blocking pair for $M^{\textnormal{out}}$ must involve some resident in $\mathcal R^\star$; observe that $|\mathcal R^\star| \leq \sum_{h\in\mathcal H}\mathsf{\underline{q}}(h) = \mathsf{\underline{q}}_{\Sigma}$.
  Since each resident in $\mathcal R^\star$ is incident to at most $\Delta_{\mathcal R}-1$ edges not in $M^{\textnormal{out}}$, we also have that $M^{\textnormal{out}}$ admits at most $(\Delta_{\mathcal R}-1)|\mathcal R^\star| \leq (\Delta_{\mathcal R}-1)\mathsf{\underline{q}}_{\Sigma}$ blocking pairs.
\end{proof}

If both $\Delta_{\mathcal R}$ and $\mathsf{\underline{q}}_{\Sigma}$ are constant, then Theorem~\ref{thm-poly-approx} implies that {\sc HRLQ} becomes polynomial-time solvable.
Indeed, we can use the following simple strategy, depending on the number $b$ of blocking pairs allowed: if $b \geq (\Delta_{\mathcal R}-1) \mathsf{\underline{q}}_{\Sigma}$, then we apply Theorem~\ref{thm-poly-approx} directly; if $b < (\Delta_{\mathcal R}-1) \mathsf{\underline{q}}_{\Sigma}$, then we use the algorithm by Hamada et al.~\cite{HamadaEtAl2014} running in time $O(|I|^{b+1})$ which is polynomial, since $b$ is upper-bounded by a constant.

\begin{corollary}
\label{cor-hrlq-poly}
  If both the maximum length~$\Delta_{\mathcal R}$ of residents' preference lists and the total sum~$\mathsf{\underline{q}}_{\Sigma}$  of all lower quotas is constant, then {\sc HRLQ} is polynomial-time solvable.
\end{corollary}

Another application of Theorem~\ref{thm-poly-approx} is an approximation algorithm that works regardless of whether~$\Delta_{\mathcal R}$ or $\mathsf{\underline{q}}_{\Sigma}$ is a constant.
In fact, the algorithm of Theorem~\ref{thm-poly-approx} can be turned into a $(\Delta_{\mathcal R}-1) \mathsf{\underline{q}}_{\Sigma}$-factor approximation algorithm as follows. 
First, we find a stable assignment $M_s$ for $I$ in polynomial time using the extension of the Gale-Shapley algorithm for the \textsc{Hospitals/Residents} problem.
If $M_s$ is not feasible, then by the Rural Hospitals Theorem~\cite{GaleSotomayor1985}, we know that any feasible assignment for~$I$ must admit at least one blocking pair; hence, the algorithm presented in Theorem~\ref{thm-poly-approx} clearly yields an approximation with (multiplicative and also additive) factor $(\Delta_{\mathcal R}-1) \mathsf{\underline{q}}_{\Sigma}$. 

To close this section, we also state an analogue of Theorem~\ref{thm-poly-approx} that deals with SMC: it can handle covering constraints on both sides, but assumes that all quota upper bounds are $1$. 
\begin{theorem}
\label{thm-smc2-poly-approx}
  There is an algorithm that in polynomial time either outputs a feasible matching for an instance $I$ of SMC with at most $(\Delta_{\mathcal W}-1)|\mathcal{M}^{\star}|+(\Delta_{\mathcal M}-1)|\mathcal{W}^{\star}|$ blocking pairs, or concludes that $I$ admits no feasible matching.
\end{theorem}
\begin{proof}
 The proof uses the same ideas as those used in our proof for Theorem~\ref{thm-poly-approx}, so the reader may skip the proof below, which we only include for completeness. 

  We start by finding an arbitrary matching $M$ that covers each distinguished person (if no such matching exists, then we can stop and reject); such a matching, if existent, can be found in polynomial time by standard flow techniques.  
  We assume, without loss of generality, that each edge in $M$ is incident to some distinguished person. 
  Let us define $\mathcal{X}^\star= \mathcal{W}^\star \cup \mathcal{M}^\star$, and let $\mathcal{U}^\star$ be the set of those persons $x \in \mathcal X^\star$ whose partner $M(x)$ is also in  $\mathcal{X}^\star$. 

  We proceed by modifying $M$ into a matching $M_q$ that covers $\mathcal{X}^\star$ and has the following property:
  \begin{equation}
  \label{eqn:propstarsmc2}
    \tag{$\maltese$}
    \parbox{0.9\textwidth}{
    \emph{If a person $x \in \mathcal X^\star \setminus \mathcal U^\star$ belongs to 
        a blocking pair $\{x,y\}$ for $M_q$, 
        then $M_q(y) \in \mathcal{X}^\star$.}
    }
  \end{equation}
  Such an assignment can be obtained as follows.
  We greedily assign partners to the men and women in $\mathcal X^\star \setminus \mathcal U^\star$, one-by-one: 
  at each step, we take a person $x \in \mathcal X^\star \setminus \mathcal U^\star$, and if $x$ forms a blocking pair (with respect to the current matching) with some $y$ that is not the partner of a distinguished person, then we replace the partner of $x$ with $y$: we add the edge $\{x,y\}$ to the matching, and delete all the other edges incident to $x$ or $y$.
  Observe that the obtained matching is still feasible.
  If this step cannot be applied anymore, then we arrive at a matching~$M_q$ with the desired property~\eqref{eqn:propstarsmc2}; 
  note also that each edge in~$M_q$ is incident to some distinguished person.

  Given $M_q$, we delete all men and women covered by $M_q$.
  We then find a stable matching~$M_s$ in the resulting instance $I'$; note that $I'$ is an instance of \textsc{Stable Marriage}, so we can find~$M_s$ in polynomial time \cite{GaleShapley1962}.
  Finally, we output $M^{\textnormal{out}}=M_s \cup M_q$.
  Clearly, $M^{\textnormal{out}}$ is feasible.
  Also, any blocking pair that~$M^{\textnormal{out}}$ admits must involve a person covered by $M_q$ due to the stability of~$M_s$ with respect to $I'$.  

  We claim that any blocking pair $\{x,y\}$ involves a person whose partner by~$M_q$ is distinguished, so either $M_q(x) \in \mathcal{X}^\star$ or $M_q(y) \in \mathcal{X}^\star$. 
  We can assume that $x$ is covered by $M_q$ (because this holds for at least one of $x$ and $y$). 
  To see the claim, first note that if $x$ is not distinguished, then $M_q(x)$ must be distinguished, because each edge of $M_q$ contains a distinguished person.
  Second, if $x \in \mathcal{X}^\star$, then either $x \in \mathcal{U}^\star$ (in which case $M_q(x) \in \mathcal{X}^\star$) or \mbox{$M_q(y) \in \mathcal{X}^\star$} because of property~\eqref{eqn:propstarsmc2}.
  Therefore, we can conclude that each blocking pair for~$M^{\textnormal{out}}$ must involve the partner of some distinguished resident.
  The partners of distinguished women can be incident to at most $|\mathcal{W}^\star|(\Delta_{\mathcal M}-1)$ blocking pairs, and similarly, the partners of distinguished men can be incident to at most $|\mathcal{M}^\star|(\Delta_{\mathcal W}-1)$ blocking pairs, proving the theorem.
\end{proof}

\section{SMC with Bounded Number of Distinguished Persons or Blocking Pairs}
\label{sec:boundednumberofwomentobecovered}
In Theorem~\ref{thm:smc-minblock-mainhardness} we proved $\mathsf{W}[1]$-hardness of SMC-1 for the case where 
$\Delta_{\mathcal M}=\Delta_{\mathcal W}=3$, with parameter $b+|\mathcal{W}^{\star}|$.
Here we investigate those instances of SMC and SMC-1 where the length of preference lists may be unbounded, 
but either~$b$, or the number of distinguished persons is constant. 

First, if the number $b$ of blocking pairs allowed is constant, then SMC can be solved by simply running the extended Gale-Shapley algorithm after guessing and deleting all blocking pairs. 
This complements the result by Hamada et al.~\cite{HamadaEtAl2014}.
\begin{observation}
\label{observ:bounded-b}
  SMC can be solved in time $O(|I|^{b+1})$, where $b$ denotes the number of blocking pairs allowed in the input instance $I$. 
\end{observation}


In Theorem~\ref{thm:1woman-strict} we prove hardness of SMC-1 even if only one woman must be covered.
If we require preferences to follow master lists, then a slightly weaker version of Theorem~\ref{thm:1woman-strict}, where $|\mathcal{W}^{\star}|=2$, still holds.
\begin{theorem}
\label{thm:1woman-strict}
  SMC-1 is $\mathsf{W}[1]$-hard parameterized by $b+\Delta_{\mathcal{M}}$, even~if $\mathcal{W}^{\star} = \{s\}$, $ \Delta_{\mathcal{W}} = 3$, and $|L(s)|=1$. 
\end{theorem}
\begin{proof}
  We present a reduction based on the one from~{\sc Multicolored Clique} given in the proof of Theorem~\ref{thm:smc-minblock-mainhardness}.
  Given some graph $G$ and an integer $k$ as inputs, we are going to re-use the instance~$I$ constructed in the proof of Theorem~\ref{thm:smc-minblock-mainhardness}. 
  Recall that $I$ has a feasible matching with at most $b=\binom{k}{2} +2k$ blocking pairs if and only if $G$ has a clique of size $k$.
  Recall also that the set of women that must be covered in $I$ is $S \cup T \cup U$; here we denote this set by $\mathcal{W}_I^{\star}$.
  \tikzset{rdvertex/.style={minimum size=2mm,circle,fill=white,draw, inner sep=0pt},
    sqvertex/.style={minimum size=2mm,diamond,fill=black,draw, inner sep=0pt},
    opvertex/.style={minimum size=2mm,diamond,fill=white,draw, inner sep=0pt},
    decoration={markings,mark=at position .5 with {\arrow[black,thick]{stealth}}}}
  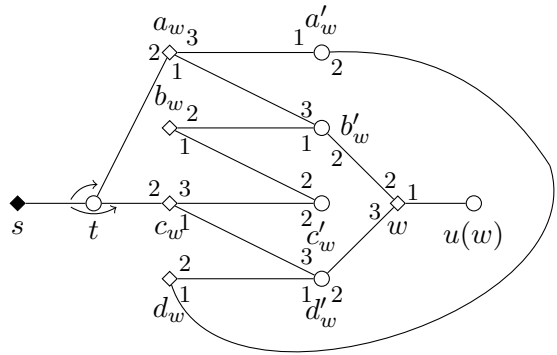
\begin{figure}[thpb]
     \tikzset{rdvertex/.style={minimum size=2mm,circle,fill=white,draw, inner sep=0pt},
        sqvertex/.style={minimum size=2mm,diamond,fill=black,draw, inner sep=0pt},
        opvertex/.style={minimum size=2mm,diamond,fill=white,draw, inner sep=0pt},
        decoration={markings,mark=at position .5 with {\arrow[black,thick]{stealth}}}}
        \centering
        \begin{tikzpicture}
        \node (s) at (0,0)[sqvertex,label=below:$s$]{};
        \node (t) at (1,0)[rdvertex,label=below:$t$]{};
        \node (aw) at (2,2)[opvertex,label=above:$a_w$]{};
        \node (aw1) at (1.8,2){\footnotesize $2$};
        \node (aw2) at (2.3,2.2){\footnotesize $3$};
        \node (aw3) at (2.1,1.75){\footnotesize $1$};
        \node (awp) at (4,2)[rdvertex,label=above:$a'_w$]{};
        \node (awp1) at (3.7,2.2){\footnotesize $1$};
        \node (awp2) at (4.2,1.8){\footnotesize $2$};
        \node (bw) at (2,1)[opvertex,label=above:$b_w$]{};
        \node (bw1) at (2.3,1.2){\footnotesize $2$};
        \node (bw2) at (2.2,0.75){\footnotesize $1$};
        \node (bwp) at (4,1)[rdvertex,label=right:$b'_w$]{};
        \node (bwp1) at (3.8,1.25){\footnotesize $3$};
        \node (bwp2) at (3.8,0.8){\footnotesize $1$};
        \node (bwp3) at (4.2,0.6){\footnotesize $2$};
        \node (cw) at (2,0)[opvertex,label=below:$c_w$]{};
        \node (cw1) at (1.8,0.2){\footnotesize $2$};
        \node (cw2) at (2.2,0.2){\footnotesize $3$};
        \node (cw3) at (2.2,-0.25){\footnotesize $1$};
        \node (cwp) at (4,0)[rdvertex,label=below:$c'_w$]{};
        \node (cwp1) at (3.8,0.3){\footnotesize $2$};
        \node (cwp2) at (3.8,-0.2){\footnotesize $1$};
        \node (dw) at (2,-1)[opvertex,label=below:$d_w$]{};
        \node (dw1) at (2.2,-0.8){\footnotesize $2$};
        \node (dw2) at (2.2,-1.2){\footnotesize $1$};
        \node (dwp) at (4,-1)[rdvertex,label=below:$d'_w$]{};
        \node (dwp1) at (3.8,-0.7){\footnotesize $3$};
        \node (dwp2) at (3.8,-1.2){\footnotesize $1$};
        \node (dwp3) at (4.2,-1.2){\footnotesize $2$};
        \node (w) at (5,0)[opvertex,label=below:$w$]{};
        \node (w1) at (4.9,0.3){\footnotesize $2$};
        \node (w2) at (4.7,-0.1){\footnotesize $3$};
        \node (w3) at (5.2,0.15){\footnotesize $1$};
        \node (uw) at (6,0)[rdvertex,label=below:$n(w)$]{};
        \draw (s)--(t);
        \draw (t)--(aw);
        \draw (t)--(cw);
        \draw (aw)--(awp);
        \draw (bw)--(bwp);
        \draw (cw)--(cwp);
        \draw (dw)--(dwp);
        \draw (bwp)--(w);
        \draw (dwp)--(w);
        \draw (w)--(uw);
        \draw (aw)--(bwp);
        \draw (bw)--(cwp);
        \draw (cw)--(dwp);
        \path (dw) edge[bend right=90] node [left] {} (7,0.5);
        \path (7,0.5) edge[bend right] node [left] {} (awp);
        \path (0.70,0.05) edge[bend left,->] node [right] {} (1.05,0.25);
        \path (0.70,-0.05) edge[bend right,->] node [left] {} (1.30,-0.05);
        \end{tikzpicture}
    \caption{Illustration depicting the forcing gadget $F_w$ in the proof of Theorem~\ref{thm:1woman-strict}.      \label{fig:forcing-gadget}}
  \end{figure}
  We define a modified instance $I'$ of SMC as follows. 
  For each $ w \in \mathcal{W}^{\star}_I$, we create a \emph{\mbox{forcing} gadget}~$F_w$ which apart from $w$ contains the newly introduced women $a_w,b_w,c_w,d_w$ and men $a'_w,b'_w,c'_w,d'_w$.
  We also add the distinguished woman~$s$, who must be covered in $I'$, and the unique man $t$ in~$L(s)$.
  See Figure~\ref{fig:forcing-gadget} for an illustration. 

  Let $n(w)$ denote the unique man acceptable for some $w \in \mathcal{W}^{\star}_I$ in $I$. 
  Additionally, we let $Y=\{a_w,c_w \mid w \in \mathcal{W}^{\star}_I \}$, and we write $[Y]$ for an arbitrarily fixed ordering of the elements of~$Y$. 
  The preferences of the newly introduced men and women, as well as the modified preferences of those agents that find them acceptable, is given below.
  Here, again, indices take all possible values, and~$w$ can be any woman in $\mathcal{W}_I^{\star}$.
  We let $I'$ contain all other women and men defined in~$I$, having the same preferences as in $I$. 
  \begin{equation*}
    \begin{array}{llll} 
      L(s)    & = (t),                     &   
      L(t)    & = ([Y],s),                   \\
      L(a_w)  & = (b'_w,t,a'_w),           &
      L(a'_w) & = (a_w, d_w),                 \\ 
      L(b_w)  & = (c'_w,b'_w),             &
      L(b'_w) & = (b_w, w, a_w),              \\ 
      L(c_w)  & = (d'_w,t,c'_w),           &
      L(c'_w) & = (c_w, b_w),                  \\ 
      L(d_w)  & = (a'_w,d'_w),             &
      L(d'_w) & = (d_w, w, c_w),               \\ 
      L(w)    & = (n(w),b'_w,d'_w). \qquad & &
    \end{array}
  \end{equation*}
  We will show that $I'$ has a feasible matching with at most $b$ blocking pairs if and only if~$I$ has such a matching; this clearly proves the theorem. 

  First observe that any feasible matching $M'$ for $I'$ contains the edge $\{s,t\}$. 
  Thus, if some woman~$y$ in $Y$ is not matched by $M'$ to her first choice, then $\{y,t\}$ is blocking in $M'$. 
  Consider now $F_w$ for some $w \in \mathcal{W}^{\star}_I$. 
  It is straightforward to check that if $M'(w) \neq n(w)$, then there are at least two blocking pairs incident to a woman in $F_w$.
  Indeed, assume first that $\{t,a_w\}$ is the only blocking pair in $F_w$; this quickly implies $M'(c_w)=d'_w$ and $M'(a_w)=a'_w$, which in turn leads to $\{d_w,d'_w\}$ blocking~$M'$, a contradiction. 
  Second, assume that $\{t,a_w\}$ does not block $M'$; from this follows $M'(a_w)=b'_w$ and we have that $\{b'_w,w\}$ is a blocking pair for $M'$. 
  Now either $\{t,c_w\}$ is blocking (in which case our claim holds), or we get $M'(c_w)=d'_w$, which implies that $\{d'_w,w\}$ blocks~$M'$, again a contradiction.

  Now, let $\mathcal{W}_i$ be the women in $G_i$ that must be covered in $I$, i.e., $\mathcal{W}_i=\{s_i,t_i,u_i^1, \dots, u_i^{b+1}\}$. 
  Consider the number $\beta_i$ of blocking pairs for $M'$ that involve a woman either in the gadget~$G_i$ or in a gadget $F_w$ for some $w \in \mathcal{W}_i$. 
  On the one hand, if some $w \in \mathcal{W}_i$ is not matched by~$M'$ to $n(w)$, then $\beta_i \geq 2$ because of the blocking pairs in $F_w$.
  On the other hand, if each $w \in \mathcal{W}_i$ is matched by $M'$ to~$n(w)$, then using the arguments of the proof for Theorem~\ref{thm:smc-minblock-mainhardness}, we again know $\beta_i \geq 2$ because of the blocking pairs in $G_i$.
  Also, $\beta_i=2$ can only be achieved if (i) $M'(t_i)=n(t_i)$, as otherwise $\{t_i,n(t_i)\}$ would be blocking for~$M'$, in addition to the two blocking pairs in~$F_{t_i}$,  and (ii) $M'(u_i^h)=n(u_i^h)$ for each $h \in \{1, \dots, b+1\}$, as otherwise we would have $M'(s_i)=n(s_i)$ (so as to avoid having four blocking pairs due to women in~$F_{u_i^h}$ and $F_{s_i}$), implying at least one blocking pair in $G_i$ in addition to those in $F_{u_i^h}$. 

  Analogously, let $\beta_{i,j}$ denote the number of blocking pairs for $M'$ that involve a woman either in the gadget~$G_{\{i,j\}}$ or in a gadget $F_w$ for some $w \in \{s_{i,j},t_{i,j} \}$.
  Then either $\beta_{i,j} \geq 2$, or we know that $M'(w)=n(w)$ for both women $w \in \{s_{i,j},t_{i,j} \}$; in this case, from the proof of Theorem~\ref{thm:smc-minblock-mainhardness} we get $\beta_{i,j} \geq 1$. 
  However, supposing that $M'$ has at most $b=2k+\binom{k}{2}$ blocking pairs, it follows that $\beta_i=2$ and $\beta_{i,j}=1$ must hold for each $i \in \{1, \dots, k\}$ and each $i,j$ with $1 \leq i<j \leq k$, respectively. 

  Along the same lines as in the proof of Theorem~\ref{thm:smc-minblock-mainhardness}, it can also be verified that $\beta_{i,j}=1$ for each pair of indices $i,j$ can only be achieved if $M' \triangle M_s$ contains a path in each gadget~$G_i$.
  From $M'(w)=n(w)$ for each $w \in \mathcal{W}_i \setminus \{s_i\}$ we get that such a path contains at least one blocking pair. 
  This implies $M'(s_i)=n(s_i)$, as otherwise we would end up with $\beta_i \geq 3$ because of the blocking pairs incident to women of $F_{s_i}$. 

  Altogether, we have proved that $M'(w)=n(w)$ for each $w \in \mathcal{W}^{\star}_I$.   
  Hence, the restriction of~$M'$ to~$I$ yields a feasible matching for $I$ that admits at most~$b$ blocking pairs.

  For the other direction, suppose that $I$ has a feasible matching $M$. 
  Then it is easy to see that adding the edges $\{a_w,b'_w\}$, $\{b_w,c'_w\}$, $\{c_w,d'_w\}$, and $\{d_w,a'_w\}$ for each $w \in \mathcal{W}^{\star}_I$ together with the edge $\{s,t\}$ to~$M$ yields a feasible matching for $I'$ that contains exactly the same number of blocking pairs in $I'$ as~$M$ does in~$I$.
\end{proof}

\begin{theorem}
\label{thm:2women}
  SMC-1 is $\mathsf{W}[1]$-hard parameterized by $b+\Delta_{\mathcal{M}}$, even~if there is a master list over men as well as one over women, $|\mathcal{W}^{\star}|=2$, $\Delta_{\mathcal{W}} \leq 3$, and $|L(w)|=1$ for each $w\in\mathcal{W}^\star$.
\end{theorem}
\begin{proof}
  The proof is very similar to the one for Theorem~\ref{thm:1woman-strict}, so we will only sketch it. 
  Again, we are going to re-use the instance $I$ constructed in the proof of Theorem~\ref{thm:smc-minblock-mainhardness}, and construct a modified instance $I'$ of SMC, adding only two new women $z_1$ and $z_2$ and two men $m_1$ and $m_2$ to~$I$.
  We append $z_1$ and $z_2$, in this order, to the master list over women, and similarly, we append $m_1$ and~$m_2$ to the master list over men. 
  We define the women to be covered in $I'$ as $z_1$ and~$z_2$. 

  Again, we denote the set of women to be covered in $I$ by $\mathcal{W}^{\star}_I$, and we denote by $n(w)$ the unique man acceptable for some $w \in \mathcal{W}^{\star}_I$ in $I$.
  The preferences of the newly introduced men and women, as well as the modified preferences of those agents that find them acceptable, is given below (here, $[\mathcal{W}^{\star}_I]_{\prec}$ denotes the ordering of $\mathcal{W}^{\star}_I$ given by the master list). 
  We let $I'$ contain all other women and men defined in $I$, having the same preferences as in $I$.
  \begin{equation*}
    \begin{array}{llll} 
      L(z_1) & = (m_1), & 
      L(m_1) & = ([\mathcal{W}^{\star}_I]_{\prec}, z_1), \\ 
      L(z_2) & = (m_2), & 
      L(m_2) & = ([\mathcal{W}^{\star}_I]_{\prec}, z_2), \\ 
      L(w)   & = (n(w),m_1,m_2) \qquad &
\multicolumn{2}{l}{\forall w \in \mathcal{W}^{\star}_I.} 
    \end{array}
  \end{equation*}
  Arguing analogously as before in the proof of Theorem~\ref{thm:1woman-strict}, one can show that~$I'$ has a feasible matching with at most $b$ blocking pairs if and only if $I$ has such a matching; this suffices to prove the theorem.
\end{proof}

To contrast our intractability results, we show next that if each of the four parameters $|\mathcal{W}^{\star}|$, $|\mathcal{M}^{\star}|$, $\Delta_{\mathcal W}$, and~$\Delta_{\mathcal M}$ is constant, then SMC becomes polynomial-time solvable.
Our algorithm relies on the observation that in this case, the number of blocking pairs in an optimal solution is at most\linebreak $(\Delta_{\mathcal M}-1)|\mathcal{W}^{\star}|+(\Delta_{\mathcal W}-1)|\mathcal{M}^{\star}|$ by Theorem~\ref{thm-smc2-poly-approx}. 
Note that for instances of SMC-1, Theorem~\ref{thm:smc-allconstant} yields a polynomial-time algorithm already if both $|\mathcal{W}^\star|$ and $\Delta_{\mathcal M}$ are constant. 

\begin{theorem}
\label{thm:smc-allconstant}
  SMC can be solved in time $O(|I|^{(\Delta_{\mathcal M}-1)|\mathcal{W}^{\star}|+(\Delta_{\mathcal W}-1)|\mathcal{M}^{\star}|+1})$.
\end{theorem}
\begin{proof}
  By Theorem~\ref{thm-smc2-poly-approx}, there is a matching with at most $b_{\textup{max}}=(\Delta_{\mathcal M}-1)|\mathcal{W}^{\star}|+(\Delta_{\mathcal W}-1)|\mathcal{M}^{\star}|$ blocking pairs.
  Hence, if the number $b$ of blocking pairs allowed is at least~$b_{\max}$, then we can simply run the algorithm of Theorem~\ref{thm-smc2-poly-approx}.
  Otherwise, we can use Observation~\ref{observ:bounded-b}, which gives us the required run time.
\end{proof}

Importantly, restricting only three of the values $|\mathcal{W}^{\star}|$, $|\mathcal{M}^{\star}|$, $\Delta_{\mathcal W}$, and $\Delta_{\mathcal M}$ to be constant does not yield tractability for SMC, showing that Theorem~\ref{thm:smc-allconstant} is tight in this sense.
Indeed, Theorem~\ref{thm:1woman-strict} implies immediately that restricting the maximum length of the preference lists on only one side still results in a hard problem:
SMC remains $\mathsf{W}[1]$-hard with parameter $b+\Delta_{\mathcal{M}}$, even if $\Delta_{\mathcal{W} } = 3$, $|\mathcal{W}^\star|=1$, and $|\mathcal{M}^\star|=0$. 
On the other hand, Theorem~\ref{thm:smc-minblock-mainhardness} shows that the problem remains hard even if
$\Delta_{\mathcal{W} } = \Delta_{\mathcal{M} } =3$ and $|\mathcal{M}^\star|=0$. 

\section{SMC with Preference Lists of Length at most Two}
\label{sec:prefatmost2}
In this section we investigate the computational complexity of SMC where the maximum length of 
preference lists is bounded by $2$ on one side. 
This restriction leads to important tractable special cases: we obtain both polynomial-time algorithms
and fixed-parameter tractability results for various parameterizations. 

Let $I$ be an instance of SMC with underlying graph $G$.
Let $M_s$ be a stable matching in $I$, and let~$\mathcal{M}^{\star}_0$ and $\mathcal{W}^{\star}_0$ denote the set of distinguished men and women, respectively, unmatched by~$M_s$.
Furthermore, let $\mathcal{M}_0$ and $\mathcal{W}_0$ denote the set of all men and women, respectively, 
unmatched by $M_s$.
A path $P$ in $G$ is called an \emph{augmenting path}, if $M_s \triangle P$ is a matching, and either both endpoints of~$P$ are in $\mathcal{M}^\star_0 \cup \mathcal{W}^\star_0$, or one endpoint of $P$ is in $\mathcal{M}^\star_0 \cup \mathcal{W}^\star_0$, and its other endpoint is not distinguished. 
This definition ensures that for an augmenting path $P$, the set of distinguished men and women that are matched in $M_s \triangle P$ strictly contains the set of distinguished men and women matched in $M_s$.\footnote{
We remark that our concept of an augmenting path is analogous, but not identical, to the standard notion of an augmenting path in general matching theory. According to the standard definition, an augmenting path for a given matching $M$ is an $M$-alternating path~$P'$
such that $M \triangle P'$ is a matching containing more edges than $M$. 
In our case, however, instead of increasing the number of edges in the matching, we aim for a path which can be used to increase the number of distinguished men and women that are matched. 
} 
We will call an augmenting path $P$ \emph{masculine} or \emph{feminine} if it contains a man in $\mathcal{M}^{\star}_0$ or a woman in~$\mathcal{W}^{\star}_0$, respectively;
if $P$ is both masculine and feminine, then we call it \emph{neutral}.
If $P$ is not neutral, then we say that it \emph{starts} at the (unique) person from $\mathcal{M}^{\star}_0 \cup \mathcal{W}^{\star}_0$ it contains, and \emph{ends} at its other endpoint. 

\subsection{Covering constraints on one side}
\label{sect:polycase}
Here we deal with the SMC-1 problem where only women need to be covered. 
We first give a polynomial-time algorithm for SMC-1 when each man finds at most two women acceptable, 
and then show $\mathsf{NP}$-hardness of SMC-1 for instances where each woman finds at most two men acceptable. 
We start by considering the special case of SMC-1 where $\Delta_{\mathcal M} \leq 2$.
\begin{theorem}
\label{thm:restricted-matching-poly}
  There is a polynomial-time algorithm for the special case of \mbox{SMC-1} where each man finds at most two women acceptable.
\end{theorem}

{\bf High-level description.}
The main observation behind Theorem~\ref{thm:restricted-matching-poly} is that if $\Delta_{\mathcal{M}} \leq 2$, then any two augmenting paths starting from different women in $\mathcal{W}^{\star}_0$ are almost disjoint, namely they can only intersect at their endpoints. 
Thus, we can modify the stable matching $M_s$ by selecting augmenting paths starting from each woman in $\mathcal{W}^{\star}_0$ in an almost independent fashion: intuitively, we simply need to take care not to choose paths  sharing an endpoint---a task which can be managed by finding a bipartite matching in an appropriately defined auxiliary graph.
To ensure that the number of blocking pairs in the output is minimized, we will assign costs to the augmenting paths.
Roughly speaking, the cost of an augmenting path~$P$ determines the number of blocking pairs introduced when modifying $M_s$ along~$P$ (though certain special edges need not be counted); hence, our problem reduces to finding a bipartite matching with minimum weight in the auxiliary graph.

\smallskip
To present the algorithm of Theorem~\ref{thm:restricted-matching-poly} in detail, we start with the following properties of augmenting paths which are easy to prove using that $\Delta_{\mathcal{M}} \leq 2$:
\begin{proposition}
\label{prop-augmenting-paths}
  Suppose $\Delta_{\mathcal{M}} \leq 2$. Let  $P_1$ and $P_2$ be augmenting paths starting at women $w_1$ and~$w_2$, respectively.
  \begin{enumerate} 
    \item[(a)] If $w_1 \neq w_2$, then $P_1$ and $P_2$ are either vertex-disjoint, or they both end at some $m \in \mathcal M_0$, with $V(P_1) \cap V(P_2)=\{m\}$.
    \item[(b)] If there is an edge $\{m,w\}$ of $G$ (with $m \in \mathcal M$ and $w \in \mathcal W$) connecting $P_1$ and $P_2$, then $m \in \mathcal M_0$ and $P_1$ or $P_2$ must end at $m$. 
    \item[(c)] If $w_1 = w_2$ and $P$ is the maximal common subpath of $P_1$ and $P_2$ starting at $w_1$, then either $V(P_1) \cap V(P_2) = V(P)$, 
      or $P_1$ and $P_2$ both end at some $m \in \mathcal M_0$ and $V(P_1) \cap V(P_2) = V(P) \cup \{m\}$. 
  \end{enumerate}
\end{proposition}
With a set $P$ of edges (typically a set of augmenting paths) where $M_s \triangle P$ is a matching, we associate a \emph{cost}, which is the number of blocking pairs that $M_s \triangle P$ admits.
A pair $\{m,w\}$ for some $m \in \mathcal M$ and $w \in \mathcal W$ is \emph{special}, if $m \in \mathcal M_0$ and $w$ is the second (less preferred) woman in~$L(m)$. 
As it turns out, such edges can be ignored during certain steps of the algorithm; thus, we define the \emph{special cost} of $P$ as the number of non-special blocking pairs in $M_s \triangle P$.

\begin{lemma}
\label{lemma-costs}
  For vertex-disjoint augmenting paths $P_1$ and $P_2$ with cost $c_1$ and~$c_2$, respectively, the cost of $P_1 \cup P_2$ is at most $c_1+c_2$.
  Further, if the cost of $P_1 \cup P_2$ is less than $c_1+c_2$, then the following holds for $\{i_1,i_2\}=\{1,2\}$: 
  there is a special edge $\{m,w\}$ with $P_{i_1}$ ending at $m$ and $w$ appearing on~$P_{i_2}$; moreover, $\{m,w\}$ is blocking in $M_s \triangle P_{i_2}$, but not in $M_s \triangle (P_1 \cup P_2)$.
\end{lemma}
\begin{proof}
  First observe that if some edge $\{m,w\}$ has a common vertex with only one of the paths~$P_1$ and~$P_2$, say $P_1$, then $\{m,w\}$ is blocking in~$M_s \triangle P_1$ if and only if it is blocking in $M_s \triangle (P_1 \cup P_2)$.

  Consider now the case when $\{m,w\}$ connects $P_1$ and $P_2$. 
  By Proposition~\ref{prop-augmenting-paths}, this implies that one of the paths, say $P_1$, ends at $m \in \mathcal M_0$ (and $w$ lies on~$P_2$).
  Clearly, $\{m,w\}$ is not blocking in~$M_s \triangle P_1$, by the stability of $M_s$. 
  If, on the one hand, $w$ is the first choice of $m$, then $\{m,w\}$ is blocking in $M_s \triangle P_2$ exactly if it is blocking in $M_s \triangle (P_1 \cup P_2)$.
  If, on the other hand, $\{m,w\}$ is special, then it cannot be blocking in $M_s \triangle (P_1 \cup P_2)$, but it might be blocking in $M_s \triangle P_2$. 
  Putting all these facts together, the lemma follows immediately.
\end{proof}

We are ready to provide the algorithm, in a sequence of four steps.

\medskip
\noindent
{\bf Step 1: Computing all augmenting paths.}
By Proposition~\ref{prop-augmenting-paths}, if we de\-lete~$\mathcal M_0$ from the union of all augmenting paths starting at some $w \in \mathcal W^{\star}_0$, then we obtain a tree.
Furthermore, these trees are mutually vertex-disjoint for different starting vertices of $\mathcal W^{\star}_0$.
This allows us to compute all augmenting paths in linear time, e.g., by an appropriately modified version of a depth-first search algorithm (so that only augmenting paths are considered).
During this process, we can also compute the special cost of each augmenting path in a straightforward way. 

\medskip
\noindent
{\bf Step 2: Constructing an auxiliary graph.}
Using the results of the computation of Step~1, we construct an edge-weighted single bipartite graph~$G_{\textup{path}}=(U,V;E)$ 
as follows. 
To define the vertices of $G_{\textup{path}}$ we set $U:=\mathcal W^{\star}_0$ and $V:=\mathcal M_0 \cup \{w' \mid w \in \mathcal W^{\star}_0\}$, so for each woman $w\in\mathcal W^{\star}_0$ we create a corresponding new vertex $w'$.
The edge set $E$ contains an edge $\{w,w'\}$ for each $w \in U$, as well as an edge $\{w,m\}$ whenever $w \in U$, $m \in\mathcal M_0$ and there exists an augmenting path with endpoints $w$ and~$m$.
We define the weight of an edge $\{w,w'\}$ as the minimum special cost~$c_w^{\textup{min}}$ of any augmenting path starting at~$w$ and \emph{not} ending in $\mathcal M_0$, and we define the weight of an edge $\{w,m\}$ with $w \in U$ and $m \in\mathcal M_0$ as the minimum special cost of any augmenting path with endpoints $w$ and $m$.

\medskip
\noindent
{\bf Step 3: Computing a minimum weight matching.}
We compute a matching $M_P$ in~$G_{\textup{path}}$ covering~$U$ and having minimum weight; this can be done in polynomial time by, e.g., the
Hungarian method~\cite{kuhn-weighted-matching}.
Observe that the matching $M_P$ corresponds to a set of augmenting paths $\mathcal{P}=\{ P_w \mid w \in \mathcal W^{\star}_0\}$ that are mutually vertex-disjoint by Proposition~\ref{prop-augmenting-paths}. 
Recall that the special cost of~$P_w$ is the weight of the edge in~$M_P$ incident to~$w$. 

\medskip
\noindent
{\bf Step 4: Eliminating blocking special edges.}
In this step, we modify $\mathcal{P}$ iteratively. 
We start by setting $\mathcal{P}_{\textup{act}}=\mathcal{P}$.
At each iteration we modify~$\mathcal{P}_{\textup{act}}$ as follows. 
We check whether there exists a special edge  $\{m^*,w^*\}$ that is blocking in $M_s \triangle \mathcal{P}_{\textup{act}}$. 
If yes, then notice that~$m^*$ is not matched in $M_s \triangle \mathcal{P}_{\textup{act}}$, because 
$\{m^*,w^*\}$ is special and thus $m^* \in \mathcal M_0$. 
Let $P$ be the path of~$\mathcal{P}_{\textup{act}}$ containing~$w^*$.
We modify~$\mathcal{P}_{\textup{act}}$ by truncating~$P$ to its subpath between its starting vertex and~$w^*$, and appending to it the edge $\{m^*,w^*\}$. 
This way, $\{m^*,w^*\}$ becomes an edge of the matching~$M_s \triangle \mathcal{P}_{\textup{act}}$. 
The iteration stops when there is no special edge blocking~$M_s \triangle \mathcal{P}_{\textup{act}}$. 
Note that once a special edge ceases to be blocking in $M_s \triangle \mathcal{P}_{\textup{act}}$, it cannot become blocking again during this process, so the algorithm performs at most $|\mathcal M_0|$ iterations.
For each $w \in \mathcal W^{\star}_0$, let~$P^*_w$ denote the augmenting path in~$\mathcal{P}_{\textup{act}}$ covering~$w$ at the end of Step~4; we define $\mathcal{P}^* = \{ P^*_w \mid w \in \mathcal W^{\star}_0\}$ and output the matching~$M_s \triangle \mathcal{P}^*$. 

\medskip
This completes the description of the algorithm; we now provide its analysis.
\begin{lemma}
\label{lem:menlength2correctness}
  $M_{\textup{sol}}:=M_s \triangle \mathcal{P}^*$ is a feasible matching for $I$, and the number of blocking pairs for~$M_{\textup{sol}}$ is at most the weight of $M_P$.
\end{lemma}
\begin{proof}
  Consider the situation when the iteration in Step 4 deals with a special edge~$\{m^*,w^*\}$ blocking in $\mathcal{P}_{\textup{act}}$.
  Notice that since $w^*$ is the second woman in $L(m^*)$ (by the definition of a special edge), and since $\{w^*,m^*\}$ is blocking in~$M_s \triangle \mathcal{P}_{\textup{act}}$, we know that $m^*$ is unmatched in $M_s \triangle \mathcal{P}_{\textup{act}}$, that is, $m^*$ does not lie on any of the augmenting paths in~$\mathcal{P}_{\textup{act}}$. 
  From this follows that the augmenting paths in~$\mathcal{P}_{\textup{act}}$, and hence in $\mathcal{P}^*$, remain mutually vertex-disjoint. 
  Therefore, $M_{\textup{sol}}$ is indeed a matching.
  As it covers~$\mathcal W^{\star}_0$, and no augmenting path ends at a woman in $\mathcal W^{\star} \setminus \mathcal W^{\star}_0$,  matching $M_{\textup{sol}}$ is feasible. 

  Clearly, Step 4 ensures that there are no blocking special edges in $M_{\textup{sol}}$.
  Note that when the algorithm modifies $P_w$ for some $w \in \mathcal W^{\star}_0$, at most one new blocking pair may arise with respect to $M_s \triangle \mathcal{P}_{\textup{act}}$, and from the stability of $M$ and Proposition~\ref{prop-augmenting-paths} it follows that such an edge must be a special edge (incident to the man at which $P_w$ ends before its modification).
  This means that Step~4 gets rid of all blocking special edges without introducing any non-special blocking edges.
  Hence, we obtain that the cost of $P^*_w$ is at most the special cost of $P_w$, for each $w \in \mathcal W^{\star}_0$.
  By Lemma~\ref{lemma-costs}, the number of blocking pairs that $M_{\textup{sol}}$ admits is at most the sum of the costs of all augmenting paths in~$\mathcal{P}^*$; this finishes the proof.
\end{proof}

To show that our algorithm is correct and $M_{\textup{sol}}$ is optimal, by Lemma~\ref{lem:menlength2correctness} it suffices to prove that 
the weight of $M_P$ is at most the number of blocking pairs in $M^{\textup{opt}}$, where $M^{\textup{opt}}$ denotes an optimal solution in~$I$.
To this end, we are going to define a matching covering $\mathcal W^{\star}_0$ in~$G_{\textup{path}}$ whose weight is at most the number of blocking pairs in $M^{\textup{opt}}$.

Clearly, $M_s \triangle M^{\textup{opt}}$ contains an augmenting path $Q_w$ covering $w$ for each $w \in \mathcal W^{\star}_0$.
If some~$Q_w$ ends at a man $m \in \mathcal M_0$, then clearly no other path in~$M_s \triangle M^{\textup{opt}}$  can end at $m$.
So let us take the matching~$M_Q$ in~$G_{\textup{path}}$ that includes all pairs $\{m,w\}$ where $Q_w$ ends at $m \in \mathcal M_0$ for some $w \in \mathcal W^{\star}_0$.
Also, we put $\{w,w'\}$ into~$M_Q$ if $Q_w$ does not end at a man of $\mathcal M_0$. 
Note that~$M_Q$ is indeed a matching. 

It remains to show that the weight of $M_Q$ is at most the number of blocking pairs in~$M^{\textup{opt}}$.
By definition, the weight of $M_Q$ is at most the sum of the special costs of the paths~$Q_w$ for every $w \in \mathcal W^{\star}_0$.
By Lemma~\ref{lemma-costs}, any non-special blocking pair in $M_s \triangle Q_w$ remains a blocking pair in $M_s \triangle ( \bigcup_{w \in\mathcal W^{\star}_0} Q_w)$, and hence in~$M^{\textup{opt}}$ as well.
Hence, there is a matching in~$G_{\textup{path}}$ with weight at most the number of blocking pairs 
in an optimal solution, implying the correctness of our algorithm.
As the algorithm runs in polynomial time, Theorem~\ref{thm:restricted-matching-poly} follows.

\medskip
By contrast to Theorem~\ref{thm:restricted-matching-poly}, if men may have preference lists of length~$3$, then SMC-1 (and hence SMC) is $\mathsf{NP}$-hard even if each woman finds at most two men acceptable. 

\begin{theorem}
\label{thm-case23}
  SMC-1 is $\mathsf{NP}$-hard even if $\Delta_{\mathcal{W}} = 2$ and $\Delta_{\mathcal{M}} = 3$.
\end{theorem}
\begin{proof}
  We give a reduction from the $\mathsf{NP}$-hard {\sc Vertex Cover} problem, asking whether the input graph~$G$ has a vertex cover of size at most $k$.
  We order the vertices of $G$ arbitrarily, and denote the $h$-th neighbor of some vertex $x$ by $n(x,h)$ for any $h \in \{1, \dots, d_G(x)\}$. 

  Let us construct an instance $I$ of {\sc SMC} as follows; see Figure~\ref{fig:case23} for an illustration. 
  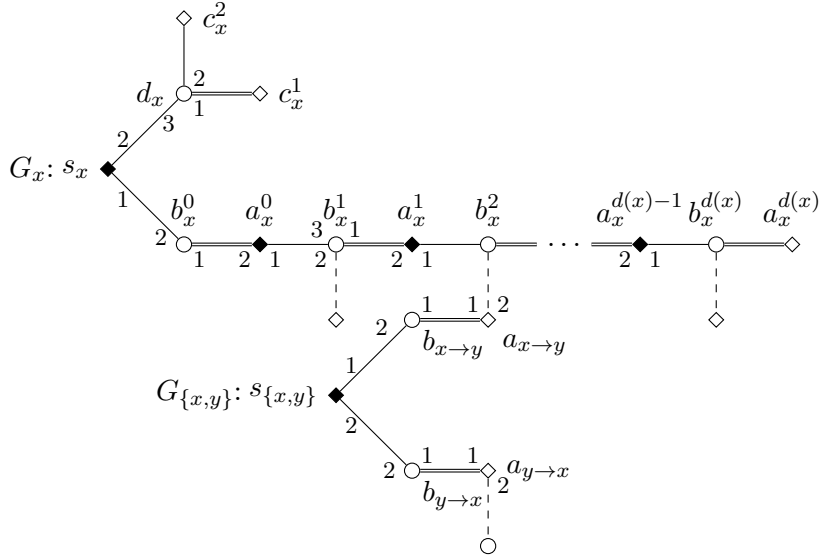
\begin{figure}
    \tikzset{rdvertex/.style={minimum size=2mm,circle,fill=white,draw, inner sep=0pt},
        sqvertex/.style={minimum size=2mm,diamond,fill=black,draw, inner sep=0pt},
        opvertex/.style={minimum size=2mm,diamond,fill=white,draw, inner sep=0pt},
        decoration={markings,mark=at position .5 with {\arrow[black,thick]{stealth}}}}
        \centering
        \begin{tikzpicture}
        \node (gx) at (0,0){$G_x$:};
        \node (sx) at (1,0) [sqvertex,label=left:$s_x$]{};
        \node (sx-1) at (1.2,-0.4) {{\footnotesize $1$}};
        \node (sx-2) at (1.2,0.4) {{\footnotesize $2$}};
        \node (dx) at (2,1) [rdvertex,label=left:$d_x$]{};
        \node (dx-1) at (2.2,0.8){{\footnotesize $1$}};
        \node (dx-2) at (2.2,1.2){{\footnotesize $2$}};
        \node (dx-3) at (1.8,0.6){{\footnotesize $3$}};
        \node (cx1) at (3,1) [opvertex,label=right:$c_x^1$]{};
        \node (cx2) at (2,2) [opvertex,label=right:$c_x^2$]{};
        \node (ex0) at (2,-1) [rdvertex,label=above:$b_x^0$]{}; 
        \node (ex0-1) at (2.2,-1.2){{\footnotesize $1$}}; 
        \node (ex0-2) at (1.7,-0.9){{\footnotesize $2$}}; 
        \node (ax0) at (3,-1) [sqvertex,label=above:$a_x^0$]{};
        \node (ax0-1) at (3.2,-1.2){{\footnotesize $1$}};
        \node (ax0-2) at (2.8,-1.2){{\footnotesize $2$}};
        \node (ex1) at (4,-1) [rdvertex,label=above:$b_x^1$]{};
        \node (ex1p) at (4,-2) [opvertex]{};
        \node (ex1-1) at (4.25,-0.8){{\footnotesize $1$}};
        \node (ex1-2) at (3.8,-1.2){{\footnotesize $2$}};
        \node (ex1-3) at (3.75,-0.8){{\footnotesize $3$}};
        \node (ax1) at (5,-1) [sqvertex,label=above:$a_x^1$]{};
        \node (ax1-1) at (5.2,-1.2){{\footnotesize $1$}};
        \node (ax1-2) at (4.8,-1.2){{\footnotesize $2$}};
        \node (ex2) at (6,-1) [rdvertex,label=above:$b_x^2$]{};
        \node (axy) at (6,-2) [opvertex,label={below right:$a_{x\rightarrow y}$}]{}; 
        \node (axy-1) at (5.8,-1.8){{\footnotesize $1$}};
        \node (axy-2) at (6.2,-1.8){{\footnotesize $2$}};
        \node (exy) at (5,-2) [rdvertex,label={[label distance=-3pt]-60:$b_{x\rightarrow y}$}]{};
        \node (exy-1) at (5.2,-1.8){{\footnotesize $1$}};
        \node (exy-2) at (4.6,-2.1){{\footnotesize $2$}};
        \node (sxy) at (4,-3) [sqvertex,label=left:$s_{\{x,y\}}$]{};
        \node (sxy-1) at (4.2,-2.6){{\footnotesize $1$}};
        \node (sxy-2) at (4.2,-3.4){{\footnotesize $2$}};
        \node (gxy) at (2.2,-3) {$G_{\{x,y\}}$:};
        \node (eyx) at (5,-4) [rdvertex,label={[label distance=-3pt]-60:$b_{y\rightarrow x}$}]{};
        \node (eyx-1) at (5.2,-3.8){{\footnotesize $1$}};
        \node (eyx-2) at (4.7,-4){{\footnotesize $2$}};
        \node (ayx) at (6,-4) [opvertex,label=right:$a_{y\rightarrow x}$]{};
        \node (ayx-1) at (5.8,-3.8){{\footnotesize $1$}};
        \node (ayx-2) at (6.2,-4.2){{\footnotesize $2$}};
        \node (ayxp) at (6,-5) [rdvertex]{};
        \node (axdx1) at (8,-1) [sqvertex,label=above:$a_x^{d_G(x)-1}$]{};
        \node (axdx1-1) at (8.2,-1.2){{\footnotesize $1$}};
        \node (axdx1-2) at (7.8,-1.2){{\footnotesize $2$}};
        \node (exdx) at (9,-1) [rdvertex,label=above:$b_x^{d_G(x)}$]{};
        \node (exdxp) at (9,-2) [opvertex]{};
        \node (axdx) at (10,-1) [opvertex,label=above:$a_x^{d_G(x)}$]{};
        \node (hdots) at (7,-1) {$\hdots$};
        \draw (sx)--(dx);
        \draw (dx)--(cx2);
        \draw[double] (dx)--(cx1);
        \draw (sx)--(ex0);
        \draw[double] (ex0)--(ax0);
        \draw (ax0)--(ex1);
        \draw[dashed] (ex1)--(ex1p);
        \draw[double] (ex1)--(ax1);
        \draw (ax1)--(ex2);
        \draw[dashed] (ex2)--(axy);
        \draw[double] (axy)--(exy);
        \draw (exy)--(sxy);
        \draw (sxy)--(eyx);
        \draw[double] (eyx)--(ayx);
        \draw[dashed] (ayx)--(ayxp);
        \draw[double] (ex2)--(hdots);
        \draw[double] (hdots)--(axdx1);
        \draw (axdx1)--(exdx);
        \draw[dashed] (exdx)--(exdxp);
        \draw[double] (exdx)--(axdx);
        \end{tikzpicture}
    \caption{Illustration of a node gadget $G_x$ and an edge gadget $G_{\{x,y\}}$ constructed in the proof of Theorem~\ref{thm-case23}.
        Double edges denote edges of a stable matching for $I$, and dashed edges are those leaving some gadget.
        The example depicted assumes $y=n(x,2)$.\label{fig:case23}}
  \end{figure}  
  For each vertex $x \in V(G)$ we construct a \emph{node gadget} $G_x$ which contains women $s_x$, $a_x^0, \dots, a_x^{d_G(x)}$, $c_x^1$ and~$c_x^2$, and men $b_x^0, \dots, b_x^{d_G(x)}$, and~$d_x$. 
%
  For each edge $\{x,y\} \in E(G)$ we also construct an \emph{edge gadget} $G_{\{x,y\}}$ involving women $s_{\{x,y\}}$, $a_{x \to y}$ and~$a_{y \to x}$, and men $b_{x \to y}$ and $b_{y \to x}$. 
%
  Furthermore, there are two edges in the underlying graph connecting $G_{\{x,y\}}$ to $G_x$ and $G_y$, namely $\{a_{x \to y}, b_x^h\}$ and $\{a_{y \to x} , b_y^{\ell}\}$ where $y=n(x,h)$ and $x=n(y,\ell)$. 

  The preference lists of $I$ are given in Table~\ref{table-prefs-case23}.
  Let the set of women with covering constraints be
  \begin{equation*}
    \mathcal W^\star = 
    		\bigcup_{x \in V(G)} \{s_x, a_x^0, \dots, a_x^{d_G(x)-1} \} 
    		\cup 
		\bigcup_{\{x,y\} \in E(G)} \{s_{\{x,y\}} \} 
		,
  \end{equation*}
  and set the number of blocking pairs allowed to be $|V(G)|+k$. 
  \begin{table}
    \renewcommand{\arraystretch}{1.3}
    \begin{tabular}{lll}
      $L(s_x)$ & $= (b_x^0,d_x),$                  & \\
      $L(b_x^0)$ & $= (a_x^0,s_x),$                  & \\ 
      $L(b_x^h)$ & $= (a_x^h,a_{x \to y},a_x^{h-1})$ & \textrm{ where $1 \leq h \leq d_G(x)$ and $y=n(x,h)$,} \\
      $L(a_x^h)$ & $= (b_x^{h+1},b_x^h)$             & \textrm{ where $0 \leq h < d_G(x)$,} \\
      $L(a_x^{d_G(x)})$ & $= (b_x^{d_G(x)}),$ & \\
      $L(c_x^h)$ & $= (d_x),$ 						& \textrm{ where $h \in \{1,2\}$,} \\  
      $L(d_x)$ & $= (c_x^1,c_x^2,s_x),$ & \\
      $L(s_{\{x,y\}})$ & $= (b_{x \to y},b_{y \to x})$ & \textrm{ where $x$ precedes $y$,}\\
      $L(b_{x \to y})$ & $= (a_{x \to y},s_{\{x,y\}}),$ & \\ 
      $L(a_{x \to y})$ & $= (b_{x \to y},b_x^h)$ & \textrm{ where $y=n(x,h)$.}
    \end{tabular}
    \caption{Preference lists of women and men in the proof of Theorem~\ref{thm-case23}.
      When not stated otherwise, indices take all possible values.\label{table-prefs-case23}}
  \end{table}
  We are going to prove that $I$ admits a feasible matching with at most $|V(G)|+k$ blocking pairs if and only if there is a vertex cover of size $k$ in the graph $G$. 

  \medskip
  ''$\Rightarrow$'':
  Let $M$ be a feasible matching with at most $|V(G)|+k$ blocking pairs. 
  We say that the \emph{cost} of some gadget $G_x$ (or $G_{\{x,y\}}$) is the number of edges blocking~$M$ which are incident to some man of $G_x$ (or $G_{\{x,y\}}$, respectively.)  
  We will prove that the set $S$ of vertices $x$ for which~$G_x$ has cost at least~$2$ is a vertex cover of $G$.

  First, let us consider some $x$ for which $M(s_x)=d_x$. 
  In this case, both $c_x^1$ and~$c_x^2$ form a blocking pair for~$M$ with $d_x$, implying $x \in S$. 
  Second, let us consider some~$x$ with $M(s_x)=b_x^0$. 
  Since each $a_x^h$ with $0 \leq h <d_G(x)$ must be matched by $M$ (because it is contained in $\mathcal{W}^\star$), we obtain $M(a_x^h)=b_x^{h+1}$ for each such $h$. 
  Hence, $a_x^{d_G(x)}$ and $b_x^{d_G(x)}$ form a blocking pair for $M$.  
  Moreover, if the woman~$a_{x \to y}$ is unmatched in~$M$ for some~$y$, then $\{a_{x \to y}, b_x^h\}$ is also a blocking pair in $M$ (where $y=n(x,h)$), and implies a cost of at least $2$ for~$G_x$.
  Therefore, we can observe that if $x \notin S$, then~$a_{x \to y}$ must be matched by $M$ to $b_{x \to y}$ for each neighbor $y$ of $x$ in $G$.   

  However, for any $\{x,y\} \in E(G)$, $M$ must match $s_{\{x,y\}}$ either to $b_{x \to y}$ or to~$b_{y \to x}$ (because $s_{\{x,y\}}$ is contained in $\mathcal{W}^\star$), which means that  $x \in S$ or $y \in S$.
  This proves that $S$ is indeed a vertex cover for~$G$. 
  Moreover, the number of vertices in $S$ can be at most $k$, since each~$G_x$ with $x \in S$ has cost at least~$2$, each~$G_x$ with $x \notin S$ has cost at least $1$, and the total cost of all gadgets cannot exceed our budget $|V(G)|+k$. 

  \medskip
  ''$\Leftarrow$'':
  Given a vertex cover $S$ of size at most $k$ for $G$, we define a matching~$M$ with the desired properties.
  Namely, for each $x \in S$ we set~$M(s_x)=d_x$ and $M(a_x^h)=b_x^h$ for each $h \in \{ 0, \dots, d_G(x)\}$.
  In this case, $c_x^1, c_x^2$ are unmatched by $M$, both forming a blocking pair with~$d_x$.
  By contrast, all of the men $b_x^0, \dots, b_x^{d_G(x)}$ get their first choices.

  Next, for each $x \in V(G) \setminus S$ we set $M(s_x)=b_x^0$, $M(c_x^1)=d_x$, and $M(a_x^h)=b_x^{h+1}$ for each $h \in \{ 0, \dots, d_G(x)-1\}$. 
  Note that $a_x^{d_G(x)}$ is unmatched by $M$, and thus forms a blocking pair with~$b_x^{d_G(x)}$.
  Observe also that $d_x$ is not contained in any blocking pair.

  Finally, for some $\{x,y\} \in E(G)$, let us assume $y \in S$ (since $S$ is a vertex cover, it contains~$x$ or~$y$). 
  We set $M(s_{\{x,y\}})= b_{y \to x}$ and $M(a_{x \to y})=b_{x \to y}$.
  Note that~$a_{x \to y}$ gets her first choice, so it cannot be involved in a blocking pair. 
  Although $a_{y \to x}$ is unmatched by $M$, we know that it cannot form a blocking pair with $b_y^{\ell}$ where $x=n(y,\ell)$, because $y \in S$ and hence $b_y^{\ell}$ is assigned her first choice by~$M$.
  Thus, no man or woman of some edge gadget participates in a blocking pair, and therefore we obtain that the total number of blocking pairs for $M$ is exactly $|V(G)|+k$. 

  Since $M$ is feasible, the theorem follows.
\end{proof}

\subsection{Covering constraints on both sides}
Let us now investigate the complexity of SMC with covering constraints both for men and women. 
If we restrict the maximum length of preference lists on both sides to be at most $2$, SMC becomes linear-time solvable. To see this, observe that by $\max (\Delta_{\mathcal{W}},\Delta_{\mathcal{M}}) \leq 2$, the underlying graph $G$  must be a collection of paths and cycles. Thus, we can process the connected components of $G$ one by one. Applying dynamic programming on each component $K$, we can determine the minimum number $b_K$ of blocking pairs for any matching that is feasible for $K$,  together with a matching $M_K$ for $K$ that admits $b_K$ blocking pairs. If $K$ is a path, then this can be done in a straightforward manner, traversing the edges of $K$ one by one in linear time. If $K$ is a cycle, then we can pick any edge $e$ of $K$, guess whether it is contained in $M_K$, or blocks $M_K$, or neither of the two; we can then process the remainder of $K$ (which is a path) taking into account our guess for $e$. To compute the minimum number of blocking pairs that a feasible matching admits in our instance, we can simply sum up the values $b_K$ over each connected component $K$ of $G$.
Hence, we arrive at the following.

\begin{observation}
\label{observ:maxlist2}
  Instances of SMC with $\max(\Delta_{\mathcal{W}},\Delta_{\mathcal{M}}) \leq 2$ are linear-time solvable.
\end{observation}

Recall that the case where $\Delta_{\mathcal{W}}=2$ and $\Delta_{\mathcal{M}}=3$ is $\mathsf{NP}$-hard by Theorem~\ref{thm-case23}, 
even if there are no distinguished men to be covered. 
However, switching the roles of men and women in Theorem~\ref{thm:restricted-matching-poly}, we obtain that if there are no women to be covered, then $\Delta_{\mathcal W} \leq 2$ guarantees polynomial-time solvability for SMC. 
This raises the natural question whether SMC with $\Delta_{\mathcal W} \leq 2$ can be solved efficiently  
if the number of distinguished women is bounded. 
Next we show that this is unlikely, as the problem turns out to be $\mathsf{NP}$-hard for $|\mathcal{W}^{\star}|=1$.

\begin{theorem}
\label{thm:smc2-npc-1woman}
  SMC is $\mathsf{NP}$-hard, even if $\Delta_{\mathcal{W}}=2$, $|\mathcal{W}^{\star}|=1$ and there is only one man $m$ with $|L(m)|>3$. 
\end{theorem}
\begin{proof}
  We present a reduction from the following special case of \textsc{Exact-3-Cover}.
  We are given a set $U=\{u_1, \dots, u_n\}$, a family $\mathcal{S}$ of subsets $S_1, \dots, S_m$ of $U$, each having size~$3$, such that each element of~$U$ occurs in at most three sets of $\mathcal{S}$.
  The task is to decide whether there exists a collection of $n/3$ sets in~$\mathcal{S}$ whose union covers $U$; such a collection of subsets is called an \emph{exact cover} for $U$. 
  This problem is $\mathsf{NP}$-complete~\cite[GT2]{GareyJohnson1979}. 
  We construct an equivalent instance $I$ of SMC as follows.

  The set $\mathcal{W}$ of women in $I$ contains the women $s_j$, $p_j^1$, $p_j^2$, $p_j^3$, and $q_j$ for each $j \in \{ 1,\hdots,m\}$, women~$x$ and $y$, as well as two women $a_{i,j},b_{i,j}$ for each element $u_i$ contained in $S_j$ for each $j \in \{ 1,\hdots,m\}$. 
  The men defined in $I$ are $\hat{p}^1_j$, $\hat{p}^2_j$, $\hat{p}^3_j$, $\hat{q}_j$, and $t_j$ for each $j \in \{ 1,\hdots,m\}$, a man~$c_i$ for each $u_i \in U$, a man $\hat{b}_{i,j}$ for each element $u_i$ contained in~$S_j$ for each $j \in \{ 1,\hdots,m\}$, plus one additional man~$\hat{y}$.
  (The pairs $\{w,\hat{w}\}$ form a stable matching in $I$.)
  The only distinguished woman in $I$ is $x$, and the set of distinguished men is $\mathcal{M}^\star = \{c_i \mid i = 1,\hdots,n\} \cup \{t_j \mid j=1,\hdots,m\}$.
  The preferences of each person are as shown in Table~\ref{table-prefs-smc2}.
  \begin{table}
    \renewcommand{\arraystretch}{1.3}
    \begin{tabular}{lll} 
      $L(x)$ & $= (\hat{y}),$ & \\
      $L(y)$ & $= (\hat{y}),$ & \\
      $L(s_j)$ & $= (\hat{p}^3_j,\hat{y}),$ & \\
      $L(p^1_j)$ & $= (\hat{p}^1_j,t_j),$ & \\
      $L(p^h_j)$ & $= (\hat{p}^{h-1}_j,\hat{p}^h_j)$ & \textrm{ for $h \in \{2,3\}$,} \\
      $L(q_j)$ & $= (\hat{q}_j,t_j),$ & \\
      $L(a_{i,j})$ & $= (\hat{b}_{i,j},\hat{p}^h_j)$ & \textrm{ for the unique $h \in \{1,2,3\}$ that satisfies $i=\textup{ind}(j,h)$, }\\
            $L(b_{i,j})$ & $= (\hat{b}_{i,j},c_i),$ & \\
            $L(\hat{y})$ & $= (y,s_1,s_2, \dots, s_m,x),$ &\\
            $L(t_j)$ & $= (p^1_j,q_j),$ & \\
            $L(\hat{p}^h_j)$ & $= (p^h_j,a_{i,j},p^{h+1}_j)$ & \textrm{ for $h \in \{1,2\}$ and $i=\textup{ind}(j,h)$,} \\
            $L(\hat{p}^3_j)$ & $= (p^3_j,a_{i,j},s_j)$ & \textrm{ where $i=\textup{ind}(j,3)$,} \\
            $L(\hat{q}_j)$ & $= (q_j),$ & \\
            $L(\hat{b}_{i,j})$ & $= (b_{i,j},a_{i,j}),$ & \\
            $L(c_i)$ & $= ([B_i])$ & \textrm{ where $B_i=\{b_{i,j} \mid u_i \in S_j\}$ and $[B_i]$ is some fixed ordering of $B_i$.} \\
        \end{tabular}
    \caption{Preference lists of women and men in the proof of Theorem~\ref{thm:smc2-npc-1woman}.
    We denote by $\textup{ind}(j,h)$ the index~$i$ for which~$u_i$ is the $h$-th element in $S_j$. 
    When not stated otherwise, indices take all possible values.\label{table-prefs-smc2}}
  \end{table}
  Note that since each subset $S_j$ contains three elements, and each element $u_i$ is contained in at most three subsets  from~$\mathcal{S}$, we get that all men except for $\hat{y}$ have a preference list of length at most~3, as promised.
  To finish the construction, we set the number of allowed blocking pairs to be $b=2m+2n/3+1$. 

  We claim that $I$ admits a feasible matching with at most $b$ blocking pairs if and only if $(U,\mathcal{S})$ is a ``yes''-instance of \textsc{Exact-3-Cover}. 

  ``$\Rightarrow$'':
  Suppose that $M$ is a feasible matching for $I$ with at most $b$ blocking pairs.
  Clearly, as $x$ is distinguished,~$M$ must contain the edge $\{x, \hat{y}\}$.
  Thus, $\{ y,\hat{y}\}$ is blocking in $M$.
  Second, since~$t_j$ is distinguished for each $j \in \{1, \dots, m\}$, we get that $M$ matches $t_j$ either to~$q_j$ or to $p^1_j$, which in turn implies that either $\{q_j, \hat{q}_j\}$ or $\{p^1_j, \hat{p}^1_j\}$ blocks $M$, leading to $m$ additional blocking pairs for~$M$. 
  Third, consider now any man~$c_i$, $i \in \{1, \dots, n\}$: as~$c_i$ is distinguished, we know $M(c_i)=b_{i,j}$ for some $j$ such that $S_j$ contains $u_i$.
  In this case, $\{b_{i,j}, \hat{b}_{i,j}\}$ is also a blocking pair for $M$, yielding~$n$ blocking pairs of such form.
  Thus, if $b_U$ denotes the number of blocking pairs among the edges $\{b_{i,j}, \hat{b}_{i,j}\}$ for indices $i$ and~$j$ with $u_i \in S_j$, then we get $b_U\geq n$. 
  This adds up to $m+b_U+1 \geq m+n+1$  blocking pairs so far.

  Let us define the set $E_j$ of those edges that are incident to $s_j$, $\hat{p}^3_j$,~$\hat{p}^2_j$, or~$\hat{p}^1_j$, but not to $p^1_j$ for some $j \in \{1, \dots, m\}$ ; note that these sets are pairwise disjoint, and none of them contains any of the (possibly) blocking edges mentioned in the previous paragraph.
  Let $k$ be the number of indices $j$ for which~$E_j$ contains no blocking pairs for $M$; we call such indices (and the subsets~$S_j$ corresponding to them) \emph{selected}.
  The $m-k$ non-selected indices clearly correspond to at least $m-k$ blocking pairs for~$M$ (each contained in~$E_j$ for some~$j$).

  Suppose now that $j$ is selected.
  Then, since $\{s_j,\hat{y}\}$ is not blocking, we get $M(s_j)=\hat{p}^3_j$, since $\hat{y}$ prefers $s_j$ to its partner $x$.
  This implies that $M(p^3_j)=\hat{p}^2_j$, as otherwise $\{p^3_j, \hat{p}^3_j\}$ would be blocking in~$M$. 
  Similarly, from this we obtain $M(p^2_j)=\hat{p}^1_j$. 
  Moreover, for each $h \in \{1,2,3\}$, to ensure that $\{\hat{p}^h_j,a_{i,j}\}$ does not block $M$ where $u_i$ is the $h$-th  element of $S_j$, we must have $M(a_{i,j})=\hat{b}_{i,j}$. 
  Hence, $\{b_{i,j},\hat{b}_{i,j}\}$ must be blocking in $M$.
  Since this holds for each $h \in \{1,2,3\}$ and each selected $j$, we get $b_U \geq 3k$. 

  Summing up the blocking pairs identified so far, we know that $M$ admits at least $1+m+(m-k)+\max(n,3k)$ blocking pairs. 
  Using that this must be upper-bounded by $b=1+m+n+(m-n/3)$, it is easy to show that only $k=n/3$ is possible.
  This yields that there exist exactly $n/3$ selected indices, and for all such indices $j$ all the edges 
  $\{b_{i,j},\hat{b}_{i,j}\}$ for which $u_i \in S_j$ are blocking with respect to $M$.
  Moreover, we also must have $b_U=n$, as otherwise the number of blocking pairs would exceed $b$.

  However, observe that for each $i \in \{1, \dots, n\}$, there must exist some $j$ with $u_i \in S_j$ for which the pair $\{b_{i,j},\hat{b}_{i,j}\}$ is blocking in $M$ (because $c_i$ is distinguished), implying that for each $u_i \in U$ there must exist some selected~$S_j$ that contains $u_i$.    
  Since there are exactly $n/3$ selected sets in $\mathcal{S}$, we get that they form an exact covering of~$U$.

  ``$\Leftarrow$'': 
  Suppose that $(U,\mathcal{S})$ is a ``yes''-instance of \textsc{Exact-3-Cover}.
  Let $J$ be the set of indices describing a solution, meaning that the subsets $S_j \in \mathcal{S}$ with $j \in J$ form an exact  covering of~$U$; clearly, $|J|=n/3$.
  We define $\sigma(i)$ as the unique index $j$ in $J$ for which $u_i \in S_j$. 
  We define a feasible matching~$M$ for $I$ with exactly $b$ blocking pairs as follows (indices take all possible values, if not stated otherwise).
  \begin{equation*}
    \begin{array}{lll} 
      M(\hat{y}) &= x, & \\  
      M(t_j) &= p^1_j & \textrm{ if $j \in J$,} \\
      M(t_j) &= q_j  & \textrm{ if $j \notin J$,} \\
      M(\hat{p}^h_j) &= p^{h+1}_j & \textrm{ if $j \in J$, $h \in \{1,2\}$,} \\      
      M(\hat{p}^3_j) &= s_j   & \textrm{ if $j \in J$,} \\      
      M(c_i) &= b_{i,\sigma(i)}, & \\
      M(\hat{b}_{i,\sigma(i)}) &= a_{i,\sigma(i)}, & \\      
      M(\hat{w}) &= w  & \textrm{ if $w \in \mathcal{W}$ and neither $M(w)$ nor $M(\hat{w})$ is defined yet.} 
    \end{array}
  \end{equation*}
  It is easy to check that $M$ indeed is feasible, and the blocking pairs it admits are exactly the pairs $\{ y, \hat{y}\}$, 
  $\{ b_{i,\sigma(i)}, \hat{b}_{i,\sigma(i)}\}$ for each $i \in \{1,\dots, n\}$,
  $\{ p^1_j, \hat{p}^1_j\}$ for each $j \in J$, 
  $\{ q_j, \hat{q}_j\}$ for each $j \notin J$, and
  $\{ s_j, \hat{y}\}$ for each $j \notin J$.
  This proves the lemma. 
\end{proof}

Contrasting Theorem~\ref{thm:smc2-npc-1woman}, we establish fixed-parameter tractability of the case $\Delta_{\mathcal{W}} \leq 2$ with three different parameterizations.
Considering our five parameters, the relevant cases (whose tractability or intractability does not follow from our results obtained so far) are as follows, assuming $\Delta_{\mathcal W} \leq 2$ throughout. 
Since letting the number $b$ of blocking pairs and
the number $|\mathcal{W}^{\star}|$ of distinguished women to be unbounded 
(while assuming $\Delta_{\mathcal{W}} \leq 2$) results in $\mathsf{NP}$-hardness by Theorem~\ref{thm-case23}, in order to obtain fixed-parameter tractability, we need to take either $b$ or $|\mathcal{W}^{\star}|$ as a parameter. 
However, taking \emph{only} $|\mathcal{W}^{\star}|$ as a parameter is not likely to result in tractability, as the case $|\mathcal{W}^{\star}|=1$ is still $\mathsf{NP}$-hard by Theorem~\ref{thm:smc2-npc-1woman}.
Thus we need to take either $|\mathcal{M}^{\star}|$ or $\Delta_{\mathcal{M}}$ as an additional parameter.
Altogether, this results in the following parameterizations of the {\sc SMC} problem, each case subject to the assumption $\Delta_{\mathcal{W}} \leq 2$: 
\begin{itemize}
  \item taking $b$ as the parameter,
  \item taking $|\mathcal{W}^{\star}|+|\mathcal{M}^{\star}|$ as the parameter, and
  \item taking $|\mathcal{W}^{\star}|+\Delta_{\mathcal{M}}$ as the parameter.
\end{itemize}

We show fixed-parameter tractability for each three of these parameterizations.
The first two parameterizations can be dealt with an algorithm whose properties are stated in Theorem~\ref{thm:fpt-case2x} (see also Corollary~\ref{thm:fpt-case2xCor}), while the third parameterization will be considered by Theorem~\ref{thm:fpt-case2x-extended}.

\begin{theorem}
\label{thm:fpt-case2x}
  There is a fixed-parameter algorithm for the special case of {\sc SMC} where each woman finds at most two men acceptable (i.e., $\Delta_{\mathcal{W}} \leq 2$), parameterized by the number~$|\mathcal{W}^{\star}_0|+|\mathcal{M}^{\star}_0|$ of distinguished men and women left unmatched by some stable matching.
\end{theorem}

Let $M^{\textup{opt}}$ denote an optimal solution for our instance $I$ such that~$M_s \triangle M^{\textup{opt}}$ contains the minimum number of edges; recall that~$M_s$ is a fixed stable matching for~$I$. 
Without loss of generality, we can further assume that there does not exist another optimal solution~$M'$ such that (i) $M'$ has the same number of common edges with $M_s$ as $M^{\textup{opt}}$, and (ii) for each man~$m$, either $M'(m)=M^{\textup{opt}}(m)$ or $m$ prefers $M'(m)$ to $M^{\textup{opt}}(m)$.\footnote{
  Note that there may be optimal solutions that are incomparable to $M^{\textup{opt}}$ in the sense that they are preferred by some of the men while not preferred by some other men, but the existence of such a matching is of no importance to us: we can pick any optimal matching that fulfils our requirement above.
  } Indeed, as long as such a ``superior'' matching $M'$ exists, we can simply pick that instead of~$M^{\textup{opt}}$, until this is no longer possible. This way, we eventually end up with an optimal matching that satisfies our requirement.
  
We denote by~$b$ the number of blocking pairs in~$M^{\textup{opt}}$. 

\medskip
{\bf High-level description.}
Let us remark first that simply guessing the optimal partners for each woman in $\mathcal{W}^{\star}_0$ and then using the polynomial-time algorithm presented in Section~\ref{sect:polycase} (after exchanging the roles of men and women) does not work, since that algorithm heavily relies on the assumption that we start with a stable matching. 
In fact, the main difficulty to overcome is that feminine and masculine augmenting paths may ``interact'' in the sense that certain blocking pairs introduced by a feminine augmenting path can be ``eliminated'' 
(i.e., made non-blocking again) by an appropriately chosen masculine path.
Therefore, we apply the following strategy.
In Phase I, we find all feminine paths (as well as all cycles) in $M_s \triangle M^{\textup{opt}}$, and in Phase II we proceed with choosing the masculine paths carefully. 
Note that in Phase I it does not suffice to find a cheapest set of feminine augmenting paths, since we may not be able to eliminate as many blocking pairs afterwards as it is possible after an optimal choice of feminine paths. 
Instead, we need to find the exact feminine augmenting paths (and cycles) present in $M_s \triangle M^{\textup{opt}}$; this can be accomplished by guessing certain properties of~$M^{\textup{opt}}$.

In Phase II, the main obstacle is that we do not know which blocking edges should be eliminated in an optimal solution, nor can we guess these edges efficiently. 
We deal with this problem by guessing the sets of those men in $\mathcal{M}^{\star}_0$ whose augmenting paths in $M_s \triangle M^{\textup{opt}}$ contribute to the elimination of a blocking pair; 
this information allows us to find these masculine paths.
Finally,
we apply the algorithm of Theorem~\ref{thm:restricted-matching-poly}. 

\smallskip
For the detailed description of our algorithm, we need a couple of simple observations and some additional notation.
We begin with the following implications of the fact that each woman finds at most two men acceptable. 

\begin{proposition} 
\label{prop:2xpaths}
  Suppose $\Delta_{\mathcal{W}} \leq 2$. Let $P_1$ and $P_2$ be two augmenting paths.
  \begin{enumerate}
    \item[(a)] If $P_1$ and $P_2$ start at some $w \in \mathcal{W}^{\star}_0$ through the same edge, then one of them is a subpath of the other.
    \item[(b)] If $P_1$ and $P_2$ start at different women $w_1$ and $w_2$, respectively, and $P_1$ and $P_2$ are not disjoint, then the set of their common vertices induces a suffix of either $P_1$ or $P_2$ (or both); their first common vertex is a man.
    \item[(c)] If $P_1$ and $P_2$ are disjoint and $e$ is an edge incident to both, then one of the paths starts or ends at a women $w$, and $e$ connects $w$ with a man on the other path. 	
  \end{enumerate}
\end{proposition}

Observations (a) and (b) of Proposition~\ref{prop:2xpaths} immediately suggest that feminine paths are easy to find, since once we decide which edge to start with, all  possible augmenting paths lead in the same direction; the only difficulty arises in deciding when to stop. Observation (c) describes the limited ways in which two augmenting path can interact; we next look closer at such interactions.


We say that an edge $f=\{m,w\}$ of $G$ (with $m \in \mathcal{M}$ and $w \in \mathcal{W}$) 
is \emph{dependent} if it connects two different connected components $K_1$ and $K_2$ of $M_s \triangle M^{\textup{opt}}$ and, in addition, it holds that~$M_s \triangle K_1$ admits more blocking pairs than $M_s \triangle (K_1 \cup K_2)$.
We will say that $f$, and with a slight abuse of the notation, also $K_1$ \emph{relies} on~$K_2$.
We say that 
\begin{itemize} 
  \item $f$ has \emph{type A}, if $w$ is the endpoint of $K_2$ (which is a path), $f$ connects~$w$ with a man~$m$ on $K_1$ that prefers $M_s(m)$ to $w$, and $w$ to $M^{\textup{opt}}(m)$, and~$w$ is unmatched by $M_s$ and prefers $M^{\textup{opt}}(w)$ to $m$;
  \item $f$ has \emph{type B}, if $w$ is the endpoint of $K_1$ (which is a path), unmatched by $M^{\textup{opt}}$, and~$f$ connects~$w$ with a man $m$ on $K_2$ that prefers $M^{\textup{opt}}(m)$ to~$w$, and $w$ to $M_s(m)$. 
\end{itemize}
See Figure~\ref{fig:dependentedges} for an illustration of the above definitions.

\begin{figure}[thpb]
    \tikzset{rdvertex/.style={minimum size=2mm,circle,fill=white,draw, inner sep=0pt},
        sqvertex/.style={minimum size=2mm,diamond,fill=black,draw, inner sep=0pt},
        opvertex/.style={minimum size=2mm,diamond,fill=white,draw, inner sep=0pt},
        decoration={markings,mark=at position .5 with {\arrow[black,thick]{stealth}}}}
        \centering
        \begin{tikzpicture}
        \def\h{0.8}
        \node (typea) at (\h,\h){type A:};
        \node (k1) at (0,0){$K_1:$};
        \node (k2) at (0,-1) {$K_2:$};
        \node (a1) at (\h,0){$\hdots$};
        \node (a2) at (2*\h,0)[opvertex]{};
        \node (a3) at (2*\h+1,0)[rdvertex,label=above:$m$]{};
        \node (a4) at (2*\h+2,0)[opvertex]{};
        \node (a5) at (2*\h+1,-1)[opvertex,label=below:$w$]{};
        \node (a5-1) at (2*\h+1-0.2,-0.8){\footnotesize $2$};
        \node (a5-2) at (2*\h+1+0.2,-1.2){\footnotesize $1$};
        \node (a6) at (2*\h+2,-1)[rdvertex]{};
        \node (a7) at (2*\h+3,-1)[opvertex]{};
        \node (a8) at (3*\h+3,-1){$\hdots$};
        \node (a9) at (3*\h+2,0){$\hdots$};
        \node (af) at (2*\h+1.15,-0.5){$f$};
        \draw[double] (a1)--(a2);
        \draw (a2)--(a3);
        \draw[double] (a3)--(a4);
        \draw (a4)--(a9);
        \draw[dashed] (a3)--(a5);
        \draw (a5)--(a6);
        \draw[double] (a6)--(a7);
        \draw (a7)--(a8);
        \path (2*\h+0.70,-0.05) edge[bend right,->] node [left] {} (2*\h+0.95,-0.25);
        \path (2*\h+1.05,-0.25) edge[bend right,->] node [left] {} (2*\h+1.30,-0.05);
        \draw (5.8,1.25)--(5.8,-1.5);
        \def\s{6.3}
        \node (typeb) at (\s+\h,1){type B:};
        \node (k1) at (\s,0){$K_1:$};
        \node (k2) at (\s,-1){$K_2:$};
        \node (b1) at (\s+\h,0){$\hdots$};
        \node (b2) at (\s+2*\h,0)[opvertex]{};
        \node (b3) at (\s+2*\h+1,0)[rdvertex]{};
        \node (b4) at (\s+2*\h+2,0)[opvertex,label=right:$w$]{};
        \node (b4-1) at (\s+2*\h+1.8,0.2){\footnotesize $1$};
        \node (b4-2) at (\s+2*\h+2.2,-0.2){\footnotesize $2$};
        \node (b5) at (\s+2*\h+2,-1)[rdvertex,label=below:$m$]{};
        \node (b6) at (\s+2*\h+1,-1)[opvertex]{};
        \node (b7) at (\s+2*\h+3,-1)[opvertex]{};
        \node (b8) at (\s+2*\h,-1){$\hdots$};
        \node (b9) at (\s+3*\h+3,-1){$\hdots$};
        \node (bf) at (\s+2*\h+1.85,-0.5){$f$};
        \draw[double] (b1)--(b2);
        \draw (b2)--(b3);
        \draw[double] (b3)--(b4);
        \draw[dashed] (b4)--(b5);
        \draw (b5)--(b6);
        \draw[double] (b5)--(b7);
        \draw[double] (b6)--(b8);
        \draw (b7)--(b9);
        \path (\s+2*\h+1.95,-0.75) edge[bend right,->] node [left] {} (\s+2*\h+1.70,-0.95);
        \path (\s+2*\h+2.30,-0.95) edge[bend right,->] node [left] {} (\s+2*\h+2.05,-0.75);
        \end{tikzpicture}
    \caption{Illustration of a dependent edge $f$, running between two connected 
        components~$K_1$ and~$K_2$  of $M_s \triangle M^{\textup{opt}}$ where $K_1$ relies on $K_2$.
        Double lines here 
        denote edges of~$M_s$, 
        single lines 
        denote edges of~$M^{\textnormal{opt}}$, 
        and~$f$ is drawn with a dashed line. 
        \label{fig:dependentedges}}
\end{figure}

\begin{lemma}
\label{lemma-type}
  Let $\mathcal K$ be the set of connected components of~$M_s \triangle M^{\textup{opt}}$, and let $H \subset \mathcal{K}$ and $K \in \mathcal{K} \setminus H$.
  Then any edge $f$ that is blocking in $M_s \triangle H$ but is not blocking in $M_s \triangle (H \cup K)$ connects $K$ with a connected component of $H$, is a dependent edge relying on $K$, and has either type~A or type~B.
\end{lemma}
\begin{proof}
  Let $f$ be an edge that blocks $M_s \triangle H$ but does not block $M_s \triangle (H \cup K)$.
  Since $M_s$ is stable, $f$ must have an endpoint in a connected component of $H$, because it blocks $M_s \triangle H$.
  However, as it ceases to be blocking in $M_s \triangle (H \cup K)$, it also must have an endpoint in $K$.
  Let $w$ and $m$ be the woman and the man connected by $f$, respectively.
  We distinguish now between two cases.

  First, let us assume that $w$ is contained in $K$.
  Since $m$ is not contained in~$K$, and $w$ can only be connected to two men, it follows that the degree of~$w$ in $K$ is at most (and hence exactly) 1.
  Since each connected component of $M_s \triangle M^{\textup{opt}}$ is either a cycle or a path, this implies that $K$ is a path with $w$ being one of its endpoints.

  Since $f$ blocks $M_s \triangle H$ and the partner of $m$ under $M_s \triangle H$  is $M^{\textup{opt}}(m)$, we know that $m$ prefers $w$ to $M^{\textup{opt}}(m)$. 
  But as $M_s$ is stable, $m$ must prefer $M_s(m)$ to $w$. 

  Notice now that $w$ is unmatched either in $M_s$ or in $M^{\textup{opt}}$.
  However, as~$f$ blocks $M_s \triangle H$ (where $w$ is matched as in $M_s$) but does not block $M_s \triangle (H \cup K)$ (where $w$ is matched as in $M^{\textup{opt}}$), it must be the case that~$w$ is matched by $M^{\textup{opt}}$ but is unmatched in $M_s$.
  Further, as $f$ is not blocking in $M_s \triangle (H \cup K)$ even though $m$ prefers $w$ to $M^{\textup{opt}}(m)$, we get that $w$ prefers $M^{\textup{opt}}(w)$ to $m$.
  This proves that $f$ has type~A.

  Second, let us assume that $w$ is contained in a connected component $K'$ of $H$.
  As in the previous case, we quickly get that $K'$ must be a path, with $w$ being an endpoint of $K'$.
  Since~$f$ blocks $M_s \triangle H$, we know that $w$ is matched by $M_s$ but is unmatched by $M^{\textup{opt}}$.
  The stability of $M_s$ implies also that $w$ prefers $M_s(w)$ to $m$. 

  Regarding $m$, the fact that $f$ blocks $M_s \triangle H$ implies that $m$ prefers $w$ to~$M_s(m)$.
  However, since $f$ ceases to be blocking in $M_s \triangle (H \cup K)$, we get that~$m$ prefers $M^{\textup{opt}}(m)$ to $w$.
  This proves that $f$ has type~B. 
\end{proof}


We are now ready to present our algorithm, which is a branching algorithm:
throughout its course, we make several ``guesses'' for which all possibilities have to be explored.
When certain guesses turn out to be trivially wrong, such guesses are discarded, and we might not explicitly mention this in the algorithm.
(In Step 1, we describe such issues in detail for illustration, but later we omit them.)
Phases~I and II consist of Steps 1-5 and Steps 6-8, respectively.

\medskip
\noindent
{\bf Step 1: Guessing the first edges of feminine augmenting paths.}
First, for each $w \in \mathcal{W}^{\star}_0$ with $|L(w)|=2$, we guess the edge of $M^{\textup{opt}}$ incident to~$w$.
This results in at most $2^{|\mathcal{W}^{\star}_0|}$ possibilities, all of which must be explored. 
Naturally, we discard those guesses where the edges $\{w,M^{\textup{opt}}(w)\}$, $w \in \mathcal{W}^{\star}_0$, do not form a matching. 
From now on we assume that we know $M^{\textup{opt}}(w)$ for each $w \in \mathcal{W}^{\star}_0$.

Additionally, we delete those edges $\{m,w\}$ for which $w \in \mathcal{W}^{\star}_0$
and $w$ prefers $M^{\textup{opt}}(w)$ to~$m$. 
Such edges are neither needed in~$M^{\textup{opt}}$, nor can they block any matching that contains all the 
edges $\{w,M^{\textup{opt}}(w)\}$, $w \in \mathcal{W}^{\star}_0$, guessed in this step. 

\medskip

Before proceeding to Step 2, we state an important lemma about augmenting paths.
\begin{lemma}
  \label{lem:non-aug-paths-case2x}
  Each connected component of $M_s \triangle M^{\textup{opt}}$ that is not a cycle is an  augmenting path. 
  Further, assume that Step~1 has already been performed, and~$K_1$ and $K_2$ are connected components of $M_s \triangle M^{\textup{opt}}$ such that~$K_1$ relies on~$K_2$ via a dependent edge $f$.
  Then 
  \begin{itemize}
    \item[(a)] if $f$ has type A, then $K_2$ is a masculine path and not a feminine path;
    \item[(b)] if $f$ has type B, then $K_1$ is a feminine path and not a masculine path, and $K_2$ is either a cycle or a feminine path.
  \end{itemize}
\end{lemma}
\begin{proof}
  We begin by proving the first sentence of the lemma.
  Let us first suppose that $Q$ is a connected component of $M_s \triangle M^{\textup{opt}}$ that is not a cycle (thus is a path) but is not an augmenting path.
  The feasibility of~$M^{\textup{opt}}$ implies that if $Q$ has a distinguished person~$p$ as its endpoint, then $p$ must be unmatched by~$M_s$.
  This means that~$Q$~can only be non-augmenting if neither of its endpoints is distinguished. 
  This implies that $M_Q:=M^{\textup{opt}} \triangle Q$ is a feasible matching. 
  Recall that~$b$ is the number of blocking pairs~$M^{\textup{opt}}$ admits.
  If $M_Q$ admits at most~$b$ blocking pairs as well, then this contradicts the choice of $M^{\textup{opt}}$, because there are strictly less edges in $M_s \triangle M_Q$ than in~$M_s \triangle M^{\textup{opt}}$.

  Hence, $M^{\textup{opt}} \triangle Q$ admits at least $b+1$ blocking pairs.
  Since $Q$ is a maximal path in $M_s \triangle M^{\textup{opt}}$ and~$M_s$ is stable, Proposition~\ref{prop:alt-paths-blocking-pairs} implies that 
  there must be an edge along $Q$ that blocks $M^{\textup{opt}}$. 
  Obviously, there is no edge on $Q$ that blocks~$M_s$, since $M_s$ is stable. 
  Hence, modifying~$M^{\textup{opt}}$ by switching the edges of~$M_s$ and~$M^{\textup{opt}}$ along~$Q$ \emph{decreases} the number of blocking edges 
  \emph{among the edges of} $Q$.
  However, the same operation \emph{increases} the \emph{total} number of blocking pairs (from $b$ to at least $b+1$), thus we get that there must exist at least two edges that become blocking only as a result of this switch along $Q$. 
  In other words, there exist at least two edges that are blocking in $M^{\textup{opt}} \triangle Q$ but are not blocking in $M^{\textup{opt}}$.  
  By Lemma~\ref{lemma-type} we obtain that these edges must be dependent edges relying on~$Q$.\footnote{Formally, we need to apply Lemma~\ref{lemma-type} with setting $H$ to contain all connected component of $M_s \triangle M^{\textup{opt}}$ except for $Q$, and setting $K$ as $Q$.}
  By the definition of a type~A edge, at most one type~A edge can rely on $Q$ (as only one endpoint of $Q$ can be a woman unmatched by~$M_s$, and this woman may be incident to at most one edge outside $Q$), so we get that at least one of these dependent edges relying on $Q$ must have type~B.

  Let us call the man endpoint of a type~B dependent edge a \emph{joiner}; 
  by the previous paragraph,~$Q$ contains at least one joiner. 
  Our aim is to fix an ``outer-most'' joiner $m$ on $Q$. 
  However, for technical reasons we also have to take into account a special case where some man $m_\ell$ lying on $Q$, 
  called the \emph{looper} for~$Q$, fulfils the following properties: 
  (i) one endpoint of $Q$ is a woman~$w_\ell$ that is adjacent to~$m_\ell$ in $G$ but prefers $M_s(w_\ell)$ to $m_\ell$, and 
  (ii) $m_\ell$ prefers $w_\ell$ to $M_s(m_\ell)$ but prefers $M^{\textup{opt}}(m_\ell)$ to $w_\ell$. 
  \footnote{Note that there cannot exist two loopers for $Q$, 
  because only one woman endpoint of $Q$ can be matched by $M_s$ and such a woman can only be connected 
  to at most one man on $Q$ other than its partner in $M_s$.}
  Now, we choose a man $m$ that is either a joiner or the looper for $Q$ so that the following holds: if $m$ splits~$Q$ into two subpaths~$Q_1$ and~$Q_2$ with $Q_1$ containing~$M^{\textup{opt}}(m)$, then $V(Q_1) \setminus \{m\}$ contains neither joiners nor the looper for $Q$. 
  
  {\bf Case for a joiner.} First, let us assume that $m$ is a joiner. 
  In this case, there might be several women who form a dependent edge of type B with $m$, so let $w$ denote the one that is most preferred by~$m$.
  Let~$f$ be the edge $\{m,w\}$, and let $P$ be the path of~$M_s \triangle M^{\textup{opt}}$ that has $w$ as its endpoint.
  We illustrate these concepts in Figure~\ref{fig:outermost-joiners}.
  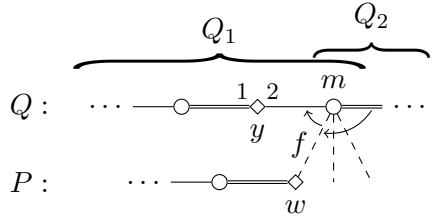
\begin{figure}[thpb]
    \tikzset{rdvertex/.style={minimum size=2mm,circle,fill=white,draw, inner sep=0pt},
        sqvertex/.style={minimum size=2mm,diamond,fill=black,draw, inner sep=0pt},
        opvertex/.style={minimum size=2mm,diamond,fill=white,draw, inner sep=0pt},
        decoration={markings,mark=at position .5 with {\arrow[black,thick]{stealth}}}}
        \centering
        \begin{tikzpicture}
        \node (q) at (0,0){$Q:$};
        \node (p) at (0,-1){$P:$};
        \node (q1) at (1,0){$\hdots$};
        \node (q2) at (2,0)[rdvertex]{};
        \node (q3) at (3,0)[opvertex,label=below:$y$]{};
        \node (q3-1) at (2.8,0.2){\footnotesize $1$};
        \node (q3-2) at (3.2,0.2){\footnotesize $2$};
        \node (q4) at (4,0)[rdvertex,label=above:$m$]{}; 
        \node (q5) at (5,0){$\hdots$};
        \node (p1) at (1.5,-1){$\hdots$};
        \node (p2) at (2.5,-1)[rdvertex]{};
        \node (p3) at (3.5,-1)[opvertex,label=below:$w$]{};
   		\draw[decoration={brace,amplitude=6pt,raise=10pt},decorate]  (q1.north west) -- node[midway,above=15pt] {$Q_1$} (q4.north east); 		
   		\draw[decoration={brace,amplitude=6pt,raise=16pt},decorate]  (q4.north west) -- node[midway,above=21pt] {$Q_2$} (q5.north east); 		
        \node (f) at (3.55,-0.5){$f$};
        \draw (q1)--(q2);
        \draw[double] (q2)--(q3);
        \draw (q3)--(q4);
        \draw[double] (q4)--(q5);
        \draw (p1)--(p2);
        \draw[double] (p2)--(p3);
        \draw[dashed] (q4)--(p3);
        \draw[dashed] (q4)--(4,-1);
        \draw[dashed] (q4)--(4.5,-1);
        \path (3.85,-0.25) edge[bend left,->] node [right] {} (3.65,-0.05);
        \path (4.5,-0.05) edge[bend left,->] node [right] {} (3.87,-0.35); 
        \end{tikzpicture}
    \caption{The joiner $m$ splits $Q$ into subpaths $Q_1$ and $Q_2$. 
        Here and in later figures, double lines denote edges of $M_s$. 
        Single lines denote edges of $M^{\textup{opt}}$, 
        and dashed lines are for dependent edges.        \label{fig:outermost-joiners}}
\end{figure}

  We claim that $M_f=M^{\textup{opt}} \triangle (Q_1 \cup \{f\}) = (M^{\textup{opt}} \triangle Q_1) \cup \{f\}$ is an optimal solution. 
  Observe that~$M_s \triangle M_f$ can be obtained from $M_s \triangle M^{\textup{opt}}$ by deleting~$Q$ and substituting~$P$ by the path $P+f+Q_2$ where the plus sign means concatenation. 
  Let~$x$ denote the endpoint of $Q_1$ that is not~$m$.
  First, $M_f$ is clearly feasible, since $x$ is not distinguished (by our assumption that $Q$ is non-augmenting).
  Next, suppose that some edge $e$ is blocking in~$M_f$ but is not blocking in~$M^{\textup{opt}}$. 
  It is easy to see that by our choice of $w$, $e$ cannot be incident to $w$ or $m$. 
  As only vertices in $Q_1 \cup \{f\}$ are matched differently in~$M_f$ as in $M^{\textup{opt}}$,
we obtain that one endpoint of $e$ must lie on $Q_1$.
  Using the stability of $M_s$, we also get that the other endpoint of $e$ must lie either on a connected component of $M_s \triangle M^{\textup{opt}}$ other than~$Q$ or on $Q_2$. 
  In the former case, Lemma~\ref{lemma-type} implies\footnote{
  Again, we need to apply Lemma~\ref{lemma-type} with setting $H$ to contain all connected component of $M_s \triangle M^{\textup{opt}}$ except for $Q$, and setting $K$ as $Q$.} that $e$ is a dependent edge. By our choice of~$m$, $e$ cannot be of type~B, hence by Lemma~\ref{lemma-type}, $e$ must have type~A, and thus it is incident to $x$ which must be a woman not covered by~$M_s$. 
  In the latter case, supposing that the man endpoint $m'$ of $e$ lies on $Q_1$ we quickly get that it must be the looper for $Q$.
  To see this, first note that the only woman that~$m'$ can be adjacent to on $Q_2$ must be the common endpoint of $Q_2$ and $Q$ (by $\Delta_{\mathcal W} \leq 2$); let $w'$ denote this woman.
  Second, since~$e$ blocks $M_f$ but it does not block $M^{\textup{opt}}$, we know that~$m'$ prefers~$w'$ to $M_s(m')$, but prefers $M^{\textup{opt}}(m')$ to $w'$.
  Third, since~$m'$ prefers~$w'$ to $M_s(m')$, the stability of $M_s$ implies that~$w'$ is matched by~$M_s$ and~$w'$ prefers $M_s(w')$ to $m'$. 
  Hence, $m'$ is indeed the looper for $Q$, contradicting our choice of~$m$.
  Thus, we know that it must be the woman endpoint of $e$ that lies on $Q_1$.
  Since the two endpoints of $Q_1$ are $m \in \mathcal{M}$ and~$x$, by $\Delta_{\mathcal{W}} \leq 2$ we obtain that the woman endpoint of $e$ lying on~$Q_1$ can only be~$x$, which therefore must be a woman not covered by~$M_s$.
  In either case, we can conclude that there can only exist at most one such edge~$e$ (because~$x$ has degree at most $2$ in $G$).
  So the number of edges that are blocking in~$M_f$ but not blocking in~$M^{\textup{opt}}$ is at most one. 
  
  Furthermore, since $f$ is a dependent edge of type~B, by definition we know that~$m$ prefers $y=M^{\textup{opt}}(m)$ to $M_s(m)$, and hence, $y$ must prefer $M_s(y)$ to~$m$ (as otherwise $\{m,y\}$ would be blocking in~$M_s$, which is not possible).
  Note that, in particular, $M_s(y)$ exists and $Q_1$ contains at least two edges. 
  By Proposition~\ref{prop:alt-paths-blocking-pairs},~$Q_1$ must contain at least one edge that blocks~$M^{\textup{opt}}$, and this edge is not blocking in~$M_f$, simply because any edge of $Q_1$ that blocks $M^{\textup{opt}}$ but be an edge contained in $M_s$ and hence in $M_f$.
  Thus, the number of edges blocking~$M_f$ cannot be more than $b$.
  Hence, $M_f$ is an optimal solution such that there are less edges in $M_s \triangle M_f$ than in $M_s \triangle M^{\textup{opt}}$, a contradiction.
  This proves the first statement of the lemma for the case of~$m$ being a joiner.

 \begin{figure}[thpb]
    \tikzset{rdvertex/.style={minimum size=2mm,circle,fill=white,draw, inner sep=0pt},
        sqvertex/.style={minimum size=2mm,diamond,fill=black,draw, inner sep=0pt},
        opvertex/.style={minimum size=2mm,diamond,fill=white,draw, inner sep=0pt},
        decoration={markings,mark=at position .5 with {\arrow[black,thick]{stealth}}}}
        \centering
        \begin{tikzpicture}
        \node (q) at (0,0){$Q:$};
        \node (q1) at (1,0){$\hdots$};
        \node (q2) at (2,0)[rdvertex]{};
        \node (q3) at (3,0)[opvertex]{};
        \node (q3-1) at (2.8,0.2){\footnotesize $1$};
        \node (q3-2) at (3.2,0.2){\footnotesize $2$};
        \node (q4) at (4,0)[rdvertex,label=above:$m$]{}; 
        \node (q5) at (5,0){$\hdots$};
        \node (q6) at (6,0)[opvertex,label=right:$w_{\ell}$]{};
        \node (q6-1) at (5.8,0.2){\footnotesize $1$};
        \node (q6-2) at (5.9,-0.3){\footnotesize $2$};
   		\draw[decoration={brace,amplitude=6pt,raise=10pt},decorate]  (q1.north west) -- node[midway,above=14pt] {$Q_1$} (q4.north east); 		
   		\draw[decoration={brace,amplitude=6pt,raise=16pt},decorate]  (q4.north west) -- node[midway,above=21pt] {$Q_2$} (q6.north east); 		
        \node (f) at (5,-0.8){$f$};
        \draw (q1)--(q2);
        \draw[double] (q2)--(q3);
        \draw (q3)--(q4);
        \draw[double] (q4)--(q5);
        \draw[double] (q5)--(q6);
	    \draw [dashed] (q4) to [out=310,in=230] (q6);
        \path (4.22,-0.28) edge[bend left,->] node [right] {} (3.65,-0.05);
        \path (4.5,-0.05) edge[bend left,->] node [right] {} (4.35,-0.3); 
        \end{tikzpicture}
    \caption{The looper $m$ of $Q$ splits $Q$ into subpaths $Q_1$ and $Q_2$. 
	\label{fig:looper}}
\end{figure}

  {\bf Case for a looper.} 
  Let us now assume that $m=m_\ell$ is the looper for~$Q$.
  In this case let $f$ be the edge connecting $m_\ell$ to the woman endpoint $w_\ell$ of~$Q$ that is matched by $M_s$; it is easy to see that $w_\ell$ is the endpoint of $Q_2$ as well.
  See Fig.~\ref{fig:looper} for an illustration.
  Let us again define $M_f=M^{\textup{opt}} \triangle (Q_1 \cup \{f\})$.
  As in the previous case, it is immediate that $M_f$ is feasible.
  Arguing similarly as before and using that~$Q_1$ does not contain any joiners, we get that any edge that blocks $M_f$ but does not block~$M^{\textup{opt}}$ can only be a type~A dependent edge incident to the (woman) endpoint~$x$ of $Q_1$ that is not $m_\ell$. 
  Hence, there can be at most one edge blocking~$M_f$ but not $M^{\textup{opt}}$.
  From this point on, we can use the same reasoning as in the previous case to arrive at the conclusion that~$M_f$ is an optimal solution that has more edges common with $M_s$ than $M^{\textup{opt}}$, a contradiction.
  This proves the first statement of our lemma.

  Let us prove (b) now.
  Suppose that $K_1$ and $K_2$ are two connected components of $M_s \triangle M^{\textup{opt}}$ such that~$K_1$ relies on $K_2$ via an edge $f=\{w,m\}$ of type~B. 
  By the definition of a type B edge, we know that $K_1$ is a path with an endpoint that is a woman covered by $M_s$, so its other endpoint is either a woman not covered by $M_s$, or a man covered by $M_s$.
  Thus, $K_1$ cannot be a masculine path,  so by the first statement of the lemma, it is feminine.
  It remains to show that if $K_2$ is a path, then it is feminine. 
  Assume for contradiction that $K_2$ is a non-feminine path~$Q$ in~$M_s \triangle M^{\textup{opt}}$ and some other path $P$ in $M_s \triangle M^{\textup{opt}}$ relies on $Q$ via a type~B edge $f$.
  In this case we can argue exactly as above to show that there must exist an optimal matching~$M_f$ (defined the same way as we did while proving the first statement of the lemma) for which $M_s \triangle M_f$ contains less edges than $M_s \triangle M^{\textup{opt}}$, a contradiction. 

  To show (a), suppose that $K_1$ and $K_2$ are two connected components of $M_s \triangle M^{\textup{opt}}$ such that~$K_1$ relies on $K_2$ via an edge $f=\{w,m\}$ of type~A.
  By the definition of a type~A edge, the woman endpoint~$w$ of $f$ is unmatched in~$M_s$, and $K_2$ is a path that has $w$ as an endpoint. 
  Also,~$m$ is the second choice of~$w$ and $M^{\textup{opt}}(w)$ is the first choice of $w$. 
  Hence $w \in \mathcal{W}^\star_0$ is not possible, as in that case the edge $f$ would have been deleted in Step~1 of the algorithm.
  Furthermore, since~$K_2$ is an $M_s$-alternating path with an endpoint in $\mathcal{W}_0$, the other endpoint of~$K_2$ cannot be a woman in $\mathcal{W}_0$. 
  Hence, $K_2$ is not a feminine augmenting path. 
  Since we already know that~$K_2$ is an augmenting path, (a) follows. 
\end{proof}

By Lemma~\ref{lem:non-aug-paths-case2x}, each connected component in $M_s \triangle M^{\textup{opt}}$ that is a path must be an augmenting path; 
we let~$P_x^{\textup{opt}}$ denote the augmenting path in $M_s \triangle M^{\textup{opt}}$ that contains 
some $x \in \mathcal{W}^{\star}_0 \cup \mathcal{M}^{\star}_0$ as an endpoint. 

Consider now any woman $w \in \mathcal{W}_0^{\star}$. 
Let us define a \emph{candidate path} for~$w$ as an augmenting path that starts through the edge $\{w,M^{\textup{opt}}(w)\}$.
Observe that~$P_w^{\textup{opt}}$ itself is a candidate path for $w$.
Furthermore, although we cannot determine $P_x^{\textup{opt}}$ directly, Proposition~\ref{prop:2xpaths} implies that there exists a unique maximal candidate path $A_w$ for $w$ that contains $P_w^{\textup{opt}}$ as a subpath. 
We can compute~$A_w$ easily: starting from the edge $\{w,M^{\textup{opt}}(w)\}$, we can build $A_w$ by adding edges of $M_s$ and edges not in $M_s$ in an alternating manner, always appending a new edge to the last vertex of the current subpath of $A_w$.
Observe that after an addition of an edge of $M_s$, the resulting subpath of $A_w$ ends at a woman, and thus $\Delta_{\mathcal{W}} \leq 2$ implies that there exists at most one edge that we can add to our path. 
This shows the uniqueness of~$A_w$.

\medskip
\noindent
{\bf Step 2: Finding cycles in }$M_s \triangle M^{\textup{opt}}$\textbf{.}
We make one more guess for each $w \in \mathcal{W}^{\star}_0$ by guessing whether~$P^{\textup{opt}}_w$ relies on some cycle of $M_s \triangle M^{\textup{opt}}$ or not; this again yields $2^{|\mathcal{W}^{\star}_0|}$ possibilities. 
If~$P^{\textup{opt}}_w$ relies on some cycle~$C$, then by Proposition~\ref{prop:2xpaths}, both $P^{\textup{opt}}_w$ and~$C$ can be found in time $O(|P^{\textup{opt}}_w|+|C|)$ as follows.
We compute the longest possible candidate path $A_w$ for $w$.
Since $P^{\textup{opt}}_w$ relies on cycle~$C$, the last person on~$A_w$ must be a woman $x$ incident to an edge $\{x,m\}$ for which $m$ is already on~$A_w$.
Then the subpath of~$A_w$ between~$m$ and $x$ together with the edge~$\{x,m\}$ form the cycle~$C$.
For an illustration, see Fig.~\ref{fig:finding-a-cycle}. 
\begin{figure}[t]
    \tikzset{rdvertex/.style={minimum size=2mm,circle,fill=white,draw, inner sep=0pt},
        sqvertex/.style={minimum size=2mm,diamond,fill=black,draw, inner sep=0pt},
        opvertex/.style={minimum size=2mm,diamond,fill=white,draw, inner sep=0pt},
        decoration={markings,mark=at position .5 with {\arrow[black,thick]{stealth}}}}
        \centering
        \begin{tikzpicture}
        \node (w) at (0,0)[sqvertex,label=left:$w$]{};
        \node (path1) at (1,0)[rdvertex,label=below:$M^{\textnormal{opt}}(w)$]{};
        \node (path2) at (2,0)[opvertex]{};
        \node (hdots1) at (3,0){$\hdots$};
        \node (path3) at (4,0)[rdvertex]{};
        \node (path4) at (5,0)[opvertex]{};
        \node (path4-1) at (4.8,0.2){{\footnotesize $1$}};
        \node (path4-2) at (5.2,-0.2){{\footnotesize $2$}};
        \node (path5) at (5,-1)[rdvertex,label=below:$m$]{};
        \path (4.95,-0.70) edge[bend right,->] node [left] {} (4.70,-0.95);
        \path (5.30,-0.95) edge[bend right,->] node [left] {} (5.05,-0.70);
        \node (path6) at (6,-1)[opvertex]{};
        \node (path7) at (4,-1)[opvertex,label=below:$x$]{};
        \node (path7-1) at (3.8,-0.8){{\footnotesize $1$}};
        \node (path7-2) at (4.2,-0.8){{\footnotesize $2$}};
        \node (hdots2) at (3,-1){$\hdots$};
        \node (path8) at (2,-1)[rdvertex]{};
   		\draw[decoration={brace,amplitude=6pt,raise=8pt},decorate]  (w.north west) -- node[midway,above=14pt] {$P_w^{\textnormal{opt}}$} (path4.north east); 		
        \draw (w)--(path1);
        \draw[double] (path1)--(path2);
        \draw (path2)--(hdots1);
        \draw (hdots1)--(path3);
        \draw[double] (path3)--(path4);
        \draw[dashed] (path4)--(path5);
        \draw[double] (path5)--(path6);
        \draw (path5)--(path7);
        \draw[double] (path7)--(hdots2);
        \draw[double] (hdots2)--(path8);
        \path (path8) edge[bend right] node [left] {} (path6);
        \end{tikzpicture}
    \caption{Step~2 of the algorithm for Theorem~\ref{thm:fpt-case2x}, for finding cycles in $M \triangle M^{\textup{opt}}$.
        \label{fig:finding-a-cycle}}
\end{figure}
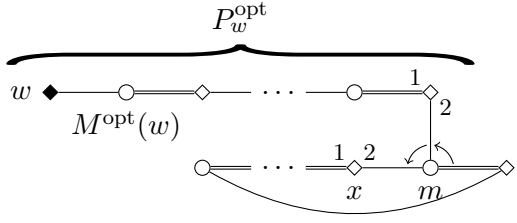

\begin{lemma} 
\label{lem:finding-all-cycles}
  Assuming that the guesses made by the algorithm are correct, all cycles in~$M_s \triangle M^{\textup{opt}}$ and all paths relying on some cycle in~$M_s \triangle M^{\textup{opt}}$ are found in Step~2.
\end{lemma}
\begin{proof}
  Let $C$ be a cycle in $M_s \triangle M^{\textup{opt}}$.
  By the stability of $M_s$, no edge of $C$ blocks the matching $M^{\textup{opt}} \triangle C$.
  Since $M^{\textup{opt}} \triangle C$ has more edges in common with~$M_s$ than~$M^{\textup{opt}}$ does, we know by the choice of $M^{\textup{opt}}$ that there must be an edge incident to $C$ that is blocking in $M^{\textup{opt}} \triangle C$ but is not blocking in~$M^{\textup{opt}}$.
  By Lemma~\ref{lemma-type}, such an edge must be a dependent edge $e$ relying on $C$.
  This proves that there are no cycles in $M_s \triangle M^{\textup{opt}}$ without augmenting paths relying on them.

  By Lemma~\ref{lem:non-aug-paths-case2x} we know that dependent edges of type~A rely on (masculine) paths\footnote{ 
  We remark that it is still possible that $C$ itself relies on a masculine path via a dependent edge of type~A, but such edges do not play a role in Step 2.}, so $e$ must have type~B.
  From this, Lemma~\ref{lem:non-aug-paths-case2x} yields that all paths relying on a cycle must be feminine paths. 
  Hence, Step~2 indeed finds all paths that rely on some cycle, together with the cycles that these paths rely upon, which (as pointed out above) means that Step~2 finds all cycles of $M_s \triangle M^{\textup{opt}}$, proving our claim.
\end{proof} 

\medskip
\noindent
{\bf Step 3: Finding neutral paths.}
In this step, for each $w \in \mathcal W^{\star}_0$ we guess whether $w$ lies on a neutral path in $M_s \triangle M^{\textup{opt}}$, yielding at most
$2^{|\mathcal{W}_0^\star|}$ possibilities.
Clearly, if $w$ lies on a neutral path, then by Proposition~\ref{prop:2xpaths}, $P_w^{\textup{opt}}$ is the maximal candidate path starting with the edge $\{w,M^{\textup{opt}}(w)\}$, that is, $P_w^{\textup{opt}}=A_w$. 

\medskip
\noindent
{\bf Step 4: Finding feminine paths relying on other feminine paths.}
In this step, we find all feminine paths that rely on some other feminine path; 
note that feminine paths that do not depend on other feminine paths will be taken care of by Step~5.

In Step 4, we first guess for each $w \in \mathcal{W}^{\star}_0$ whether $P_w^{\textup{opt}}$
relies on another feminine path, and if so, we also guess on which one. 
This means at most $|\mathcal{W}^{\star}_0|$ possibilities for each $w \in \mathcal{W}^{\star}_0$, 
a total of $|\mathcal{W}^{\star}_0|^{|\mathcal{W}^{\star}_0|}$ possibilities. 

Supposing that, according to our guesses, $P_w^{\textup{opt}}$ relies on $P_y^{\textup{opt}}$ for some~$w$ and $y$ in $\mathcal{W}^{\star}_0$, we can find~$P_w^{\textup{opt}}$ easily as follows. 
Consider the maximal candidate paths $A_w$ and~$A_y$ for $w$ and $y$, respectively.
We know that $P_w^{\textup{opt}}$ and $P_y^{\textup{opt}}$ are disjoint, but supposing that our guesses were made correctly, Lemma~\ref{lem:non-aug-paths-case2x} implies that they are connected by a dependent edge~$f$ of type~B relying on~$P_y^{\textup{opt}}$.
Hence, $A_w$ contains~$P_w^{\textup{opt}}$, followed by the dependent edge~$f$ connecting~$P_w^{\textup{opt}}$ and $P_y^{\textup{opt}}$, followed possibly by a suffix of $P_y^{\textup{opt}}$. 
Thus, by Proposition~\ref{prop:2xpaths},~$A_w$ and $A_y$ must share the man endpoint $m$ of $f$ as their first common vertex, and they coincide after~$m$. 
Therefore, $P_w^{\textup{opt}}$ can be obtained by deleting all edges incident to or occurring after $m$ on $A_w$.
This way, $f$ and~$P_w^{\textup{opt}}$ can be found in $O(|P_w^{\textup{opt}}|+|P_y^{\textup{opt}}|)$ time.

Notice that this method only enables us to locate $P_w^{\textup{opt}}$ but not $P_y^{\textup{opt}}$; we only know that $P_y^{\textup{opt}}$ must contain the man $m$ incident to $f$, but we do not know where $P_y^{\textup{opt}}$ ends.
To remind us that $P_y^{\textup{opt}}$ contains~$m$, we store $m$ as an \emph{obligatory} man for $y$. 

\medskip
\noindent
{\bf Step 5: Finding all remaining feminine paths.}
In this step, we find all feminine paths that rely neither on a cycle, nor on another feminine path.
To do so, we will use a surprisingly simple method.
Intuitively, it makes sense to construct augmenting paths to be as short as possible, since the longer such a path gets, the more opportunities arise for a blocking pair to appear along the way.
This intuition is almost entirely correct, if no other path (or cycle) relies on the augmenting path $P_w^{\textup{opt}}$ we are looking for: it turns out that $P_w^{\textup{opt}}$ is either indeed the shortest possible, or if this is not the case, then we can ``get rid of'' only one blocking pair by constructing a longer augmenting path---and this, too, we can keep as short as it is possible, by stopping whenever we encounter a decrease in the number of blocking pairs along the augmenting path.
In the case when some other path relies on $P_w^{\textup{opt}}$, we need to be somewhat more careful, and take into account the obligatory men defined in Step~4. 

After this high-level explanation, let us define Step 5 formally. 
Let $\mathcal{W}_r$ be the set of those women $w \in \mathcal W^{\star}_0$ for which $P_w^{\textup{opt}}$ has not been found yet 
(that is, where $P_w^{\textup{opt}}$ is not neutral, and it does not rely on any other path or cycle in $M_s \triangle M^{\textup{opt}}$). 
Let us fix any woman $w \in \mathcal{W}_r$.
With any candidate path~$Q$ for~$w$, we associate a matching $M_{Q}=(M^{\textup{opt}} \triangle P_w^{\textup{opt}}) \triangle Q$; 
we can think of $M_{Q}$ as the matching that results from replacing $P_w^{\textup{opt}}$ with $Q$ in the symmetrical difference of $M_s$ and the solution.
Let us state a few useful properties of candidate paths and their associated matchings. 
\begin{proposition}
\label{prop:candidates}
The following holds for any candidate path $Q$ for $w$:
\begin{enumerate}
\item[(a)] $M_s \triangle Q$ is a matching.
\item[(b)] $M_{Q}$ is a feasible matching.
\item[(c)] $Q$ must be one of the following two forms: \\
(i) $Q$ ends at a man in $\mathcal{M}_0$, in which case $Q=A_w=P_w^{\textup{opt}}$,
or \\
(ii) $Q$ ends at a non-distinguished woman. 
\end{enumerate}
\end{proposition}

\begin{proof}
Recall that the definition of an augmenting path (see the beginning of Section~\ref{sec:prefatmost2}) 
directly requires that $M_s \triangle Q$ be a matching, proving (a). 
This immediately implies that if $Q$ ends at a man, 
then this man cannot be matched by $M_s$, from which $Q=A_w$ and thus $Q=P_w^{\textup{opt}}$ follows.
Using again the definition of an augmenting path, we get that the last person on $Q$ is not distinguished 
(because $P_w^{\textup{opt}}$ is not neutral by Step~3), showing (c).

To see that $M_{Q}$ is a matching, 
  let $\mathcal{K}_{-w}$ denote the union of all connected components of $M_s \triangle M^{\textup{opt}}$ except for $P_w^{\textup{opt}}$.
  Then  $M^{\textup{opt}} \triangle P_w^{\textup{opt}}=M_s \triangle \mathcal{K}_{-w}$ is a matching. 
  But as $\mathcal{K}_{-w}$ is vertex-disjoint from $Q$, 
  from (a) we get that $M_{Q}=(M_s \triangle \mathcal{K}_{-w}) \triangle Q$ must be a matching too. 
  The properties of $Q$ described in (c) show also that $M_{Q}$ is feasible.
\renewcommand{\qedsymbol}{$\Diamond$}
\end{proof}

Now, we define two candidate paths for $w$ as follows:
\begin{itemize}
  \item First, we let~$Q_w^1$ be the shortest candidate path for $w$ that contains all obligatory men for~$w$.  	
    In particular, if there are no obligatory men for~$w$, then $Q_w^1$ is the shortest candidate path for $w$;
    see Figure~\ref{fig:remaining-fem-paths} for an illustration.
  \item Second, let $Q_w^2$ be the shortest candidate path containing~$Q_w^1$ such that $M_s \triangle Q_w^2$ admits less blocking pairs than $M_s \triangle Q_w^1$.
    Such a path, however, may not exist, in which case we leave $Q_w^2$ undefined.
\end{itemize}

\begin{figure}[t]
    \tikzset{rdvertex/.style={minimum size=2mm,circle,fill=white,draw, inner sep=0pt},
        sqvertex/.style={minimum size=2mm,diamond,fill=black,draw, inner sep=0pt},
        opvertex/.style={minimum size=2mm,diamond,fill=white,draw, inner sep=0pt},
        decoration={markings,mark=at position .5 with {\arrow[black,thick]{stealth}}}}
        \centering
        \begin{tikzpicture}  
        \def\h{0}   
        \def\w{0.8}   
        \node (p1) at (-0.4,\h){$A_{w_1}$:};
		\node (w1) at (1*\w,\h)[sqvertex,label=left:$w_1$]{};
        \node (p1-m1) at (2*\w,\h)[rdvertex,label=below:$M^{\textnormal{opt}}(w_1)$]{};
        \node (p1-w1) at (3*\w,\h)[opvertex]{};        
        \node (p1-hdots1) at (4*\w,\h){$\hdots$};
        \node (p1-m2) at (5*\w,\h)[rdvertex,label=below:$m_1^{\textup{obl}}$]{};
        \node (p1-w2) at (6*\w,\h)[opvertex]{};
		\node (p1-hdots2) at (7*\w,\h){$\hdots$};
		\draw (w1)--(p1-m1);
		\draw[double] (p1-m1)--(p1-w1);		
		\draw (p1-w1)--(p1-hdots1);				
		\draw (p1-hdots1)--(p1-m2);				
		\draw[double] (p1-m2)--(p1-w2);
		\draw (p1-w2)--(p1-hdots2);
		\draw[decoration={brace,amplitude=6pt,raise=5pt},decorate]  (w1.north west) -- node[midway,above=10pt] {$Q^1_{w_1}$} (p1-w2.north east); 		
		\def\h{-1.5}
        \node (p2) at (-0.4,\h){$A_{w_2}$:};
        \node (w2) at (1*\w,\h)[sqvertex,label=left:$w_2$]{};
        \node (p2-m1) at (2*\w,\h)[rdvertex,label=below:$M^{\textnormal{opt}}(w_2)$]{};
        \node (p2-w1) at (3*\w,\h)[opvertex]{};        
        \node (p2-hdots1) at (4*\w,\h){$\hdots$};
        \node (p2-m2) at (5*\w,\h)[rdvertex,label=below:$m_2^{\textup{obl}}$]{};
        \node (p2-w2) at (6*\w,\h)[sqvertex]{};
        \node (p2-m3) at (7*\w,\h)[rdvertex]{};
        \node (p2-w3) at (8*\w,\h)[sqvertex]{};        
        \node (p2-m4) at (9*\w,\h)[rdvertex]{};
        \node (p2-w4) at (10*\w,\h)[opvertex]{};        
		\node (p2-hdots2) at (11*\w,\h){$\hdots$};
		\draw (w2)--(p2-m1);
		\draw[double] (p2-m1)--(p2-w1);		
		\draw (p2-w1)--(p2-hdots1);				
		\draw (p2-hdots1)--(p2-m2);				
		\draw[double] (p2-m2)--(p2-w2);
		\draw (p2-w2)--(p2-m3);
		\draw[double] (p2-m3)--(p2-w3);
		\draw (p2-w3)--(p2-m4);
		\draw[double] (p2-m4)--(p2-w4);
		\draw (p2-w4)--(p2-hdots2);
		\draw[decoration={brace,amplitude=6pt,raise=5pt},decorate]  (w2.north west) -- node[midway,above=10pt] {$Q^1_{w_2}$} (p2-w4.north east); 		
		\def\h{-3}
        \node (p3) at (-0.4,\h){$A_{w_3}$:};
        \node (w3) at (1*\w,\h)[sqvertex,label=left:$w_3$]{};
        \node (p3-m1) at (2*\w,\h)[rdvertex,label=below:$M^{\textnormal{opt}}(w_3)$]{};
        \node (p3-w1) at (3*\w,\h)[opvertex]{};        
        \node (p3-hdots1) at (4*\w,\h){$\hdots$};
        \node (p3-m2) at (5*\w,\h)[rdvertex,label=below:$m_3^{\textup{obl}}$]{};
        \node (p3-w2) at (6*\w,\h)[sqvertex]{};
        \node (p3-m3) at (7*\w,\h)[rdvertex]{};
		\draw (w3)--(p3-m1);
		\draw[double] (p3-m1)--(p3-w1);		
		\draw (p3-w1)--(p3-hdots1);				
		\draw (p3-hdots1)--(p3-m2);				
		\draw[double] (p3-m2)--(p3-w2);
		\draw (p3-w2)--(p3-m3);
		\draw[decoration={brace,amplitude=6pt,raise=5pt},decorate]  (w3.north west) -- node[midway,above=10pt] {$Q^1_{w_3}$} (p3-m3.north east); 			
		\end{tikzpicture}		
    \caption{An example of the paths $Q_w^1$ for three women $w \in \{w_1,w_2,w_3\} \subseteq \mathcal{W}_r$, as defined in Step~5.
    	The man $m_i^{\textup{obl}}$ for each $i \in \{1,2,3\}$ is the last obligatory man for $w_i$ on the maximal 
    	augmenting path~$A_{w_i}$. In cases where there is no obligatory man for $w_i$, one can identify $m_i^{\textup{obl}}$ with $M^{\textnormal{opt}}(w_i)$ 
    	to get an appropriate figure.
        \label{fig:remaining-fem-paths}}
\end{figure}

\begin{lemma}
\label{lem:case2x-optimalpaths}
  If $P_w^{\textup{opt}}$ does not rely on any other path or cycle in $M_s \triangle M^{\textup{opt}}$, 
  then either $P_w^{\textup{opt}}=Q_w^1$ or $P_w^{\textup{opt}}=Q_w^2$.
\end{lemma}
\begin{proof}
  First observe that for any $w \in \mathcal{W}_r$, if the guesses made by the algorithm are correct, then $P_w^{\textup{opt}}$ contains $Q_w^1$ as a subpath.
  To see this, recall that if the guesses made in Step~4 of the algorithm are correct, then $P_w^{\textup{opt}}$ indeed must contain all obligatory men for~$w$. 
  Since $P_w^{\textup{opt}}$ is a candidate path for $w$, it indeed must contain~$Q_w^1$.

  Next, suppose that $P_w^{\textup{opt}} \neq Q_w^1$, and consider the matching $M_1:=M_{Q_w^1}$. 
  By Proposition~\ref{prop:candidates}, $M_1$ is a feasible matching, and 
  $P_w^{\textup{opt}} \neq Q_w^1$ implies that $Q_w^1$ ends at a non-distinguished woman $x$.
    
  \begin{claimu}  
  \label{claim2}
    Let $Q$ be a candidate path for $w$ such that $Q_w^1 \subset Q \subseteq P_w^{\textup{opt}}$. 
    If $e$ is an edge blocking in~$M_1$ but not blocking in~$M_{Q}$, then $e$ is adjacent to $x$.
  \end{claimu} 
  \begin{proof}[Proof of Claim~\ref{claim2}.]
    Suppose for contradiction that $e$ is not incident to $x$.
    First note that $e$ must be adjacent to a person in $V(Q) \setminus V(Q_w^1)$, as all other persons (except for $x$) have the same partner in $M_1$ as in $M_{Q}$.
    Second, as $M_s$ is stable, at least one endpoint of $e$ must be matched differently in $M_1$ than in $M_s$.
    Thus, $e$ must connect a person in $V(Q) \setminus V(Q_w^1)$ either (i) with a connected component of $\mathcal{K}_{-w}$, denoting the union of all connected components of $M_s \triangle M^{\textup{opt}}$ except for $P_w^{\textup{opt}}$, or (ii) with $Q_w^1$. 

    In case (i), $e$  must be blocking in $M_s \triangle \mathcal{K}_{-w}$ (as $e$ blocks $M_1$).
    As $e$ is not blocking in $M_{Q}=M_s \triangle (\mathcal{K}_{-w} \cup Q)$, it is not blocking in $M^{\textup{opt}}=M_s \triangle (\mathcal{K}_{-w} \cup P_w^{\textup{opt}})$ either (because any person in $V(Q) \setminus V(Q_w^1)$ has the same partner in~$M_{Q}$ as in $M^{\textup{opt}}$, except possibly for the last person on $Q$ who might be left unmatched in $M_{Q}$).
    Therefore, Lemma~\ref{lemma-type} implies that $e$ is a dependent edge relying on $P_w^{\textup{opt}}$.
    By Lemma~\ref{lem:non-aug-paths-case2x}, only feminine paths may rely on~$P_w^{\textup{opt}}$. 
    However, in Step~4 we have stored the man endpoint of all edges that connect~$P_w^{\textup{opt}}$ with a feminine path relying on~$P_w^{\textup{opt}}$ as an obligatory man, so assuming correct guesses in Step~4, $e$ must have an obligatory man as its endpoint.
    But~$Q_w^1$ contains all the obligatory men for $w$, which contradicts our assumption of case (i).

    In case (ii), let $w_e$ and $m_e$ be the woman and man endpoint of $e$, respectively; we assumed $w_e \neq x$. 
    By $\Delta_{\mathcal{W}} \leq 2$, $w_e$ can only be adjacent to both $V(Q) \setminus V(Q_w^1)$ and $V(Q_w^1)$ if either $w_e=w$, or $Q$ and also $P_w^{\textup{opt}}$ ends at $w_e$. 
    In the former case, $m_e \in V(P_w^{\textup{opt}}) \setminus V(Q_w^1)$, yielding $M_1(m_e)=M_s(m_e)$. 
    As~$e$ blocks~$M_1$, $m_e$ prefers $w$ to $M_s(m_e)$; however, $w$ is not matched in $M_s$, so this means that $e$ blocks $M_s$ as well, a contradiction.
    In the latter case,~$m_e$ lies on $Q_w^1$, yielding $M_1(m_e)=M^{\textup{opt}}(m_e)$. 
    As $e$ blocks $M_1$, man~$m_e$ prefers~$w_e$ to $M^{\textup{opt}}(m_e)$; however, $w_e$ is not matched by $M^{\textup{opt}}$, so this means that $e$ blocks~$M^{\textup{opt}}$ as well, a contradiction.
    Hence the claim follows.
  \renewcommand{\qedsymbol}{$\Diamond$}
  \end{proof}
  
  Let $c_1$ denote the number of blocking pairs with respect to $M_1$ that are incident to a man or woman on $Q_w^1$.
  \begin{claimu}
  \label{claim:candidate-cost}
    If $Q$ is a candidate path for $w$ such that $Q_w^1 \subset Q \subseteq P_w^{\textup{opt}}$, and $c$ is the number of blocking pairs with respect to $M_{Q}$ that are incident to a man or woman on $Q$, then $c \geq c_1-1$.
  \end{claimu}
  \begin{proof}[Proof of Claim~\ref{claim:candidate-cost}]
    By Claim~\ref{claim2}, $c \leq c_1-2$ can only happen if both edges incident to~$x$ are blocking in~$M_1$, and none of them is blocking in $M_{Q}$. 
    However, this cannot happen, as we are going to show now.
    Let $m_1=M_s(x)$ and $m_2=M^{\textup{opt}}(x)$ be the two men in $x$'s preference list; note that~$Q_w^1$ contains $m_1$ but not $m_2$.  
    If both $\{x,m_1\}$ and $\{x,m_2\}$ are blocking in~$M_1$, then $m_1$ prefers~$x$ to $M^{\textup{opt}}(m_1)$, and $m_2$ either prefers $x$ to $M_s(m_2)$ or $m_2$ is unmatched by $M_s$.   
    Therefore the stability of~$M_s$ implies that~$x$ prefers~$m_1$ to $m_2$, as otherwise $\{x,m_2\}$ would block $M_s$.
    But this shows that $\{x,m_1\}$ is a blocking pair in~$M^{\textup{opt}}$ and thus in $M_{Q}$ as well, proving our claim. 
  \renewcommand{\qedsymbol}{$\Diamond$}
  \end{proof}


  Now, let $c^{\textup{opt}}$ denote the number of blocking pairs with respect to $M^{\textup{opt}}$ that are incident to a man or woman on $P_w^{\textup{opt}}$.
  By our assumption $Q_w^1 \neq P_w^{\textup{opt}}$, we get $c^{\textup{opt}}<c_1$, 
  as otherwise $M_1$ would also be an optimal solution which would contradict our choice of $M^{\textup{opt}}$ 
  (as $M_s \triangle M_1$ contains less edges than $M_s \triangle M^{\textup{opt}}$). 
  Hence, Claim~\ref{claim:candidate-cost} implies $c^{\textup{opt}}=c_1-1$. 
  
  By the definition of $Q_w^2$, we know that $Q_w^1 \subset Q_w^2 \subseteq P_w^{\textup{opt}}$, and 
  Proposition~\ref{prop:candidates} yields that $M_2:=M_{Q_w^2}$ is a feasible solution.
  The definition of $Q_w^2$ and Claim~\ref{claim:candidate-cost} together imply that the number of blocking pairs with respect to $M_2$ 
  that are incident to a man or woman on $Q_w^2$ is exactly $c_1-1$.
  Hence,~$M_2$ is an optimal solution.
  Therefore, $M_s \triangle M_2$ cannot contain less edges than $M_s \triangle M^{\textup{opt}}$, because assuming otherwise would contradict our choice of $M^{\textup{opt}}$.
  From this, $Q_w^2 = P_w^{\textup{opt}}$ follows. 
\end{proof}

According to Lemma~\ref{lem:case2x-optimalpaths}, we can find all remaining feminine paths in $M_s \triangle M^{\textup{opt}}$ by guessing for each $w \in \mathcal{W}_r$ whether $P_w^{\textup{opt}}$ equals $Q_w^1$ or $Q_w^2$. 
This means at most $2^{|\mathcal W^\star_0|}$ guesses; the computations needed for each guess take linear time. 

\medskip
\noindent
{\bf Step 6: Computing elimination paths.}
Let $\mathcal{F}$ be the union of cycles, feminine and neutral paths found in Steps~1 to 5.
When searching for masculine paths, we will have to deal with edges that might be type A dependent 
edges in the optimum solution.
We call an edge \emph{volatile} if it connects a woman in~$\mathcal{W}_0$ with her second choice. 
The importance of this definition is shown by the following two lemmas.

\begin{lemma}
\label{lem:eliminated-then-volatile}
  Let $f$ be an edge that is blocking in $M_s \triangle \mathcal{F}$, but not blocking in $M^{\textup{opt}}$.
  Then $f$ is volatile, its woman endpoint is contained on a non-feminine path of $M_s \triangle M^{\textup{opt}}$, and its man-endpoint is contained in $\mathcal{F}$.
\end{lemma}
\begin{proof}
  Observe that since $f$ does not block $M^{\textup{opt}}$, there must exists a connected component $Q$ of $M_s \triangle M^{\textup{opt}}$ not in $\mathcal F$ such that $f$ does not block $M_s \triangle (\mathcal{F} \cup Q)$. 
  Clearly, $Q$ is a masculine but not feminine augmenting path by the definition of $\mathcal{F}$.
  Lemma~\ref{lemma-type} then implies that $f$ is a dependent edge relying on $Q$, and
  Lemma~\ref{lem:non-aug-paths-case2x} yields also that $f$ has type A and $Q$ is masculine and not feminine. 
  The remaining statements follow from the definition of a type A edge.
\end{proof}

\begin{lemma} 
\label{lem:novolatile}
  If $f$ is a volatile edge incident to some non-feminine path in $M_s \triangle M^{\textup{opt}}$, then~$f$ is not blocking in $M^{\textup{opt}}$. 
\end{lemma}
\begin{proof}
  Recall that $M^{\textup{opt}}$ is an optimal solution for which $M_s \triangle M^{\textup{opt}}$ has as few edges as possible, and there does not exist another optimal solution~$M'$ such that (i) $M'$ has the same number of common edges with $M_s$ as $M^{\textup{opt}}$, and (ii) for each man~ $m$, either $M'(m)=M^{\textup{opt}}(m)$ or $m$ prefers $M'(m)$ to $M^{\textup{opt}}(m)$.

  Suppose for the sake of contradiction, that $f$ is a volatile edge incident to a masculine path~$P_x^{\textup{opt}}$ 
  in $M_s \triangle M^{\textup{opt}}$ for some $x \in \mathcal{M}^{\star}_0$, and $f$ blocks $M^{\textup{opt}}$.
  Let~$m$ and $w$ be the man and woman connected by~$f$. 

  Since $m$ is the second choice of $w$ and $f$ is blocking in~$M^{\textup{opt}}$, we know that~$M^{\textup{opt}}$ does not cover~$w$. 
  However, any woman on a masculine augmenting path is matched by~$M^{\textup{opt}}$ with a man adjacent to her on this path, 
  so $w$ does not lie on $P_x^{\textup{opt}}$.
  Therefore, $m$ lies on $P_x^{\textup{opt}}$.
  
  Define $P$ as the subpath of $P_x^{\textup{opt}}$ from $x$ to $m$ plus the edge $f$. 
  Consider now $M_f=M^{\textup{opt}} \triangle (P_x^{\textup{opt}} \triangle P)$; then~$M_f$ is a matching because $w \in \mathcal{W}_0$.
  Clearly, $M_f$ is feasible, since $P_x^{\textup{opt}}$ is not a neutral path. 

  \begin{claimu}
  \label{claim:volatile-cost} 
    If $f'$ is an edge that blocks $M_f$ 
    but does not block~$M^{\textup{opt}}$, then $P_x^{\textup{opt}}$ ends at a woman $y$, and~$f'$ is a volatile edge incident to $y$.
  \end{claimu}
  \begin{proof}[Proof of Claim~\ref{claim:volatile-cost}]
    Let $w'$ and $m'$ be the woman and man connected by $f'$. 
    The only persons that are matched differently in $M_f$ than in $M^{\textup{opt}}$ are those in $ (V(P_x^{\textup{opt}}) \setminus V(P)) \cup \{m,w\}$; recall that $f=\{m,w\}$ with $w \in \mathcal{W}_0$.
    Since both $m$ and $w$ are better off in $M_f$ than in $M^{\textup{opt}}$, and $f'$ blocks $M_f$ but not~$M^{\textup{opt}}$, we get that one endpoint of $f'$ must be in $V(P_x^{\textup{opt}}) \setminus V(P)$.

    Suppose first that $m' \in V(P_x^{\textup{opt}}) \setminus V(P)$.
    Then $M_f(m')=M_s(m')$ and since~$f'$ blocks $M_f$, we know that $m'$ prefers $w'$ to $M_s(m')$.
    But $M_s$ is stable, so~$w'$ prefers $M_s(w')$ to $m'$, and in particular is matched by $M_s$.
    As $f'$ blocks~$M_f$, we get that $M_f(w') \neq M_s(w')$. 
    Hence, as $w' \neq w$ (because $w \in \mathcal{W}_0$ but $w' \notin \mathcal{W}_0$), it must be the case that $M_f(w')=M^{\textup{opt}}(w')$ and $w'$ is on a connected component of $\mathcal{K}_{-x}$, denoting the union of all connected components of $M_s \triangle M^{\textup{opt}}$ other than $P_x^{\textup{opt}}$.
    This implies that $f'$ is blocking in $M_s \triangle \mathcal{K}_{-x}$ as well.
    Since $f'$ is not blocking in $M^{\textup{opt}}=M_s \triangle (\mathcal{K}_{-x} \cup P_x^{\textup{opt}})$, Lemma~\ref{lemma-type} implies that $f'$ must be a dependent edge relying on $P_x^{\textup{opt}}$.
    Note that since $m' \in V(P_x^{\textup{opt}})$, $f'$ cannot be of type~A (by the definition of a type~A edge), so it must be of type~B.
    However, Lemma~\ref{lem:non-aug-paths-case2x} tells us that a type~B edge relies either on a cycle or on a feminine path, but as $P_x^{\textup{opt}}$ is neither of the two, we arrive at a contradiction.
  
    Thus, we get $w' \in V(P_x^{\textup{opt}}) \setminus V(P)$. 
    Observe that $m'=m$ is not possible, because it would imply $f'=\{m,M^{\textup{opt}}(m)\}$ contradicting the assumption that~$f'$ blocks $M_f$ (note that $m$ prefers $w$ to $M^{\textup{opt}}(m)$ because $f$ blocks $M^{\textup{opt}}$).  
    From this we obtain that $f'$ cannot be contained in $P_x^{\textup{opt}}$: we have $f' \notin M_f$ by definition, and $f' \in M^{\textup{opt}}$ is also not possible, since in that case $m' \notin V(P_x^{\textup{opt}}) \setminus V(P)$ would imply $m'=m$.
    Thus, by $\Delta_{\mathcal W} \leq 2$ we know that $w'=y$, and since $f'$ blocks $M_f$ but not $M^{\textup{opt}}$, we also get that~$m'$ must be the second choice of $w'$. 
    Hence $f'$ is volatile, proving our claim.
  \renewcommand{\qedsymbol}{$\Diamond$}
  \end{proof}

  By Claim~\ref{claim:volatile-cost}, there can be at most one edge that is blocking in $M_f$ but not in $M^{\textup{opt}}$. 
  Since $f$ is an edge that blocks $M^{\textup{opt}}$ but not $M_f$, we get that $M_f$ is an optimal matching. 
  Moreover, either $M_f$ has more common edges with $M_s$ than~$M^{\textup{opt}}$, or $m$ is the last man on $P_x^{\textup{opt}}$, 
  in which case every man has the same partner in $M_f$ as in $M^{\textup{opt}}$ except for $m$ who is better off in~$M_f$ than in~$M^{\textup{opt}}$.
  This contradicts our choice of $M^{\textup{opt}}$.
\end{proof}


For any volatile edge $f$ we can decide in linear time if there exists a masculine augmenting path disjoint from $\mathcal{F}$ 
that contains the woman endpoint $w$ of~$f$, but not $f$ itself. 
Indeed, we can build such a path starting from $w$ by taking edges not in $M_s$ and edges in $M_s$ in an alternating manner;
by $\Delta_{\mathcal W} \leq 2$ and since we need to start with an edge different from $f$,
we always have at most one possibility to pick our next edge. 
This process may or may not result in a masculine augmenting path, but if it does, then the path is unique, proving the following claim.

\begin{proposition}
\label{unique-elim-path}
  Let $f$ be a volatile edge.
  If there exists a masculine augmenting path disjoint from $\mathcal{F}$ that contains the woman endpoint $w$ of $f$ but not $f$ itself, then this path is unique; we denote it by $Q^{\textup{elim}}_f$.
\end{proposition}
  
Let~$f$ be a volatile edge that is blocking in $M_s \triangle \mathcal{F}$. 
We say that a set $\mathcal{P}_f$ of masculine augmenting paths \emph{eliminates}~$f$ if 
(i) $Q^{\textup{elim}}_f$ exists and $Q^{\textup{elim}}_f \in \mathcal{P}_f$, and 
(ii) for any path $Q \in \mathcal{P}_f$, if there is a volatile blocking edge~$f'$ in $M_s \triangle Q$, then~$Q^{\textup{elim}}_{f'}$ exists and is contained in $\mathcal{P}_f$.
We refer to the (inclusion-wise) minimal set of masculine paths eliminating $f$ as the \emph{elimination paths} for~$f$,
and denote it by~$\mathcal{P}^{\textup{elim}}_f$. 
Further, we refer to the starting vertices of these paths as the \emph{elimination set} for~$f$.
The next lemma illuminates the role of elimination paths.

\begin{lemma}
\label{lem:elimination}
  If $f$ is an edge blocking in $M_s \triangle \mathcal{F}$ but not blocking in $M^{\textup{opt}}$, 
  then $M_s \triangle M^{\textup{opt}}$ contains all paths in $\mathcal{P}^{\textup{elim}}_f$.
\end{lemma}
\begin{proof}
  By Lemma~\ref{lem:eliminated-then-volatile}, the woman endpoint $w$  of $f$ lies on a masculine  path. 
  By Proposition~\ref{unique-elim-path}, this path can only be $Q^{\textup{elim}}_f$, and thus $Q^{\textup{elim}}_f$ must be contained in $M_s \triangle M^{\textup{opt}}$.
  For an inductive reasoning, assume that $Q$ is a masculine path in~$\mathcal{P}^{\textup{elim}}_f$ that is contained in $M_s \triangle M^{\textup{opt}}$ 
  and $f'$ is a volatile edge that is blocking in $M_s \triangle Q$.
  We claim that~$Q^{\textup{elim}}_{f'}$ is contained in $M_s \triangle M^{\textup{opt}}$, which by induction proves our lemma. 

  To see our claim, note that by Lemma~\ref{lem:novolatile}, $f'$ cannot be blocking in $M^{\textup{opt}}$.
  Since the woman endpoint of $f'$ is unmatched by $M_s$, it can only be blocking in $M_s \triangle Q$ if its man endpoint lies on $Q$.
  Thus, since~$f'$ is not blocking in~$M^{\textup{opt}}$ (but is blocking in $M_s \triangle Q$), 
  it must be the case that its woman endpoint lies on a masculine path in $M_s \triangle M^{\textup{opt}}$.
  Since such a path can only be $Q^{\textup{elim}}_{f'}$ by Proposition~\ref{unique-elim-path}, we get that $Q^{\textup{elim}}_{f'} \in \mathcal{P}^{\textup{elim}}_f$ and the lemma follows. 
\end{proof}

\medskip
\noindent
{\bf Step 7: Guessing relevant elimination sets in $M^{\textup{opt}}$.}
We call an edge \emph{relevant} in $M^{\textup{opt}}$, if it is a volatile edge blocking $M_s \triangle \mathcal{F}$, but it does not block~$M^{\textup{opt}}$. 
By Lemma~\ref{lem:elimination}, if $f$ is a relevant edge in $M^{\textup{opt}}$, then $M_s \triangle M^{\textup{opt}}$ must contain all paths in~$\mathcal{P}^{\textup{elim}}_f$. 
Since there may be several volatile edges blocking in $M_s \triangle \mathcal{F}$, we cannot determine the relevant ones among them by simply guessing them.
Instead, we only guess the elimination sets for all relevant edges. 
Clearly, these sets must be pairwise disjoint subsets of~$\mathcal{M}^{\star}_0$, 
so guessing them results in at most $(|\mathcal{M}^{\star}_0|+1)^{|\mathcal{M}^{\star}_0|}$ possibilities. 
Let us denote by $R_1, \dots, R_{\ell}$ the guessed elimination sets.

\medskip
\noindent
{\bf Step 8: Computing cheapest elimination paths.}
For each set $R_i \subseteq \mathcal{M}^\star_0$ with $i = 1,\hdots,\ell$ that, according to our guesses made in Step 7, forms the elimination set for a volatile edge relevant in $M^{\textup{opt}}$, we determine \emph{some} volatile edge $f$ incident to $\mathcal{F}$ that is blocking in $M_s \triangle \mathcal{F}$ and whose elimination set is exactly $R_i$. 
Namely, we pick an edge $f$ among all such edges in a way that the number of blocking pairs in $M_s \triangle (\mathcal{F} \cup \mathcal{P}^{\textup{elim}}_f)$ is as small as possible.
Let~$f_i$ be the volatile edge chosen this way, and let $\mathcal{P}^{\textup{elim}}=\bigcup_{1 \leq i \leq \ell} \mathcal{P}^{\textup{elim}}_{f_i}$. 

\medskip
\noindent
{\bf Step 9: Computing remaining masculine paths.}
We define $\mathcal{M}_r=\mathcal{M}^\star_0 \setminus (R_1 \cup \dots \cup R_{\ell})$ as the set of distinguished men that are neither covered by~$M_s$ nor contained in any of the sets $R_1, \dots, R_{\ell}$. 
For each such $m$ we are going to compute an augmenting path $P_m$ disjoint from~$\mathcal{F}$ such that 
the number of edges that block $M_s \triangle (\mathcal{F} \cup \bigcup_{m \in \mathcal{M}_r} P_r)$ but not $M_s \triangle \mathcal{F}$ 
is minimized. 
We set $\mathcal{P}_r=\{P_m \mid m \in \mathcal{M}_r \}$.

To compute $\mathcal{P}_r$, we use the algorithm of Theorem~\ref{thm:restricted-matching-poly} with some modifications.

\smallskip
\noindent 
{\bf Step 9.1:} We compute all augmenting paths that start at a man $m \in \mathcal{M}_r$ and are disjoint from $\mathcal{F}$.
For each such augmenting path $P$ we define a set $C(P)$ containing those edges that block $M_s \triangle (\mathcal{F} \cup P)$ but not $M_s \triangle \mathcal{F}$, and are either non-volatile or have a woman endpoint in $\mathcal{F}$; we define the \emph{contributing cost} of $P$ as $|C(P)|$.

\smallskip
\noindent 
{\bf Step 9.2:} We construct the auxiliary graph $G_{\textup{path}}=(U,V;E)$ as follows: we set $U=\mathcal{M}_r$ and $V=\mathcal{W}_0 \cup \{m' \mid m \in \mathcal{M}_r\}$.
The edge set $E$ contains an edge $\{m,m'\}$ for each $m \in U$, as well as an edge $\{m,w\}$ whenever $m \in U$, $w \in\mathcal W_0$ and there exists an augmenting path disjoint from $\mathcal{F}$ with endpoints~$m$ and~$w$.
We define the weight of an edge $\{m,m'\}$ as the minimum contributing cost~$c_m^{\textup{min}}$ of any augmenting path starting at~$m$ and \emph{not} ending in $\mathcal W_0$, and we define the weight of an edge $\{m,w\}$ with $m \in U$ and $w \in\mathcal W_0$ as the minimum contributing cost of any augmenting path with endpoints $m$ and $w$, disjoint from $\mathcal{F}$.

\smallskip
\noindent 
{\bf Step 9.3:} We compute a minimum weight matching $M_P$ in $G_{\textup{path}}$ covering~$U$ 
the same way as in the algorithm of Theorem~\ref{thm:restricted-matching-poly}; let $\mathcal{P}$ denote the set of augmenting paths corresponding to the edges of the matching $M_P$.
Note that the paths in $\mathcal{P}$ are pairwise disjoint, by the construction of $G_{\textup{path}}$ and because~$M_P$ is a matching
in $G_{\textup{path}}$.

\smallskip
\noindent 
{\bf Step 9.4:} We eliminate all volatile edges that block $M_s \triangle (\mathcal{F} \cup \mathcal{P})$ but not $M_s \triangle \mathcal{F}$ and are not incident to~$\mathcal{F}$. 
We modify $\mathcal{P}$ iteratively. 
We start by setting $\mathcal{P}_{\textup{act}}=\mathcal{P}$.
At each iteration we modify~$\mathcal{P}_{\textup{act}}$ as follows. 
We check whether there exists a volatile edge  $\{m^*,w^*\}$ with $w^* \in \mathcal{W}_0$ that is not incident to $\mathcal{F}$, and blocks $M_s \triangle (\mathcal{F} \cup \mathcal{P}_{\textup{act}})$ but not $M_s \triangle \mathcal{F}$. 
If yes, then notice that~$w^*$ is not matched in $M_s \triangle (\mathcal{F} \cup \mathcal{P}_{\textup{act}})$, because 
$m^*$ is the second choice for $w^*$.
Let $P$ be the path of~$\mathcal{P}_{\textup{act}}$ containing~$w^*$.
We modify $\mathcal{P}_{\textup{act}}$ by truncating~$P$ to its subpath between its starting vertex and~$m^*$, and appending to it the edge $\{m^*,w^*\}$. 
This way, $\{m^*,w^*\}$ becomes an edge of the matching~$M_s \triangle (\mathcal{F} \cup \mathcal{P}_{\textup{act}})$. 
The iteration stops when there is no volatile edge disjoint from $\mathcal{F}$ blocking~$M_s \triangle (\mathcal{F} \cup \mathcal{P}_{\textup{act}})$
but not $M_s \triangle \mathcal{F}$. 
Note that once a volatile edge ceases to be blocking in $M_s \triangle (\mathcal{F} \cup \mathcal{P}_{\textup{act}})$, 
it cannot become blocking again during this process, so the algorithm performs at most $|\mathcal W_0|$ iterations.
For each $m \in \mathcal{M}_r$, let~$P_m$ denote the augmenting path in~$\mathcal{P}_{\textup{act}}$ covering~$m$ at the end of this step; we set $\mathcal{P}_r:=\{P_m \mid m \in \mathcal{M}_r\}$.

\smallskip
\noindent 
{\bf Step 9.5:}
Finally, we output the matching $M^{\textup{out}}=M_s \triangle (\mathcal{F} \cup \mathcal{P}^{\textup{elim}} \cup \mathcal{P}_r)$.

\medskip
It is straightforward to verify that the number of guesses made are bounded by a function of $|\mathcal{W}^\star_0|+|\mathcal{M}^\star_0|$, and all computations in a branch can be performed in time polynomial in the size~$|I|$ of the instance, yielding a fixed-parameter algorithm with parameter~$|\mathcal{W}^\star_0|+|\mathcal{M}^\star_0|$. 
It remains to prove the correctness of the proposed algorithm.

To this end, we first prove a simple observation. 

\begin{proposition}
\label{prop:no-double-blocking}
  Suppose $\Delta_{\mathcal W} \leq 2$.
  Let $M$ be a matching and $\mathcal{K}_M$ the set of connected components in $M_s \triangle M$. 
  Let also $\mathcal{H}_1$ and $\mathcal{H}_2$ be two disjoint subsets of $\mathcal{K}_M$.
  Then there is no edge that blocks both $M_s \triangle \mathcal{H}_1$ and $M_s \triangle \mathcal{H}_2$
\end{proposition}
\begin{proof}
  Suppose that $e$ blocks both  $M_1:=M_s \triangle \mathcal{H}_1$  and  $M_2:=M_s \triangle \mathcal{H}_2$; by the stability of $M_s$, $e$ connects a person $a_1$ in $\mathcal{H}_1$ with a person $a_2$ in $\mathcal{H}_2$.
  Since $e$ blocks $M_1$, we know that $a_2$ prefers $a_1$ to its situation in $M_s$. 
  Since $e$ blocks $M_2$, we also get that $a_1$ prefers $a_2$ to its situation in $M_s$.
  This contradicts the stability of $M_s$.
\end{proof}

Next, let us prove that Step 9 works as promised.
\begin{lemma}
\label{lem:step9}
  For each $m \in \mathcal{M}_r$ let $P'_m$ be an augmenting path disjoint from~$\mathcal{F}$ and starting at $m$, and let $\mathcal{P}'_r=\{P'_m \mid m \in \mathcal{M}_r \}$.
  Then the number of edges that block $M_s \triangle (\mathcal{F} \cup \mathcal{P}_r)$ but not $M_s \triangle \mathcal{F}$
  is at most the number of edges that block $M_s \triangle (\mathcal{F} \cup \mathcal{P}'_r)$ but not $M_s \triangle \mathcal{F}$.
\end{lemma}
\begin{proof}
  Let $c'$ be the number of edges that block $M_s \triangle (\mathcal{F} \cup \mathcal{P}'_r)$ but not $M_s \triangle \mathcal{F}$.
  Then Claim~\ref{claim:step9-opt-2} implies that there can be at most $c'$ edges in $\bigcup_{m \in \mathcal{M}_r} C(P'_m)$. 
  \begin{claimu}
  \label{claim:step9-opt-2}
    If $e$ is an edge such that $e \in C(P'_m)$ for some $m \in \mathcal{M}_r$, then $e$ blocks $M_s \triangle (\mathcal{F} \cup \mathcal{P}'_r)$ but not $M_s \triangle \mathcal{F}$.
  \end{claimu}
  \begin{proof}[Proof of Claim~\ref{claim:step9-opt-2}.]
    Let $m^*$ and $w^*$ be the man and woman connected by $e$. 
    By $e \in C(P'_m)$ it does not block $M_s \triangle \mathcal{F}$, so we only need to show that $e$ blocks $M_s \triangle (\mathcal{F} \cup \mathcal{P}'_r)$.
    Suppose for contradiction that $e$ does not block  $M_s \triangle (\mathcal{F} \cup \mathcal{P}'_r)$.
    Then it must clearly be adjacent to some path $P'_x$, $x \in \mathcal{M}_r$.
    By $\Delta_{\mathcal{W}} \leq 2$ we know that $w^*$ is the endpoint of one of these paths, $w^* \in \mathcal{W}_0$, and since $e$ blocks $M_s \triangle (\mathcal{F} \cup P'_m)$ but not $M_s \triangle (\mathcal{F} \cup \mathcal{P}'_r)$ we also get that $m^*$ must be the less preferred choice of~$w^*$; this yields that $e$ is volatile.
    However, as $e$ is not incident to $\mathcal{F}$ (because it connects $P'_m$ and $P'_x$), this contradicts $e \in C(P'_m)$. 
  \renewcommand{\qedsymbol}{$\Diamond$}
  \end{proof}

  Let $m_1$ and $m_2$ be two distinct men in $\mathcal{M}_r$; we will show $C(P'_{m_1}) \cap C(P'_{m_2}) = \emptyset$.
  Assuming otherwise, let $e$ be an edge in $C(P'_{m_1}) \cap C(P'_{m_2})$. 
  Then~$e$ does not block $M_s \triangle \mathcal{F}$ but blocks both $M_s \triangle (\mathcal{F} \cup P'_{m_1})$ and $M_s \triangle (\mathcal{F} \cup P'_{m_2})$.
  Hence $e$ connects $P'_{m_1}$ with $P'_{m_1}$, and thus blocks $M_s \triangle P'_{m_1}$ and $M_s \triangle P'_{m_2}$ as well,  contradicting Proposition~\ref{prop:no-double-blocking}.
  Thus, $$\left|\bigcup_{m \in \mathcal{M}_r} C(P'_m)\right|=\sum_{m \in \mathcal{M}_r} |C(P'_m)| \leq c' \enspace .$$
  Then, by the definition of edge weights in $G_{\textup{path}}$, we get that there exists a matching $M'$ in $G_{\textup{path}}$ with weight at most $\sum_{m \in \mathcal{M}_r} |C(P'_m)| \leq c'$ that covers $U=\mathcal{M}_r$.
  Thus, the matching computed by the algorithm in Step 9.3 has weight at most $c'$ as well, yielding that the set $\mathcal{P}_r$ of augmenting path corresponding to this matching is such that $\sum_{m \in \mathcal{M}_r} |C(P_m)| \leq c'$, that is, their total contribution cost is at most $c'$.
  From this, Claim~\ref{claim:step9-opt-1} below implies that there can be at most $c'$ edges that block $M_s \triangle (\mathcal{F} \cup \mathcal{P}_r)$ but not $M_s \triangle \mathcal{F}$.
  
  \begin{claimu}
  \label{claim:step9-opt-1}
    If an edge $e$ blocks $M_s \triangle (\mathcal{F} \cup \mathcal{P}_r)$ but not $M_s \triangle \mathcal{F}$, then $e \in C(P_m)$ for some $m \in \mathcal{M}_r$.
  \end{claimu}
  \begin{proof}[Proof of Claim~\ref{claim:step9-opt-1}.]
    Let $m^*$ and $w^*$ be the man and woman connected by $e$.
    Clearly, $e$ is adjacent to some path in $\mathcal{P}_r$. 

    First, suppose that $e$ is adjacent to only one path $P_m$ where $m \in \mathcal{M}_r$.
    Then $e$ blocks $M_s \triangle (\mathcal{F} \cup P_m)$ as well.
    In that case, to prove $e \in C(P_m)$ we only need to show that either $w^*$ is in $\mathcal{F}$ or is $e$ is non-volatile.
    Suppose for contradiction that $e$ is volatile and $w^*$ is not in $\mathcal{F}$.
    However, as a result of Step 9.4, there can be no volatile edge that blocks $M_s \triangle (\mathcal{F} \cup \mathcal{P}_r)$, unless it is incident to $\mathcal{F}$.
    Hence, we must have that~$m^*$ is in $\mathcal{F}$, and since $e$ blocks $M_s \triangle (\mathcal{F} \cup P_m)$, we get that $m^*$ prefers $w^*$ to its partner in $M_s \triangle (\mathcal{F} \cup P_m)$, which is the same partner he has in $M_s \triangle \mathcal{F}$.
    However, this yields that $e$ blocks $M_s \triangle \mathcal{F}$ as well: since $w^* \in \mathcal{W}_0$ and $w^*$ is not in $\mathcal{F}$, she is unmatched in $M_s \triangle \mathcal{F}$. 
    Hence, we arrive at a contradiction, proving $e \in C(P_m)$.
  
    Second, suppose that $e$ connects two paths $P_{m_1}$ and $P_{m_2}$ where $m_1$ and $ m_2 $ are distinct men in $\mathcal{M}_r$.
    By $\Delta_{\mathcal{W}} \leq 2$, we know that $w^*$ is an endpoint of one of these paths; without loss of generality, we may assume that $w^*$ is an endpoint of $P_{m_1}$.
    Since $P_{m_1}$ is a masculine augmenting path, we get $w^* \in \mathcal W_0$.
    Now, we know that $e$ blocks $M_s \triangle (\mathcal{F} \cup P_{m_1} \cup P_{m_2})$, so it must also block $M_s \triangle (\mathcal{F} \cup P_{m_2})$, because $w^*$ is unmatched in $M_s \triangle (\mathcal{F} \cup P_{m_2})$ and $m^*$ has the same partner in $M_s \triangle (\mathcal{F} \cup P_{m_2})$ as in $M_s \triangle (\mathcal{F} \cup P_{m_1} \cup P_{m_2})$. 
    Again, $e$ cannot be volatile as a result of Step 9.4, yielding that $e$ is contained in $C(P_{m_2})$. 
  \renewcommand{\qedsymbol}{$\Diamond$}
  \end{proof}
   This proves the optimality of $\mathcal{P}_r$ as stated by the lemma.
\end{proof}

Next we state the following useful lemma.
\begin{lemma}
\label{lem:additive-elim-costs}
  Let $f$ be a volatile edge blocking in $M_s \triangle \mathcal{F}$, $Q$ an elimination path in $\mathcal{P}^{\textup{elim}}_{f}$, 
and $Q'$ a non-feminine augmenting path not contained in $\mathcal{P}^{\textup{elim}}_{f}$, disjoint from $\mathcal{F}$.
  Let $e$ be an edge. 
  Then $e$ blocks $M_s \triangle (Q \cup Q')$ if and only if it blocks $M_s \triangle Q$ or $M_s \triangle Q'$.
\end{lemma}
\begin{proof}
  First observe that the statement of the lemma is obviously true for any edge $e$ that is incident to at most one of $Q$ and $Q'$.
  So suppose that $e$ connects~$Q$ with $Q'$.
  Since both $Q$ and $Q'$ are masculine augmenting paths and $\Delta_{\mathcal{W}} \leq 2$, by Proposition~\ref{prop:2xpaths} we know that $e$ connects the woman endpoint~$w$ of one of these paths with a man $m$ on the other path; moreover, $w$ must be unmatched by $M_s$ as it is the endpoint of a masculine augmenting path.
  Let $Q_w$ be the path (either $Q$ or $Q'$) that contains $w$, and let~$Q_m$ be the one containing~$m$.
  By $w \in \mathcal{W}_0$ and the stability of $M_s$, $e$ cannot be blocking in $M_s \triangle Q_w$.

  Now, if $e$ blocks $M_s \triangle (Q_m \cup Q_w)$, then it blocks $M_s \triangle Q_m$ as well, since $m$ prefers $w$ to its partner in $M_s \triangle Q_m$ and $w$ is unmatched in $M_s \triangle Q_m$.

  It remains to show that if $e$ blocks $M_s \triangle Q_m$, then it also blocks $M_s \triangle (Q_w \cup Q_m)$.
  Supposing otherwise, it must be the case that (i) $m$ prefers $w$ to its partner in $M_s \triangle Q_m$, but (ii) $w$ prefers its partner in $M_s \triangle Q_w$ to $m$.
  Consequently,~$e$ is a volatile edge.\footnote{This case is analogous with $e$ being a type~A edge; however, the paths $Q$ and $Q'$ here need not be paths of $M_s \triangle M^{\textup{opt}}$.}
  We distinguish two cases.

  Case (A): $Q_w=Q'$. Then $Q'$ must be the unique masculine augmenting path containing the woman endpoint of $e$ but not $e$ itself, that is, $Q'=Q^{\textup{elim}}_e$.
  Hence, since $e$ is a volatile edge blocking  $M_s \triangle Q$ and $Q \in \mathcal{P}^{\textup{elim}}_{f}$, by the definition of elimination paths we get $Q' \in \mathcal{P}^{\textup{elim}}_{f}$ as well, a contradiction.

  Case (B): $Q_m=Q'$. Then $Q'$ contains the man endpoint of $e$. 
  However, again by the definition of elimination paths, as $Q$ is a path of $\mathcal{P}^{\textup{elim}}_{f}$, the man endpoint of $e$ must be contained either in $\mathcal{F}$ (if $e=f$) or another path of $\mathcal{P}^{\textup{elim}}_{f}$ that triggered the addition of $Q$ into $\mathcal{P}^{\textup{elim}}_{f}$; both possibilities contradict our conditions on $Q'$.
\renewcommand{\qedsymbol}{$\Diamond$}
\end{proof}  

Now we are ready to show that our algorithm is correct.
\begin{proof}[Proof of Theorem~\ref{thm:fpt-case2x}.]
  To prove the correctness of the proposed algorithm, we first show that if all our guesses are true, then the paths and cycles in $\mathcal{F}$ are exactly the feminine paths and the cycles of $M_s \triangle M^{\textup{opt}}$. 
  The correctness of Step~2 is stated by Lemma~\ref{lem:finding-all-cycles}.
  From the description of our algorithm, it should be clear that the correctness of Steps~3 and~4 follows directly from  Proposition~\ref{prop:2xpaths} and Lemma~\ref{lem:non-aug-paths-case2x}. 
  Lemma~\ref{lem:case2x-optimalpaths} guarantees the correctness of Step~5, which proves that in Steps~1--5 the algorithm indeed finds all cycles and feminine paths of $M_s \triangle M^{\textup{opt}}$. 

  Next, let us argue that $M^{\textup{out}}$ is indeed a matching.
  For this, apart from the correctness of Steps~1--5, we need that the masculine paths in $M_s \triangle M^{\textup{out}}$ are disjoint from $\mathcal{F}$. 
  Further, we also need that paths in $\mathcal{P}^{\textup{elim}}$ are disjoint from all remaining masculine paths. To see this, observe that any path~$P$ in $\mathcal{P}^{\textup{elim}}$ ends at a woman $w \in \mathcal{W}_0$ which is connected by a volatile edge (not on $P$) to either $\mathcal{F}$ or to another path in $\mathcal{P}^{\textup{elim}}$.
  Hence, $w$ cannot lie on any masculine path other than $P$ by Proposition~\ref{unique-elim-path}. 
  Thus,~$M^{\textup{out}}$ is a matching.
  Its feasibility is implied by the correctness of Steps 1--5, and the definition of augmenting paths.

  It remains to argue that $M^{\textup{out}}$ admits at most as many blocking pairs as~$M^{\textup{opt}}$.
  First, Lemma~\ref{lem:eliminated-then-volatile} implies that all edges blocking in $M_s \triangle \mathcal{F}$ are either relevant 
  volatile edges in $M^{\textup{opt}}$, or they are also blocking in~$M^{\textup{opt}}$. 
  Furthermore, if~$f_i^{\textup{opt}}$ is a relevant volatile edge with elimination set $R_i$ for some $i \in \{1,\dots, \ell\}$, then by Lemma~\ref{lem:elimination} we know that all elimination paths in $\mathcal{P}^{\textup{elim}}_{f^{\textup{opt}}_i}$ must be contained in $M_s \triangle M^{\textup{opt}}$.
  As $\mathcal{F}$ is the set of feminine augmenting paths and cycles of $M_s \triangle M^{\textup{opt}}$, 
  in Step~8 the algorithm is bound to find \emph{some} volatile edge $f_i$ (though not necessarily $f^{\textup{opt}}_i$) 
  that is blocking in $M_s \triangle \mathcal{F}$ and whose elimination set is~$R_i$. 
  Furthermore, by our choice of $f_i$, there are at most as many blocking pairs in $M_s \triangle (\mathcal{F} \cup \mathcal{P}^{\textup{elim}}_{f_i})$
  as there are in  $M_s \triangle (\mathcal{F} \cup \mathcal{P}^{\textup{elim}}_{f^{\textup{opt}}_i})$.
  
  Let the \emph{contribution} of a volatile edge $e$, denoted by $C(e)$, be the set of edges that 
  block $M_s \triangle (\mathcal{F} \cup \mathcal{P}^{\textup{elim}}_{e})$ but not $M_s \triangle \mathcal{F}$;
  we extend this notion to any set $E$ of volatile edges with pairwise disjoint elimination sets by defining the contribution $C(E)$ 
  of $E$ as the set of edges that 
  block $M_s \triangle (\mathcal{F} \cup \bigcup_{e \in E} \mathcal{P}^{\textup{elim}}_{e})$ but not $M_s \triangle \mathcal{F}$.
  By our choice of $f_i$,  we know  $|C(f_i)| \leq |C(f^{\textup{opt}}_i)|$.
 
  We are going to show 
  $|C(\{f_i \mid 1 \leq i \leq \ell\})| \leq |C(\{f^{\textup{opt}}_i \mid 1 \leq i \leq \ell\})|,$ which implies that 
  $M_s \triangle (\mathcal{F} \cup \bigcup_{1 \leq i \leq \ell} \mathcal{P}^{\textup{elim}}_{f_i})$ has at most as many
  blocking pairs as $M_s \triangle (\mathcal{F} \cup \bigcup_{1 \leq i \leq \ell} \mathcal{P}^{\textup{elim}}_{f^{\textup{opt}}_i})$ does.
  To prove this, it suffices to prove the following claim.
  \begin{claimu}
  \label{claim:contribution-additivity}
    Let $F \cup \{f\}$ be a set of volatile edges with pairwise disjoint elimination sets, $F \neq \emptyset$.
    Then $|C(f)|+|C(F)|=|C(F \cup \{f\})|$.
  \end{claimu}
  \begin{proof}[Proof of Claim~\ref{claim:contribution-additivity}.]
    We need to prove that any edge $e$ blocks $M_s \triangle (\mathcal{F} \cup \mathcal{P}^{\textup{elim}}_f \cup \bigcup_{f' \in F} \mathcal{P}^{\textup{elim}}_{f'})$ but not $M_s \triangle \mathcal{F}$ if and only if it blocks exactly one of $M_s \triangle (\mathcal{F} \cup \mathcal{P}^{\textup{elim}}_f )$ and $M_s \triangle (\mathcal{F} \cup \bigcup_{f' \in F}\mathcal{P}^{\textup{elim}}_{f'})$ but not $M_s \triangle \mathcal{F}$.
    So let us assume that $e$ does not block  $M_s \triangle \mathcal{F}$.
  
    Notice that if $e$ is incident to only one of $\mathcal{P}^{\textup{elim}}_f $ and $\bigcup_{f' \in F} \mathcal{P}^{\textup{elim}}_{f'}$, then our claim is immediate. 
    So suppose that $e$ connects a path $Q \in \mathcal{P}^{\textup{elim}}_f $ with a path $Q' \in \bigcup_{f' \in F} \mathcal{P}^{\textup{elim}}_{f'}$.
    Observe that it suffices to show that $e$ blocks $M_s \triangle (Q \cup Q')$ if and only if it blocks exactly one of $M_s \triangle Q$ and $M_s \triangle Q'$, because $e$ cannot be incident to any connected component of $\mathcal{F} \cup \mathcal{P}^{\textup{elim}}_f \cup \bigcup_{f' \in F} \mathcal{P}^{\textup{elim}}_{f'}$ other than $Q$ and $Q'$.
    Now the claim follows directly from Lemma~\ref{lem:additive-elim-costs} and Proposition~\ref{prop:no-double-blocking}.
  \renewcommand{\qedsymbol}{$\Diamond$}
  \end{proof}  

  It remains to consider the blocking pairs contributed by the paths $\mathcal{P}_r^{\textup{opt}}:=\{P_m^{\textup{opt}} \mid m \in \mathcal{M}_r\}$.
  \begin{claimu}  
  \label{claim:pr-opt-1}
    If $e$ is an edge that blocks $M_s \triangle (\mathcal{F} \cup \mathcal{P}_r^{\textup{opt}})$  but not $M_s \triangle \mathcal{F}$, then $e$ blocks $M^{\textup{opt}}$ but not $M^{\textup{opt}} \triangle \mathcal{P}_r^{\textup{opt}}$.
  \end{claimu}
  \begin{proof}[Proof of Claim~\ref{claim:pr-opt-1}.]
    Recall that by definition $M^{\textup{opt}}=M_s \triangle (\mathcal{F} \cup \mathcal{P}_r^{\textup{opt}} \cup \bigcup_{1 \leq i \leq \ell} \mathcal{P}^{\textup{elim}}_{f_i^{\textup{opt}}} )$ and $M^{\textup{opt}} \triangle \mathcal{P}_r^{\textup{opt}} = M_s \triangle (\mathcal{F} \cup \bigcup_{1 \leq i \leq \ell} \mathcal{P}^{\textup{elim}}_{f_i^{\textup{opt}}}) $, so if $e$ does not have an endpoint in $\bigcup_{1 \leq i \leq \ell} \mathcal{P}^{\textup{elim}}_{f_i^{\textup{opt}}}$, then the claim follows immediately.

    So suppose that $e$ is incident to a path $Q \in  \mathcal{P}^{\textup{elim}}_{f^{\textup{opt}}_i}$ for some $i \in \{1, \dots,\ell\}$.    
    As $e$ blocks $M_s \triangle (\mathcal{F} \cup \mathcal{P}_r^{\textup{opt}})$  but not $M_s \triangle \mathcal{F}$, the other endpoint of $e$ must be contained in a path $P_m^{\textup{opt}}$ for some $m \in \mathcal{M}_r$.
    Since $Q$ and $P_{m}^{\textup{opt}}$ are the only connected components of $M_s \triangle M^{\textup{opt}}$ that $e$ is adjacent to, the condition of our claim yields that $e$ blocks $M_s \triangle P_m^{\textup{opt}}$, and we need to show that $e$ blocks $M_s \triangle (Q \cup P_m^{\textup{opt}})$ but not $M_s \triangle Q$.
    To see this, first notice that by Proposition~\ref{prop:no-double-blocking}, edge~$e$ cannot block $M_s \triangle Q$, because $Q$ and $P_m^{\textup{opt}}$ are disjoint.
    Second, as $e$ blocks $M_s \triangle P_m^{\textup{opt}}$, Lemma~\ref{lem:additive-elim-costs} implies that it must block $M_s \triangle (Q \cup P_m^{\textup{opt}})$ as well, proving our claim.
  \renewcommand{\qedsymbol}{$\Diamond$}
  \end{proof}
 
  Let $c_r$ be the number of those edges that block $M^{\textup{opt}}$ but not $M^{\textup{opt}} \triangle \mathcal{P}_r^{\textup{opt}}$. 
  Then Claim~\ref{claim:pr-opt-1} implies that there are at most $c_r$ edges that block $M_s \triangle (\mathcal{F} \cup \mathcal{P}_r^{\textup{opt}})$ but not $M_s \triangle \mathcal{F}$.
  Since Steps~9.1--9.4 calculate a set $\mathcal{P}_r=\{ P_m \mid m \in \mathcal{M}_r\}$ of augmenting paths that minimizes the number of edges that block $M_s \triangle (\mathcal{F} \cup \mathcal{P}_r)$ but not $M_s \triangle \mathcal{F}$, we know that 
  there are at most $c_r$ such edges. 
  This in turn shows that there can be at most $c_r$ edges that block $M^{\textup{out}}$ but not $M^{\textup{out}} \triangle \mathcal{P}_r$, as implied by Claim~\ref{claim:pr-opt-2}.
  This proves the optimality of~$M^{\textup{out}}$.
  
  \begin{claimu}
  \label{claim:pr-opt-2}
    If $e$ is an edge that blocks $M^{\textup{out}}$ but not $M^{\textup{out}} \triangle \mathcal{P}_r$, then $e$ blocks $M_s \triangle (\mathcal{F} \cup \mathcal{P}_r)$ but not $M_s \triangle \mathcal{F}$.
  \end{claimu}
  \begin{proof}[Proof of Claim~\ref{claim:pr-opt-2}.]
    Since $M^{\textup{out}}=M_s \triangle (\mathcal{F} \cup \mathcal{P}_r \cup \mathcal{P}^{\textup{elim}} )$ and $M^{\textup{out}} \triangle \mathcal{P}_r = M_s \triangle (\mathcal{F} \cup \mathcal{P}^{\textup{elim}})$, we have that if $e$ does not have an endpoint in $\mathcal{P}^{\textup{elim}}$, then the claim follows immediately.

    So suppose that $e$ is incident to a path $Q \in  \mathcal{P}^{\textup{elim}}_{f_i}$ for some $i \in \{1, \dots,\ell\}$.    
    As~$e$ blocks $M^{\textup{out}}$  but not $M^{\textup{out}} \triangle \mathcal{P}_r$, the other endpoint of $e$ must be contained in a path $P_m$ for some $m \in \mathcal{M}_r$.
    Since $Q$ and $P_{m}$ are the only connected components of $M_s \triangle M^{\textup{out}}$ that $e$ is adjacent to, the condition of our claim yields that $e$ blocks $M_s \triangle (Q \cup P_m)$ but not $M_s \triangle Q$, and we only need to show that $e$ blocks $M_s \triangle P_m$, which is directly implied by 
    Lemma~\ref{lem:additive-elim-costs}, proving our claim.
  \renewcommand{\qedsymbol}{$\Diamond$}
  \end{proof}
  This finishes our proof of correctness for Theorem~\ref{thm:fpt-case2x}.
\end{proof}

As each augmenting path contains at least one edge that blocks~$M^{\textup{opt}}$, the number of blocking pairs admitted by~$M^{\textup{opt}}$ is at least  $(|\mathcal{W}^{\star}_0|+|\mathcal{M}^{\star}_0|)/2$.
Thus, we get Corollary~\ref{thm:fpt-case2xCor}.
\begin{corollary}
\label{thm:fpt-case2xCor}
  The special case of {\sc SMC} where each woman finds at most two men acceptable (i.e., $\Delta_{\mathcal{W}} \leq 2$) is fixed-parameter tractable for parameter $b$.
\end{corollary}

It remains to deal with the third among the parameterizations we have to consider. So let us turn our attention to the complexity of the {\sc SMC} with $\Delta_{\mathcal{W}} \leq 2$ where we take $|\mathcal{W}^{\star}|+\Delta_{\mathcal{M}}$ as the parameter. 

\begin{theorem}
\label{thm:fpt-case2x-extended}
  There is a fixed-parameter algorithm for the special case of {\sc SMC} where each woman finds at most two men acceptable (i.e., $\Delta_{\mathcal{W}} \leq 2$), and the parameter is~$|\mathcal{W}^{\star}|+\Delta_{\mathcal{M}}$, the number of distinguished women plus the maximum length of men's preference lists. 
\end{theorem}

It turns out that Theorem~\ref{thm:fpt-case2x-extended} can be proved by a modified version of the algorithm described above to prove Theorem~\ref{thm:fpt-case2x}.
In fact, we start with applying Steps 1 to 6.
However, we can no longer apply Step 7, because now $|\mathcal{M}^\star|$ can be unbounded, and thus guessing the elimination sets would not yield a fixed-parameter tractable algorithm.
To circumvent this problem, we rely on an observation stated by Lemma~\ref{lem:loose-edges} which shows that ``almost all'' volatile edges that are blocking in $M_s \triangle \mathcal{F}$ are relevant edges, and thus have to be eliminated.
This allows us to efficiently guess the set of relevant volatile edges, leading us to a fixed-parameter tractable algorithm.
Hence, instead of guessing only the elimination sets (that is, the man endpoints of the elimination paths) we can now guess the exact set of edges that the optimal solution eliminates, so in fact the modified algorithm can be viewed as a more simple approach. 


Let $E^{\textup{vol}}$ denote the set of all volatile blocking edges in $M_s \triangle \mathcal{F}$, and let $\{m,w\} \in E^{\textup{vol}}$ for some man~$m$. 
We say that $\{m,w\}$ is \emph{loose}, if 
\begin{itemize}
  \item both $M_s$ and $M^{\textup{opt}}$ assign a partner to $m$,
  \item $m$ prefers $M_s(m)$ to $w$, and $w$ to $M^{\textup{opt}}(m)$, and
  \item $M^{\textup{opt}}(m)$ is not distinguished.
\end{itemize}

\begin{lemma} 
\label{lem:loose-edges}
  All loose edges in $E^{\textup{vol}}$ are relevant, and at most $(\Delta_{\mathcal{M}}-1) |\mathcal{W}^\star|$ edges in $E^{\textup{vol}}$ are \emph{not} loose.
\end{lemma}
\begin{proof}
  First suppose that $\{m,w\}$ for some man $m$ is loose.
  Since $\{m,w\}$ is volatile, we have that $m$ is the less preferred man acceptable to $w$, and since $\{m,w\}$ is blocking in  $M_s \triangle \mathcal{F}$, we know that $w$ is unmatched in $M_s \triangle \mathcal{F}$.

  For the sake of contradiction, let us suppose that $\{m,w\}$ is not relevant, meaning that it is still blocking in $M^{\textup{opt}}$.
  In this case, $w$ must be unmatched in~$M^{\textup{opt}}$ too.
  We define a matching $M'$ obtained from~$M^{\textup{opt}}$ by modifying only the partner of $m$ to be $M'(m)=w$ and, 
  i.e., $M'=(M^{\textup{opt}} \setminus \{\{m,M^{\textup{opt}}(m)\}\}) \cup \{\{w,m\}\}$. 
  Observe that $M'$ is indeed a matching, because $w$ was unmatched in~$M^{\textup{opt}}$.
  It is also feasible, because only $M^{\textup{opt}}(m)$ becomes unmatched by this modification, and $M^{\textup{opt}}(m)$ is not distinguished (by the definition of a loose edge).
  Furthermore, while $\{m,w\}$ is not blocking in $M'$, the only edge that might be blocking in $M'$ but not in $M^{\textup{opt}}$ is the edge adjacent to $M^{\textup{opt}}(m)$ but not to~$m$, if existent (because both $m$ and $w$ are better off in $M'$ than in~$M^{\textup{opt}}$).
  Hence, $M'$ is also an optimal matching.
  It should also be clear that it has the same number of common edges with $M_s$ as $M^{\textup{opt}}$.
  Note also that each man has the same partner in $M'$ as in $M^{\textup{opt}}$ except for $m$ who is better off in~$M'$ than in~$M^{\textup{opt}}$.
  This contradicts our choice of $M^{\textup{opt}}$, proving that $\{m,w\}$ must indeed be relevant. 

  To prove the second part of the lemma, let us consider any edge $\{m,w\} \in E^{\textup{vol}}$ with $m \in \mathcal{M}$. 
  Again,~$w$ must be unmatched in the matching $M_s \triangle \mathcal{F}$, since otherwise $\{m,w\}$ would not block it (because $m$ is the less preferred man in $w$'s preference list).
  By contrast, $\{m,w\}$ cannot block $M_s$, which can be explained in three different ways, giving rise to the following three cases: 
  \begin{itemize}
    \item Case A: $\{m,w\} \in M_s$,
    \item Case B: $\{m,w\} \notin M_s$, but $M_s(w)$ exists and is preferred by $w$ to $m$, or 
    \item Case C: $\{m,w\} \notin M_s$, but $M_s(m)$ exists and is preferred by $m$ to $w$.
  \end{itemize}

  In Case~A, by $\{m,w\} \in M_s$, the edge $\{m,w\}$ must be contained in an augmenting path of $M_s \triangle \mathcal{F}$, with $w$ being an endpoint (since $w$ is not matched in $M_s \triangle \mathcal{F}$).
  In Case~B, $M_s$ matches $w$ to its first choice, say $x$, so the edge $\{w,x\}$ must be contained in an augmenting path of $M_s \triangle \mathcal{F}$, with $w$ being an endpoint. 
  Thus, in both cases we know that $w$ is an endpoint of a feminine augmenting path of $\mathcal{F}$.
  However, Cases~A and B exclude each other, because in Case~A the augmenting path ending at $w$ contains the edge connecting $w$ to her second choice, while in Case B it connects $w$ to her first choice. 
  Since there are at most $|\mathcal{W}^\star|$ feminine augmenting paths, it follows that the number of such edges $\{m,w\} \in E^{\textup{vol}}$ where the conditions of Cases~A or~B hold is at most $|\mathcal{W}^\star|$.

  In Case~C, let us first observe that $m$ gets matched in $M_s \triangle \mathcal{F}$.
  Indeed, if this were not the case, then $m$ would be the man endpoint of a feminine augmenting path, with the last edge contained in~$M_s$; however, this is not possible (simply because feminine paths start with women and with edges not in~$M_s$).
  Because $\mathcal{F}$ is a collection of connected components of $M_s \triangle M^{\textup{opt}}$, with~$m$ contained in $\mathcal{F}$, it is clear that the partner of $m$ in $M_s \triangle \mathcal{F}$ is the same as his partner in $M^{\textup{opt}}$. 
  Thus, $M^{\textup{opt}}(m)$  exists, and since $\{m,w\}$ blocks $M_s \triangle \mathcal{F}$, we know that $m$ prefers $w$ to $M^{\textup{opt}}(m)$. 
  Therefore, either $\{m,w\}$ is a loose edge, or $M^{\textup{opt}}(m)$ is distinguished; in the latter case, we say that $\{m,w\}$ \emph{belongs} to this distinguished woman. 

  Now, observe that all edges belonging to some woman $z \in \mathcal{W}^\star$ are adjacent to $M^{\textup{opt}}(z)$, and $M^{\textup{opt}}(z)$ is adjacent to at least two edges \emph{not} belonging to $z$: namely, the edges $\{M^{\textup{opt}}(z),z\}$ and $\{ M^{\textup{opt}}(z),M_s(M^{\textup{opt}}(z))\}$. 
  Hence, at most $\Delta_{\mathcal{M}}-2$ edges may belong to $z$, implying that there are at most $(\Delta_{\mathcal{M}}-2)|\mathcal{W}^\star|$ edges in total belonging to distinguished women. 
  Taking into account all three cases, we get that there are at most $(\Delta_{\mathcal{M}}-1)|\mathcal{W}^\star|$ edges in~$E^{\textup{vol}}$ that are not loose.
\end{proof}
 
Let us now describe our algorithm proving Theorem \ref{thm:fpt-case2x-extended} in detail.
This algorithm starts with Steps 1 to 6 (as given for Theorem~\ref{thm:fpt-case2x}), and then applies Steps $7^\star$ and $8^\star$ below.
The last step of our algorithm is Step 9 from the algorithm for Theorem~\ref{thm:fpt-case2x}.

\smallskip
{\noindent \bf Step $7^\star$: Guessing the relevant edges.} To determine which edges among those volatile edges that block $M_s \triangle \mathcal{F}$ are relevant, we first compute the set~$E^{\textup{loose}}$ of all loose edges in $E^{\textup{vol}}$.
Notice that when checking some edge $\{m,w\} \in E^{\textup{vol}}$, we already know $M^{\textup{opt}}(m)$: since $m$ lies on a path or cycle of $\mathcal{F}$, it gets matched in $M_s \triangle \mathcal{F}$ to the same woman as in $M^{\textup{opt}}$.
Therefore, checking whether $\{m,w\}$ is loose is straightforward. 

After determining all loose edges (which are all relevant by Lemma~\ref{lem:loose-edges}), we guess the remaining set of relevant edges in $E^{\textup{vol}} \setminus E^{\textup{loose}}$.
By Lemma~\ref{lem:loose-edges}, this yields at most $2^{(\Delta_{\mathcal{M}}-1)|\mathcal{W}^\star|}$ possibilities, and in the branch where our guess is correct, we obtain the set $E^{\textup{rel}}$ of all relevant volatile edges blocking in $M_s \triangle \mathcal{F}$.

\smallskip
{\noindent \bf Step $8^\star$: Computing cheapest elimination path.} This step is a simplification of Step 8 of the algorithm for Theorem~\ref{thm:fpt-case2x}.
By Lemma~\ref{lem:elimination}, for each $f \in E^{\textup{rel}}$ the elimination paths in $\mathcal{P}^{\textup{elim}}_f$ (determined in Step~6) are all contained in $M_s \triangle M^{\textup{opt}}$; let $\mathcal{P}^{\textup{elim}}=\cup_{f \in  E^{\textup{rel}}} \mathcal{P}^{\textup{elim}}_f$.

\smallskip
{\noindent \bf Step $9^\star$: Computing remaining masculine paths.} 
Let $\mathcal{M}_r$ denote the set of distinguished men in~$\mathcal{M}^\star_0$ not covered by any of the 
elimination paths in~$\mathcal{P}^{\textup{elim}}$. 
For each $m \in \mathcal{M}_r$ we compute an augmenting path $P_m$ disjoint from~$\mathcal{F}$ such that the number of edges that block $M_s \triangle (\mathcal{F} \cup \mathcal{P}_r)$ but not $M_s \triangle \mathcal{F}$ is minimized, where $\mathcal{P}_r$ is the union of all paths $P_m$, $m \in M_r$. 
We accomplish this with the exact same method as in Step 9 of  the algorithm for Theorem~\ref{thm:fpt-case2x}.
Finally, we output the matching $M^{\textup{out}}=M_s \triangle (\mathcal{F} \cup \mathcal{P}^{\textup{elim}} \cup \mathcal{P}_r)$.

Note that the number of guesses made in Steps 1 to 6 are bounded by a function of $|\mathcal{W}^\star_0|$, and the guesses made in Step $7^\star$ result in at most $2^{(\Delta_{\mathcal{M}}-1)|\mathcal{W}^\star|}$ possibilities. 
Since all computations in a branch can be performed in time polynomial in the size~$|I|$ of the instance, we obtain a fixed-parameter algorithm with parameter $|\mathcal{W}^\star|+\Delta_{\mathcal{M}}$. 
It remains to prove its correctness.

\begin{proof}[Proof of Theorem~\ref{thm:fpt-case2x-extended}.]
  The proof is a straightforward adaptation of our proof for Theorem~\ref{thm:fpt-case2x}. The only difference is that instead of finding \emph{some} elimination paths that eliminate the necessary number of relevant volatile edges blocking in $M_s \triangle \mathcal{F}$, our algorithm now directly finds all of the relevant edges among those volatile edges that block $M_s \triangle \mathcal{F}$, together with the corresponding elimination paths; this is a consequence of Lemma~\ref{lem:loose-edges} and the definition of Step~$7^\star$.
  As Step~$8^\star$ can also be viewed as a simplification of Step~8 as defined for the algorithm for Theorem~\ref{thm:fpt-case2x}, it can be verified that the exact same arguments that prove the correctness of Step~8 in our algorithm for Theorem~\ref{thm:fpt-case2x} also imply the correctness of Step~$8^\star$ for our algorithm for Theorem~\ref{thm:fpt-case2x-extended} in a straightforward manner.  
\end{proof}

\vspace*{-1em}
\section{Discussion}
\label{sec:discussion}
\vspace*{-0.5em}
We provided a systematic study of the computational complexity of \textsc{Stable Marriage with Covering Constraints}.
Our main result is a complete computational complexity trichotomy into polynomial-time solvable cases, $\mathsf{NP}$-hard and fixed-parameter tractable cases, and $\mathsf{NP}$-hard and $\mathsf{W}[1]$-hard cases, for all possible combinations of five natural parameters:
\vspace*{-0.5em}
\begin{itemize}
  \item $|\mathcal M^\star|$: the number of distinguished men,
  \item $|\mathcal W^\star|$: the number of distinguished women,
  \item $\Delta_{\mathcal M}$: the maximum length of preference lists for men, 
  \item $\Delta_{\mathcal W}$: the maximum length of preference lists for women, and
  \item $b$: the number of blocking pairs allowed.
\end{itemize}
As a special case, we solved a problem by Hamada et al.~\cite{HamadaEtAl2014}. 

Figure~\ref{fig-decision-diagram} provides a decision diagram showing that our results indeed fully determine the computational complexity of SMC with respect to the set $S=\{b,|\mathcal W^\star|,|\mathcal M^\star|,\Delta_{\mathcal M},\Delta_{\mathcal W}\}$ of possible parameters. 
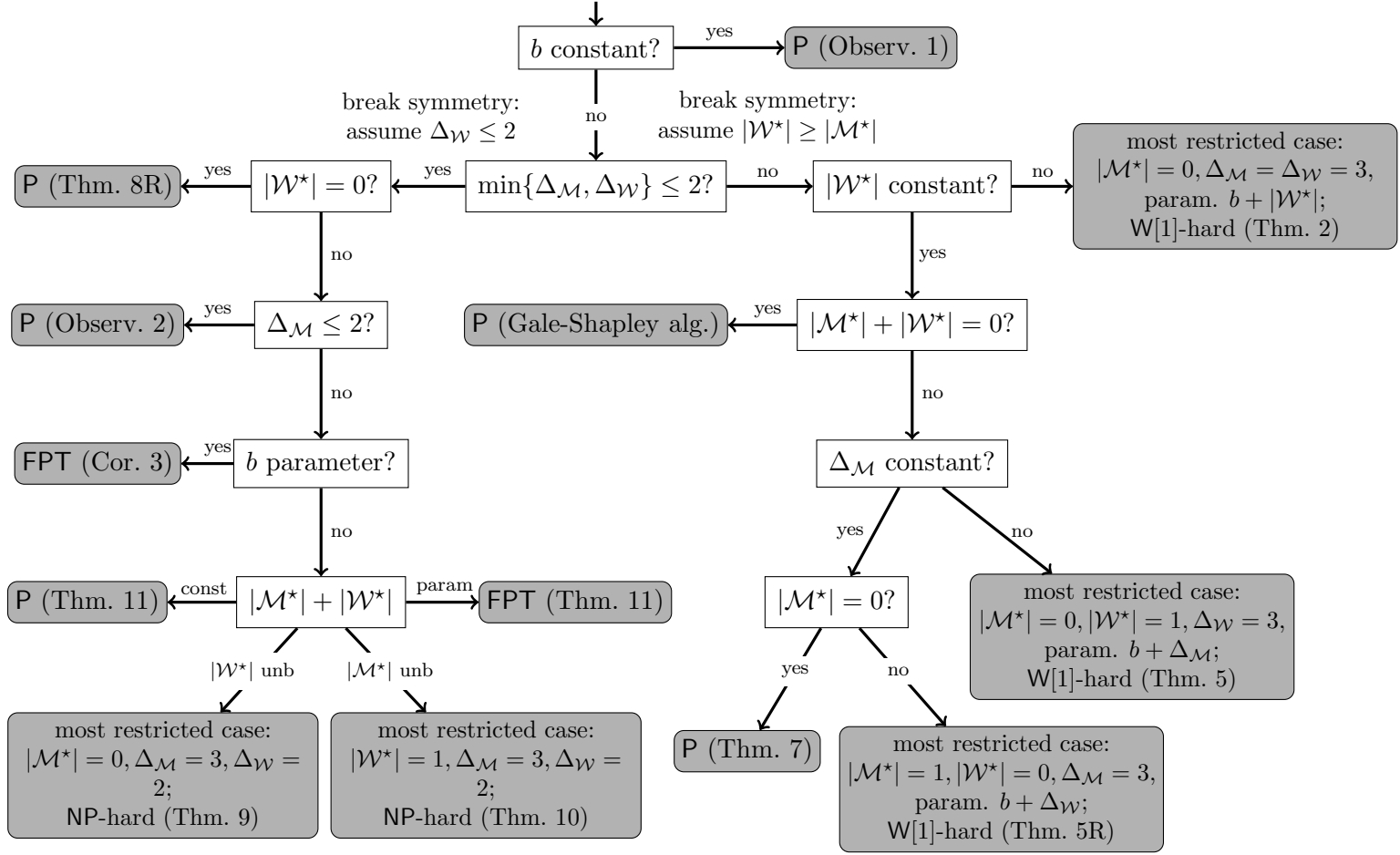
\begin{sidewaysfigure}
    \centering
    \begin{tikzpicture}[every node/.style={transform shape}]
    \node (0) at (0,0.8){};
    \node[rectangle,inner sep=5pt,draw] (1) at (0,0){$b$ constant?};
    \node[rectangle,inner sep=3pt,draw,rounded corners,fill=black!30] (2) at (4,0){$\mathsf{P}$ (Observ.~\ref{observ:bounded-b})};
    \node[rectangle,inner sep=5pt,draw] (3) at (0,-1.6){$\min\{\Delta_{\mathcal M},\Delta_{\mathcal W}\}\leq 2$?};
    \node[rectangle,inner sep=5pt,draw] (4) at (-4,-1.6){$|\mathcal W^\star| = 0$?};
    \node[rectangle,inner sep=5pt,draw] (5) at (-4,-2.8){$\Delta_{\mathcal M}\leq 2$?};
    \node[rectangle,inner sep=5pt,draw] (6) at (-4,-4){$b$ parameter?};
    \node[rectangle,inner sep=5pt,draw] (24) at (-4,-5.2){$|\mathcal W^\star|$ unbounded?};    
    \node[rectangle,inner sep=5pt,draw] (7) at (-4,-6.4){$|\mathcal W^\star|+|\mathcal M^\star|$};
    \node[rectangle,inner sep=5pt,draw] (26) at (-4,-7.6){$|\mathcal W^\star|+\Delta_{\mathcal{M}}$};
    \node[rectangle,inner sep=3pt,draw,rounded corners,fill=black!30] (8) at (-7.25,-1.6){$\mathsf{P}$ (Thm.~\ref{thm:restricted-matching-poly}R)};
    \node[rectangle,inner sep=3pt,draw,rounded corners,fill=black!30] (9) at (-7.25,-2.8){$\mathsf{P}$ (Observ.~\ref{observ:maxlist2})};
    \node[rectangle,inner sep=3pt,draw,rounded corners,fill=black!30] (10) at (-7.25,-4){$\mathsf{FPT}$ (Cor.~\ref{thm:fpt-case2xCor})};
    \node[rectangle,inner sep=3pt,draw,rounded corners,fill=black!30] (12) at (-4,-9){\small{\parbox{4.3cm}{\centering most restricted case:\\$|\mathcal W^\star|=1,\Delta_{\mathcal W}=2$;\\$\mathsf{NP}$-hard (Thm.~\ref{thm:smc2-npc-1woman})}}};
    \node[rectangle,inner sep=5pt,draw] (13) at (4,-1.6){$|\mathcal W^\star|$ constant?};
    \node[rectangle,inner sep=3pt,draw,rounded corners,fill=black!30] (14) at (8.0,-1.6){\small{\parbox{3.5cm}{\centering most restricted case:\\$|\mathcal M^\star| = 0,\Delta_{\mathcal M} = \Delta_{\mathcal W} = 3$,\\ param. $b+|\mathcal W^\star|$;\\ $\mathsf{W}[1]$-hard (Thm.~\ref{thm:smc-minblock-mainhardness})}}};
    \node[rectangle,inner sep=5pt,draw] (15) at (4,-3.0){$|\mathcal M^\star|+|\mathcal W^\star| = 0$?};
    \node[rectangle,inner sep=3pt,draw,rounded corners,fill=black!30] (16) at (8.0,-3.0){$\mathsf{P}$~(Gale-Shapley alg.)};
    \node[rectangle,inner sep=5pt,draw] (17) at (4,-4.4){$\Delta_{\mathcal M}$ constant?};
    \node[rectangle,inner sep=5pt,draw] (18) at (4,-5.6){$|\mathcal M^\star| = 0$?};
    \node[rectangle,inner sep=3pt,draw,rounded corners,fill=black!30] (19) at (2.8,-7.1){$\mathsf{P}$ (Thm.~\ref{thm:smc-allconstant})};
    \node[rectangle,inner sep=3pt,draw,rounded corners,fill=black!30] (20) at (6.6,-7.6){\small{\parbox{4.1cm}{\centering most restricted case:\\$|\mathcal M^\star| = 1, |\mathcal W^\star| = 0,\Delta_{\mathcal M} = 3$,\\ param. $b+\Delta_{\mathcal W}$;\\ $\mathsf{W}[1]$-hard (Thm.~\ref{thm:1woman-strict}R)}}};
    \node[rectangle,inner sep=3pt,draw,rounded corners,fill=black!30] (21) at (8.0,-4.4){\small{\parbox{4.1cm}{\centering most restricted case:\\$|\mathcal M^\star|=0,|\mathcal W^\star|=1,\Delta_{\mathcal W}=3$,\\ param. $b+\Delta_{\mathcal M}$;\\ $\mathsf{W}[1]$-hard (Thm.~\ref{thm:1woman-strict})}}};
    \node[rectangle,inner sep=3pt,draw,rounded corners,fill=black!30] (25) at (0.3,-5.2){\small{\parbox{4.1cm}{\centering most restricted case:\\$|\mathcal M^\star|=0,\Delta_{\mathcal W}=2,\Delta_{\mathcal M}=3$,\\$\mathsf{NP}$-hard (Thm.~\ref{thm-case23})}}};
    \node[rectangle,inner sep=3pt,draw,rounded corners,fill=black!30] (22) at (-7.2,-6.4){$\mathsf{P}$ (Thm.~\ref{thm:fpt-case2x})};
    \node[rectangle,inner sep=3pt,draw,rounded corners,fill=black!30] (23) at (-0.3,-6.4){$\mathsf{FPT}$ (Thm.~\ref{thm:fpt-case2x})};
    \node[rectangle,inner sep=3pt,draw,rounded corners,fill=black!30] (27) at (-7.2,-7.6){$\mathsf{P}$ (Thm.~\ref{thm:fpt-case2x-extended})};
    \node[rectangle,inner sep=3pt,draw,rounded corners,fill=black!30] (28) at (-0.3,-7.6){$\mathsf{FPT}$ (Thm.~\ref{thm:fpt-case2x-extended})};    
    \draw[->,very thick](0)--(1);
    \draw[->,very thick](1)--(2);
    \draw[->,very thick](1)--(3);
    \draw[->,very thick](3)--(4);
    \draw[->,very thick](4)--(8);
    \draw[->,very thick](4)--(5);
    \draw[->,very thick](5)--(9);
    \draw[->,very thick](5)--(6);
    \draw[->,very thick](6)--(10);
    \draw[->,very thick](6)--(24);
    \draw[->,very thick](18)--(19);
    \draw[->,very thick](18)--(20);
    \draw[->,very thick](3)--(13);
    \draw[->,very thick](13)--(15);
    \draw[->,very thick](15)--(17);
    \draw[->,very thick](17)--(18);
    \draw[->,very thick](15)--(16);
    \draw[->,very thick](17)--(21);
    \draw[->,very thick](7)--(22);
    \draw[->,very thick](7)--(23);
    \draw[->,very thick](13)--(14);
    \draw[->,very thick](24)--(25);
    \draw[->,very thick](24)--(7);    
    \draw[->,very thick](7)--(26);    
    \draw[->,very thick](26)--(27);
    \draw[->,very thick](26)--(28);    
    \draw[->,very thick](26)--(12);
    \node[rectangle,minimum width=2em] (a) at (1.8,0.2){\scriptsize{yes}};%
    \node[rectangle,minimum width=2em,fill=white] (b) at (0,-0.8){\scriptsize{no}};%
    \node[rectangle,minimum width=2em] (c) at (2.1,-1.4){\scriptsize{no}}; %
    \node[rectangle,minimum width=2em] (d) at (5.6,-1.4){\scriptsize{no}}; %
    \node[rectangle,minimum width=2em] (e) at (5.8,-2.8){\scriptsize{yes}}; %
    \node[rectangle,minimum width=2em] (f) at (4.3,-2.3){\scriptsize{yes}}; %
    \node[rectangle,minimum width=2em] (f) at (4.3,-3.7){\scriptsize{no}}; %
    \node[rectangle,minimum width=2em] (g) at (4.3,-5){\scriptsize{yes}};
    \node[rectangle,minimum width=2em] (f) at (5.4,-4.2){\scriptsize{no}}; %
    \node[rectangle,minimum width=2em] (h) at (5.3,-6.3){\scriptsize{no}};
    \node[rectangle,minimum width=2em] (i) at (3,-6.35){\scriptsize{yes}};
    \node[rectangle,minimum width=2em] (j) at (-2.3,-1.4){\scriptsize{yes}}; %
    \node[rectangle,minimum width=2em] (k) at (-5.5,-1.4){\scriptsize{yes}}; %
    \node[rectangle,minimum width=2em] (l) at (-5.5,-2.6){\scriptsize{yes}}; %
    \node[rectangle,minimum width=2em] (m) at (-2.3,-5){\scriptsize{yes}}; 
    \node[rectangle,minimum width=2em] (m) at (-5.5,-3.8){\scriptsize{yes}}; %
    \node[rectangle,minimum width=2em] (n) at (-3.7,-2.2){\scriptsize{no}}; %
    \node[rectangle,minimum width=2em] (o) at (-3.7,-3.4){\scriptsize{no}};
    \node[rectangle,minimum width=2em] (o) at (-3.7,-4.6){\scriptsize{no}};
    \node[rectangle,minimum width=2em] (p) at (-3.7,-5.8){\scriptsize{no}};
    \node[rectangle,minimum width=2em] (p) at (-3.3,-7){\scriptsize{$|\mathcal M^\star|$ unb}};
    \node[rectangle,minimum width=2em] (p) at (-3.3,-8.2){\scriptsize{$\Delta_{\mathcal M}$ unb}};
    \node[rectangle,minimum width=2em] (t) at (-5.7,-6.2){\scriptsize{const}};
    \node[rectangle,minimum width=2em] (u) at (-2.2,-6.2){\scriptsize{param}};
        \node[rectangle,minimum width=2em] (tt) at (-5.7,-7.4){\scriptsize{const}};
    \node[rectangle,minimum width=2em] (uu) at (-2.2,-7.4){\scriptsize{param}};
    \node (A) at (-2.4,-0.8){\small{\parbox{0.25\textwidth}{\centering break symmetry:\\ assume $\Delta_{\mathcal W}\leq 2$}}};
    \node (B) at (2.5,-0.8){\small{\parbox{0.25\textwidth}{\centering break symmetry:\\ assume $|\mathcal W^\star|\geq |\mathcal M^\star|$}}};
    \end{tikzpicture}
    \caption{Decision diagram for determining the complexity of SMC.
        We remark that if a certain restriction of SMC where one of the values $v \in S$ is a constant $k$ proves to be $\mathsf{NP}$-hard or $\mathsf{W[1]}$-hard with some parameterization, then it is easy to see that the same hardness result also holds for the case where $v \geq k$ (and all other assumptions are the same). 
        We refer to the ``reflection'' of a result by adding the postfix `R' to its name (so Theorem~$x$R denotes the reflection of Theorem~$x$); here by reflection we mean the statement obtained by switching the roles of men and women. 
    }
    \label{fig-decision-diagram}
\end{sidewaysfigure}

Going through this decision diagram should convince the reader that any parameterized restriction of SMC with respect to the set $S$
is classified as either polynomial-time solvable ($\mathsf{P}$) or $\mathsf{NP}$-hard, and in the latter case, either fixed-parameter tractable ($\mathsf{FPT}$), or $\mathsf{W[1]}$-hard with the given parameterization (if any). 
In particular, when we provide parameterized results, this means that the parameterized restriction of SMC in question is $\mathsf{NP}$-hard without parameterization. 


Given the strong polynomial-time inapproximability bounds, as well as the parameterized intractability results of this paper, we pose as an open question whether \emph{fixed-parameter approximation algorithms} can beat either of these obstacles for solving SMC. 

Another challenge for future research is to investigate possible adaptations of the proposed algorithms to the Hospitals/Residents model
(note that, naturally, all our hardness results for SMC-1 apply to the HRLQ problem), or to a setting where ties are allowed in the preference lists.


\medskip
\noindent
\textbf{Acknowledgements.} We thank two anonymous reviewers for their remarks which greatly helped to improve a previous version of this paper; we are grateful for their patience and their rigorous reading of earlier drafts.

\bibliographystyle{abbrvnat}      
\bibliography{stablematchings}  

\end{document}